%% file: elsarticle-template-num.tex
\documentclass[preprint,12pt, 3p, sort&compress]{elsarticle}
\usepackage{fullpage}
\usepackage[T1]{fontenc}
\usepackage{amssymb,amsmath,amsthm,amsfonts,amsbsy,mathrsfs}
\usepackage{calligra}
\usepackage[caption=false]{subfig}
\usepackage{graphicx,color}
\usepackage{multirow}
\usepackage{booktabs}

\usepackage[colorlinks   = true, urlcolor     = cyan, linkcolor    = blue, citecolor   = red]{hyperref}
\usepackage{bm}
\usepackage{pict2e}
\usepackage{epsfig}
\usepackage{epstopdf}
\usepackage{comment}
\usepackage{url}
\usepackage[ruled,vlined]{algorithm2e}
\usepackage{tikz}
\usepackage{pgfplots}
\usepackage{pgfplotstable}
\usepgfplotslibrary{external}
\usetikzlibrary{shapes.geometric, arrows}
\pgfplotsset{
    compat=newest,
    table/header=false,
    title style={font=\small},
    tick label style={font=\scriptsize},
    label style={font=\scriptsize},
    legend style={font=\scriptsize},
    legend cell align=left
}

\usepackage{tabularx,multirow,booktabs}
\newcolumntype{L}{>{\raggedright\arraybackslash}X}
\newcolumntype{C}{>{\centering\arraybackslash}X}
\newcolumntype{R}{>{\raggedleft\arraybackslash}X}

\journal{arXiv}

\input{preamble}

\begin{document}

%%%%%%%%%%%%%%%%%%%%%%%%%%%%%%%%%%%%%%%%%%%%%%%%%%%%%%%%%%%%%%%%%%%%%%%%%%%%%%%%%%%%%%%%%%%%%%%%%%%%%%%%%%%%%%%%%%%%%%%%%%%%%%%%%%%%
%%%%%%%%%%%%%%%%%%%%%%%%%%%%%%%%%%%%%%%%%%%%%%%%%%%%%%%%%%%%%%%%%%%%%%%%%%%%%%%%%%%%%%%%%%%%%%%%%%%%%%%%%%%%%%%%%%%%%%%%%%%%%%%%%%%%

\begin{frontmatter}
%\title{Ab initio framework for simulating systems with helical symmetry: formulation, implementation and applications to torsional deformations in nanostructures}
%\title{Ab initio framework for simulating systems with helical symmetry: formulation of governing equations, numerical implementation and applications to torsional deformations in nanostructures}
\title{{Ab initio} framework for systems with helical symmetry: theory, numerical implementation and applications to torsional deformations in nanostructures}

%\author[asb_ucla,asb_lbl]{Amartya S.\ Banerjee}\corref{cor1}
\author[asb_ucla]{Amartya S.\ Banerjee}\corref{cor1}
\ead{asbanerjee@ucla.edu}
\cortext[cor1]{Corresponding author}
\address[asb_ucla]{Department of Materials Science and Engineering, University of California, Los Angeles, CA 90095, U.S.A}
\begin{abstract}
{We formulate and implement Helical Density Functional Theory (Helical DFT) --- a self-consistent first principles simulation method for nanostructures with helical symmetries. Such materials are well represented in all of nanotechnology, chemistry and biology, and prominent examples include  nanotubes, nanosprings, nanowires, miscellaneous chiral structures and important proteins.  The overwhelming preponderance of such \textit{helical structures} in all of science and engineering and the likelihood of these systems being associated with exotic materials properties, provides the motivation to develop systematic and predictive tools for their study.}

{Following this line of thought, we develop a mathematical and computational framework in this contribution, that allows helical structures to be studied \textit{ab initio}, using Kohn-Sham theory. We first show that the electronic states in helical structures can be characterized by means of special solutions to the single electron problem called \textit{helical Bloch waves}. We rigorously demonstrate the existence and completeness of such solutions, and then describe how they can be used to reduce the Kohn-Sham Density Functional Theory (KS-DFT) equations for helical structures to a suitable fundamental domain. Next, we develop a symmetry-adapted finite-difference strategy in helical coordinates to discretize the governing equations, and obtain a working realization of our proposed approach. We verify the accuracy and convergence properties of our numerical implementation through examples. Finally, we employ Helical DFT to study the properties of zigzag and chiral single wall black phosphorus (i.e., phosphorene) nanotubes. Specifically, we use our simulations to evaluate the torsional stiffness of a zigzag nanotube \textit{ab initio}. Additionally, we observe an insulator-to-metal-like transition in the electronic properties of this nanotube as it is subjected to twisting. We also find that a similar transition can be effected in chiral phosphorene nanotubes by means of axial strains. The strong dependence of the band gap of these materials on various modes of strain suggests their possible use as nanomaterials with tunable electronic and transport properties. Notably, self-consistent ab initio simulations of this nature are unprecedented and well outside the scope of any other systematic first principles method in existence. We end with a discussion on various future avenues and applications.}
\end{abstract}
 
\begin{keyword}
Kohn-Sham density functional theory, helical symmetry, phosphorene, nanotube,
torsional deformations. %objective structures, tunable band gap, insulator metal transition. (\textbf{Need to choose top five})
\end{keyword}

\end{frontmatter}
\interfootnotelinepenalty=10000
\section{Introduction}
\label{sec:introduction}
The discovery and characterization of novel nanomaterials and nanostructures constitutes one of the principal areas of scientific research today \citep{bhushan2017springer, cao2004nanostructures}. Such materials and structures hold the promise of unlocking remarkable and unprecedented material properties that are otherwise unavailable in the bulk phase (i.e., crystalline materials). In recent years, the discovery of novel nanostructures has garnered much attention and acclaim \citep{geim_nobel_lecture, smalley_nobel_lecture}, and the unusual properties of these new materials have led to ground breaking applications in almost every branch of science and engineering \citep{nanotech_applications_1,nanotech_gov_website}.

Nanostructures appear in various morphologies (including  fullerenes, nanotubes, nanoclusters  and two-dimensional materials), and are usually associated with non-periodic symmetries.\footnote{Atomistic and molecular structures with non-periodic symmetries have been termed as \textit{Objective Structures} in the mechanics literature \citep{James_OS}. First principles calculations for such structures was the topic of investigation of \citep{My_PhD_Thesis} and the current contribution continues and extends that line of work, i.e., it can be viewed as a particular flavor of \textit{Objective Density Functional Theory}.} The mathematical framework for  classifying nanostructures \citep{James_OS, Hahn,DEJ_ObjForm} shows that a vast class of these materials can be described as being \emph{helical}, i.e., their spatial atomic arrangement possesses helical symmetries. Helical structures include important technological materials such as nanotubes (of any chirality), nanoribbons, nanowires and nanosprings; miscellaneous chiral structures encountered in chemistry; and examples from biology, including tail sheaths of viruses and many common proteins \citep{James_OS, ren2014review}. Figure \ref{fig:collage} shows instances of helical structures that have been actively investigated in the literature.
\begin{figure}[ht]
\centering
\subfloat[Nanotubes of different chirality.]{\includegraphics[trim={7cm 7cm 4cm 5.5cm}, clip, width=0.49\textwidth]{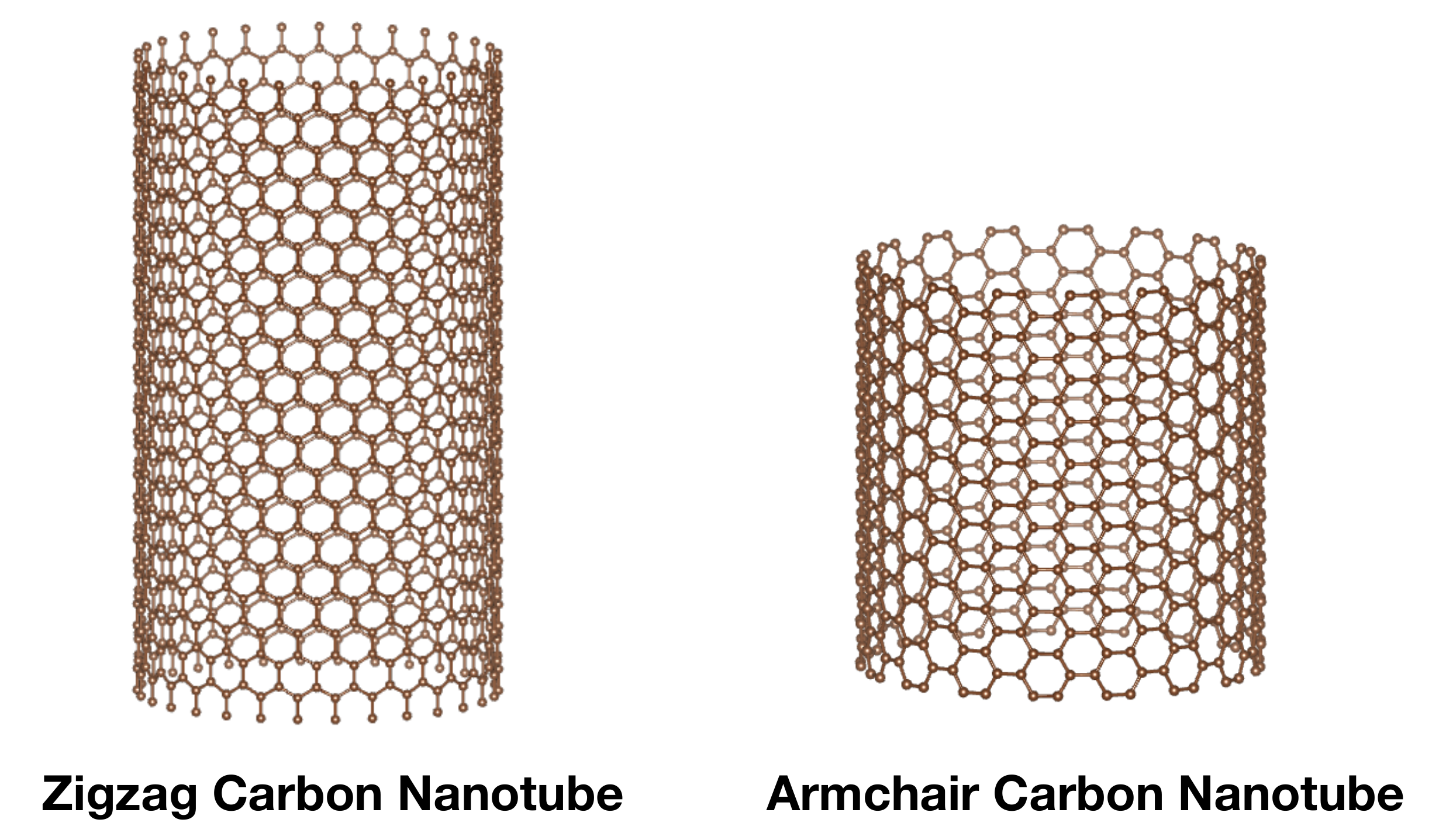}}$\;$%\label{fig:phosphorene_sheet_top_view}\quad
\subfloat[Nanoribbons of different 2D materials.]{\includegraphics[trim={7cm 3cm 5.5cm 5cm}, clip, width=0.49\textwidth]{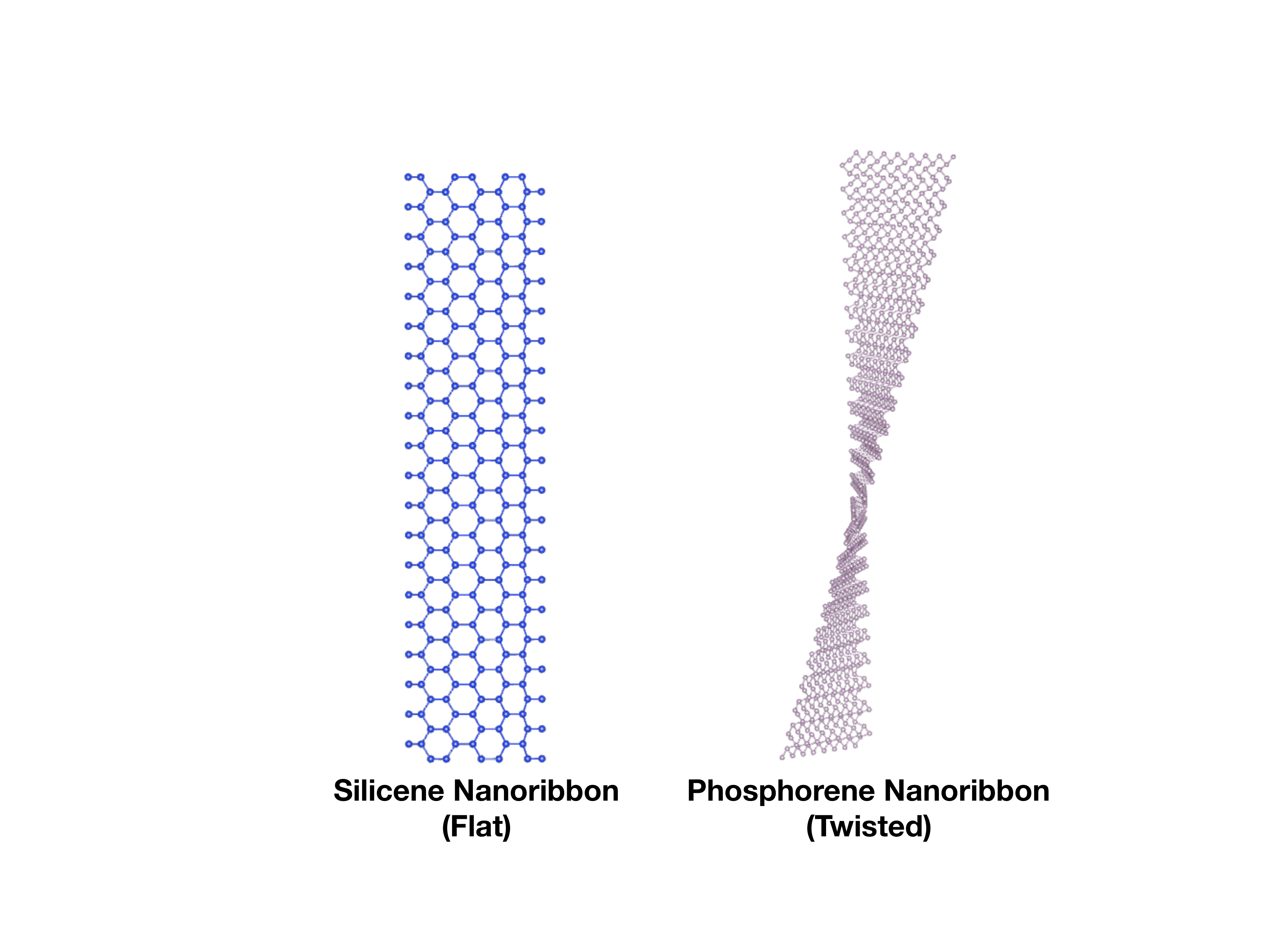}}\\
\subfloat[Molecules of biological origin.]{\includegraphics[trim={7cm 5cm 6cm 5cm}, clip, width=0.50\textwidth]{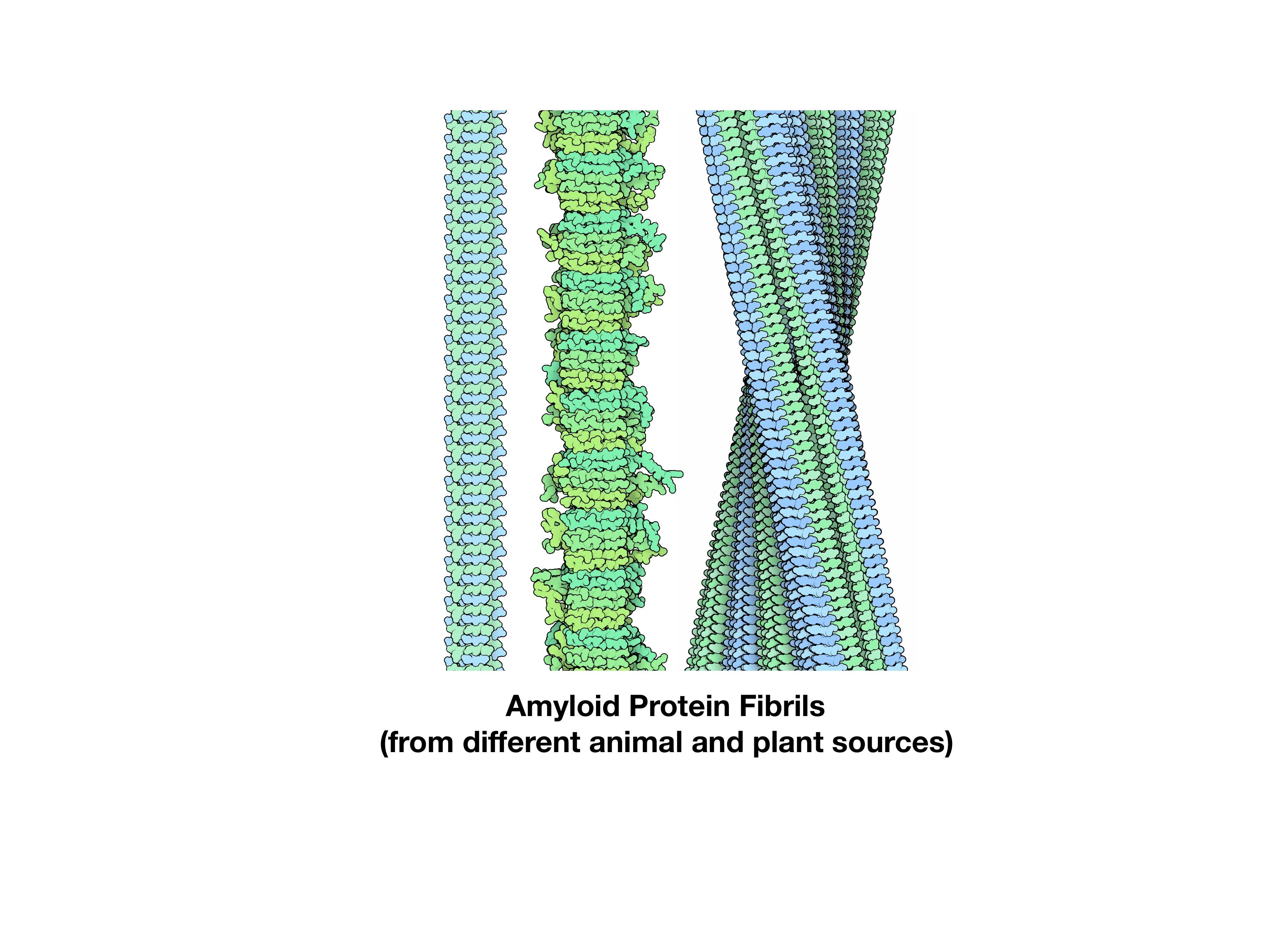}}
\caption{Examples of helical structures from nanotechnology and biology. Images of the Amyloid protein were obtained courtesy of The Protein Data Bank \citep{PDB_1, PDB_2,PDB_3}.}
\label{fig:collage}
\end{figure}

Helical structures have been conjectured to be a fertile source of novel materials with unusual and attractive properties \citep{James_OS}. This is due to the fact that atoms in such structures find themselves in locally similar environments \citep{James_OS}. Coupled with the quasi-one-dimensional nature of these systems,  as well as the presence of symmetries in  the underlying governing equations, this makes it likely that collective or correlated electronic effects (such as those leading to ferromagnetism, ferroelectricity and superconductivity) can emerge in these materials \citep{shimada2016polar}. On the other hand, helical structures are also inherently chiral and can therefore serve as natural examples of materials systems in which certain forms of symmetry breaking in the governing equations can lead to unconventional transport phenomena \citep{naaman2015spintronics, dalum2019theory, medina2015continuum, aiello2020chirality}.

Given the relative abundance of helical nanostructures in existing materials, their likelihood of being associated with hitherto undiscovered forms of matter displaying exotic materials properties, and their overall scientific and technological importance, there appears to be a pressing need for reliable and efficient computational tools for studying such systems. The broad goal of the present contribution is to take important foundational steps in addressing the above scientific issue. Specifically, we present here the mathematical formulation and numerical implementation of a novel computational method called Helical DFT, that can simulate helical structures \textit{ab initio}. We also obtain a practical working realization of this (density functional theory based) self-consistent first principles technique,  and illustrate some of its capabilities through the study of an emergent nanotube material with interesting properties.

To put our work into perspective, we remark that the use of first principles (i.e., quantum mechanical) techniques to design and study materials is a very active area of scientific endeavor today, and it forms the bulk of computational materials science research \citep{giustino2014materials, hafner2008ab, ziegler1991approximate, jain2016computational, hafner2006toward}. Among the wide array of first principles methods available, Kohn-Sham Density Functional Theory (KS-DFT) \citep {KohnSham_DFT} enjoys widespread usage since it offers a good balance between computational cost and physical accuracy as compared to other techniques \citep{LeBris_ReviewBook}. The pseudopotential plane-wave method, also called \textit{Plane-wave DFT}, is the most widely used implementation of Kohn-Sham theory \citep{Kresse_abinitio_iterative, CASTEP_1, Quantum_Espresso_1, Gonze_ABINIT_1}, and it involves expanding the unknowns into linear combinations of plane-waves.  Since plane-waves are naturally associated with periodic symmetries (they are in fact eigenfunctions of translational symmetry operators), Plane-wave DFT is ideally suited for studying bulk (i.e. periodic or crystalline) systems, and is often found to be fundamentally inadequate for studying systems with non-periodic symmetries. In particular, using a Plane-wave DFT code for studying a helical structure such as a chiral nanotube can require the use of large periodic unit cells often containing many hundreds (or even thousands) of atoms.\footnote{In contrast to plane-waves, the use of real space techniques based on finite differences \citep{Chelikowsky_Saad_1, Octopus_1, ghosh2017sparc_1, ghosh2017sparc_2} or finite elements \citep{Pask_FEM_review_1, Pask_FEM_review_2,Gavini_higher_order} allow for non-periodic boundary conditions to be imposed in a straight-forward manner. However there does not appear to be any prior work on using these techniques for self-consistent first principles calculations of helical systems.}  In contrast, a computational method which is faithful to the underlying helical symmetry of such a structure would require a small helical unit cell, containing far fewer number of atoms. Since ground state electronic structure calculations using density functional theory (DFT) scale as the cube of the number of atoms in the unit cell, while excited state calculations scale as the fourth power, the difference in simulation run times for such calculations, in these two scenarios (i.e., the correct use of helical symmetry vs. incorrect use of periodic symmetry) can be drastically different in practice.

The above considerations form our point of departure from a conventional formulation and implementation of KS-DFT, to one that is adapted for helical systems. In order to formulate the equations of KS-DFT for a helical unit cell, an appropriate version of the Bloch Theorem  \citep{Ashcroft_Mermin, Kittel} is required. We establish this result rigorously in this work, and use it to set up an electronic band theory for helical structures. Subsequently, we develop the notion of helical Bloch states, and use their properties to derive of the equations of KS-DFT, as they apply to helical systems. A key component in our mathematical treatment is the definition and use of a helical Bloch-Floquet transform to perform a block-diagonalization of the Hamiltonian in the sense of direct integrals. Our use of rigorous mathematical arguments and appropriate mathematical tools\footnote{Due to the infinite nature of helical groups, the mathematical arguments presented here are of somewhat different and more subtle nature as compared to the ones that can be employed for cyclic groups \citep{banerjee2016cyclic}. However, they can be seen as being broadly connected in the sense that they both deal with Fourier analysis of the respective symmetry groups \citep{Barut_Reps, Folland_Harmonic}.} is one of the highlights of our framework, and it allows the governing equations to be obtained systematically, and without recourse to an excessive amount of intuition.\footnote{Since a rigorous thermodynamic limit theory for the Kohn-Sham problem is unknown \citep{le2005atoms, Defranceschi_LeBris}, a derivation of the equations of the theory, as it applies to condensed matter systems often makes use of physical intuition. This process is prone to conceptual errors however, and we are aware of literature that lists certain terms of the equations incorrectly. In any case, the final form of the  equations appear to be well known in the electronic structure community at large, since DFT codes routinely make use of them for simulating the crystalline phase.} As far as we are aware, our work is the first in presenting such a derivation, and also in expressing the detailed form   of the equations of Kohn-Sham theory for helical structures. The final form of the equations are such that they are readily suited for implementation within systematically convergent electronic structure methods such as those based on finite differences \citep{Chelikowsky_Saad_1, Octopus_1, ghosh2017sparc_1, ghosh2017sparc_2}, finite elements \citep{Pask_FEM_review_1, Pask_FEM_review_2,Gavini_higher_order} or spectral basis functions \citep{My_PhD_Thesis, Banerjee2015spectral, My_Shivang_HelicES_paper}. We choose a symmetry adapted finite difference method in helical coordinates for discretizing the governing equations in this work, and set up a computational framework for numerically solving the discretized equations in a self-consistent manner. This gives us a working realization of an \textit{ab initio} computational tool  --- called Helical DFT --- that can be used to perform predictive simulations of helical systems in a systematic and efficient manner. It can therefore aid in the discovery, synthesis and characterization of helical structures. Subsequently, the remainder of this work focuses on illustrating various numerical and application oriented aspects of this novel computational tool through examples based on nanotube systems. To the best of our knowledge, Helical DFT is the first computational method for helical systems that is based on first principles, and one that also behaves systematically with respect to convergence properties. This, among other reasons, is made possible by our use of the aforementioned helical coordinate system.  To the best of our knowledge, this has not been employed in electronic structure calculations before.

While the study of helical structures has much scientific and technological merit in of itself, the development of a computational method for studying such systems also brings with it the added benefit of being able to simulate the behavior of nanomaterials under torsional deformations. As explained in \citep{James_OS, banerjee2016cyclic}, homogeneous deformation modes are commensurate with periodic symmetries (i.e., applying a homogeneous deformation to a periodic structure results in another periodic structure), while certain inhomogeneous deformation modes can be associated with non-periodic symmetries. An attempt to study such inhomogeneous deformations while using a periodic method is likely to involve various uncontrolled approximations, complications and computational inefficiencies \citep{wei2012bending, naumov2011gap}. This issue appears to have been recognized for some time in the nanomechanics and materials literature, leading to a considerable body of work centered around suggestions presented in \citep{James_OS}, whereby pure bending deformations in atomistic systems are simulated using cyclic symmetries, while helical symmetries are used to simulate torsion \citep{Dumitrica_James_OMD,cai2008torsion, mukherjee2020symmetry, Dumitrica_Bending_Graphene, CNT_Dumitrica, ma2015thermal, Pekka_Efficient_Approach, Pekka_CNT_Bending, Pekka_GNR_Bending, Pekka_Revised_Periodic}. A persistent issue with the simulations in these studies however, is that they have all been carried out using interatomic potentials or tight binding methods. Due to the well known deficiencies of these techniques in simulating real materials \citep{ismail2000ab,hauch1999dynamic, cocco2010gap, koskinen2009density, naumov2011gap}, true first principles simulation methods that behave systematically, and also take into account cyclic and/or helical symmetries have been deemed highly desirable \citep{Pekka_Revised_Periodic, James_OS, My_PhD_Thesis}. There has been recent progress on this very issue with regard to cyclic symmetries \citep{banerjee2016cyclic, ghosh2019symmetry}, and the resulting computational methods have been used to study the bending behavior of nanoribbons and sheets of two dimensional materials \textit{ab initio}. In this sense, the current contribution follows up on this line of work by making a first principles simulation framework for torsional deformations available. Consequently, through the use of this framework, we are able to extract the behavior of nanotubes of black phosphorus (i.e., phosphorene nanotubes) and study their mechanical and electronic response as they are subjected to twisting.\footnote{Exploitation of helical symmetries in ab initio calculations has also been considered in the chemistry literature in the context of Linear Combination of Atomic Orbitals (LCAO) methods  \citep{d2009single, dovesi2017crystal17, Mintmire_White1, CNT_1, CRYSTAL, CNT_4}. However, these methods differ in their perspective from the current contribution in that they concentrate on using symmetry-adapted basis functions for reducing the computational cost of the multi-center integrals and the Hamiltonian matrix elements, whereas our focus is on the formulation of symmetry-adapted cell problems (in helical coordinates), and a systematically convergent numerical treatment of these cell problems. Due to basis incompleteness and superposition errors, it is often non-trivial to systematically improve the quality of the numerical solutions obtained via LCAO  methods, in contrast to the techniques presented here. Finally, the connection of helical symmetries with torsional deformations, as well as the effect of such deformations on other material properties does not appear to have been considered in the chemistry literature.} The coupling of these responses leads to some interesting electronic transitions in this material that is likely to make it an attractive candidate for sensing, modulation and actuation applications. 

The rest of this work is organized as follows. Section \ref{sec:formulation} establishes the mathematical framework for a systematic formulation of the governing equations, and also derives the relevant expressions explicitly. Section \ref{sec:Numerical_Implementation} discusses formulation of a numerical scheme based on finite differences in helical coordinates, and Section \ref{sec:Simulation_Results_Discussion} presents simulation studies. Section \ref{sec:Conclusion} summarizes the work and suggests avenues for future research. The appendices contain additional information and discussions on mathematical tools and results that allow for this work to be self-contained.
\section{Formulation}
\label{sec:formulation}
In this section, we describe the key aspects of Helical DFT. We begin with a formal discussion of helical groups and helical structures in Section \ref{subsec:helical_groups_helical_structures}, and then discuss Kohn-Sham DFT, as it applies to such systems in  Section \ref{subsec:Kohn_Sham_Helical}. The atomic unit system with $m_{\text{e}}=1,e=1,\hbar=1,\frac{1}{4\pi\epsilon_0}=1$, is chosen for the rest of the work, unless otherwise mentioned.
\subsection{Helical symmetry groups, fundamental domains and helical structures}
\label{subsec:helical_groups_helical_structures}
A helical structure (i.e.~a structure with helical symmetries) can be defined through the action of a helical group on a set of non-degenerate points in space. This definition makes it necessary for us to make the notion of a helical group precise. Following standard practice in the literature \citep{Hahn, James_OS, Dumitrica_James_OMD, DEJ_ObjForm, feng2019phase, feng2018phase, My_PhD_Thesis}, we introduce helical groups as subgroups of the Euclidean group in three dimensions. This requires us to introduce some relevant notation and basic rules regarding operations with isometries, as we now do.
\subsubsection{Helical symmetry groups}
\label{subsubsec:helical_groups}
Let $\bfe_1,\bfe_2,\bfe_3$ denote the standard orthonormal basis\footnote{We will use the following notation in what follows:  $f(\cdot)$ will be used to denote a function when we do not wish to highlight the dependence of the function on its arguments. $\norm{\cdot}{}$ will be used to denote the norm of a function or vector and $\innprod{\cdot}{\cdot}{}$ will be used to denote the inner product. Often, we will attach a subscript to these symbols to denote the specific space in which the norm or inner product is being considered. Vectors and matrices in $\rz^2$ or $\rz^3$ will be denoted in boldface, with lower case letters reserved for vectors and uppercase letters used for matrices. We will sometimes use the $\cdot$ symbol between vectors in $\rz^2$ or $\rz^3$, to denote the inner product. {If a function has dependence on multiple arguments, we may choose to separate the arguments using `$;$' to emphasize a parametrized dependence of the function on the arguments following `$;$'.}} of $\rz^3$ and let $(x_1,x_2,x_3)$ denote the Cartesian coordinates of a generic point $\bfx \in \rz^3$. An isometry (or rigid body motion) in $\rz^3$ will be denoted using the notation $\Upsilon = (\bfR | \bfc)$, with $\bfR \in \text{SO}(3)$ denoting the rotation part of the rigid body motion, and $\bfc \in  \rz^3$ denoting the translation part. The action\footnote{As the name suggests, isometries preserve distances (and hence, also angles), i.e.,$\forall\,\bfx,\bfy \in \rz^3$, and a generic isometry ${\Upsilon}$, it holds that $\norm{{\Upsilon}(\bfx) - {\Upsilon}(\bfy)}{\rz^3} = \norm{\bfx -\bfy}{\rz^3}$.} of $\Upsilon : \rz^3 \to \rz^3$ on a point $\bfx \in \rz^3$ is written as $\Upsilon \circ \bfx = \bfR \bfx + \bfc$.  Given a collection of points $S \subset \rz^3$, we will use the notation $\Upsilon \circ S$ to denote the action of the isometry on each of the points in $S$, i.e.,
\begin{align}
\label{eq:isometry_action_set}
\Upsilon \circ S := \bigcup_{\bfx\in S}\Upsilon \circ \bfx\,.
\end{align}

There is a natural multiplicative operation associated with isometries (denoted as $\bullet$ here) that arises as a composition of their maps. Specifically, given isometries $\Upsilon_1 = (\bfR_1 | \bfc_1), \Upsilon_2 = (\bfR_2 | \bfc_2)$, we may define a third isometry $\Upsilon_3 = \Upsilon_1 \bullet \Upsilon_2$ such that $\Upsilon_3 \circ \bfx = (\Upsilon_1 \bullet \Upsilon_2)  \circ \bfx :=  \Upsilon_1 \circ( \Upsilon_2 \circ \bfx)$. It follows that $\Upsilon_3 = (\bfR_1  \bfR_2| \bfR_1 \bfc_2 + \bfc_1)$, and that in general the operation $\bullet$ is not commutative (due to non-commutativity of finite rotations about arbitrary axes). The $\bullet$ operation also allows the definition of whole number powers of $\Upsilon$, i.e., for $n = 1,2, \ldots$, we may define $\Upsilon^n := \Upsilon \bullet  \Upsilon \bullet \Upsilon \ldots (n\;\text{times})$. It is then easy to check that  $\Upsilon^n$ admits the expression $\displaystyle \Upsilon^n =\big(\bfR^n \big| (\sum_{j=0}^{n-1}\bfR^j)\,\bfc \big)$, where the notation $\bfR^0$ is used to denote the identity matrix. 

The identity isometry leaves every $\bfx \in \rz^3$ invariant and can be written as $(\bfI| \textbf{0})$, with $\bfI$ denoting the identity matrix and $\textbf{0}$ denoting the null vector in $ \rz^3$. Given the isometry $\Upsilon = (\bfR | \bfc)$, we can form the isometry $\Upsilon' = (\bfR^{-1} | -\bfR^{-1}\bfc)$, which satisfies $\Upsilon \bullet \Upsilon' = \Upsilon' \bullet \Upsilon = (\bfI| \textbf{0})$. Hence, we will denote $\Upsilon'$ as $\Upsilon^{-1}$ --- i.e.,  the inverse isometry to $\Upsilon$. The set of all isometries so defined, i.e., $\mathcal{E} = \{ \Upsilon = (\bfR | \bfc): \bfR \in \text{SO}(3), \bfc \in  \rz^3 \}$, together with the operation $\bullet$ and the inverse element defined above, form a group \citep{Miller_Symmetry}.\footnote{Since only pure rotations are included, this is the so called Euclidean group of \textit{direct isometries} in three dimensions \citep{Miller_Symmetry}. The full Euclidean group also includes improper rotations. }
 
Let $\alpha$ and $\tau$ be real numbers\footnote{Most of the discussion in this work naturally also extends to the case when $-1 < \alpha < 0$. However, we will not be considering that case here.} such that $0\leq\alpha <1$ and $\tau  > 0$, and let $\bfR_{2\pi\alpha}$ denote a rotation around axis $\bfe_3$ by angle $2\pi\alpha$. Then, the rigid body motion $\Upsilon_\mathsf{h} = (\bfR_{2\pi\alpha} | \tau\bfe_3)$ will be called a \textit{helical isometry}\footnote{Alternately referred to as a \textit{screw transformation} in the crystallography literature \citep{Hahn}.} about axis $\bfe_3$. The action of $\Upsilon_\mathsf{h}$ on a point $\bfx \in \rz^3$ is to  rotate it by angle $2\pi\alpha$ about axis $\bfe_3$, while also translating it by $\tau$ along the same axis.\footnote{A simple way to see this is to resolve $\bfx$ along and perpendicular to $\bfe_3$, i.e., $\bfx = x_3 \bfe_3 + x^{\perp}\bfe_3^{\perp}$, where $\innprod{\bfe_3}{\bfe^{\perp}}{\rz^3} = 0$ and $\norm{\bfe_3^{\perp}}{\rz^3}=1$. Then, $\Upsilon_\mathsf{h} \circ \bfx = (x_3+\tau)\,\bfe_3 + x^{\perp} (\bfR_{2\pi\alpha} \bfe_3^{\perp})$.} Furthermore, applying the formulae for the powers of isometries and their inverses shown above, we see that for  $m=1,2,\ldots$, $\Upsilon_\mathsf{h}^m = (\bfR_{2\pi m \alpha} | m\tau\bfe_3)$ and $\Upsilon_\mathsf{h}^{-1} = (\bfR_{- 2\pi \alpha} | -\tau\bfe_3)$. Combining these, we may define $\Upsilon_\mathsf{h}^{m}$ for any $m \in \gz$ as $\Upsilon_\mathsf{h}^m = (\bfR_{2\pi m \alpha} | m\tau\bfe_3)$, with the $m=0$ case automatically resulting in the identity isometry $(\bfI| \textbf{0})$. We may therefore state:
\begin{proposition}[\textbf{Helical group generated by a single element}]
\label{Prop:single_generator_group}
The set of isometries 
\begin{align}
\label{eq:Group_G1}
\calG_1 = \big\{\Upsilon_\mathsf{h}^m = (\bfR_{2\pi m \alpha} | m\tau\bfe_3): m\in \gz \big\}\,,
\end{align} forms a {discrete} group under the operation $\bullet$. 
\end{proposition}
Additionally, let $\mathfrak{N} \in \nz$, let $\displaystyle {\Theta} = \frac{2\pi}{\mathfrak{N}}$ and for $n=0,1,\ldots,\mathfrak{N}-1$, let $\bfR_{n\Theta}$ denote a rotation around axis $\bfe_3$ by angle $n\Theta$. Then the set of isometries endowed with the operation $\bullet$
\begin{align}
\label{eq:Group_Cyclic}
\mathfrak{C}=\big\{\Upsilon_\mathsf{c}^n = (\bfR_{n\Theta}|\textbf{0}): n=0,1,\ldots,\mathfrak{N}-1 \big\}\,,
\end{align}
forms a cyclic group \citep{banerjee2016cyclic} of order $\mathfrak{N}$. Note that since the rotational parts of the isometries in group $\calG_1$ and  $\mathfrak{C}$ all share $\bfe_3$ as the common axis of rotation, the elements of $\calG_1$ and $\mathfrak{C}$ commute (i.e., for any $\Upsilon_\mathsf{h}^m \in \calG_1$ and $\Upsilon_\mathsf{c}^n \in \mathfrak{C}$, $\Upsilon_\mathsf{c}^n \bullet \Upsilon_\mathsf{h}^m = \Upsilon_\mathsf{h}^m \bullet \Upsilon_\mathsf{c}^n $ holds.) . We may now consider the direct product of the groups $\calG_1$ and $\mathfrak{C}$ defined above to obtain a new helical group\footnote{While the discussion presented here already makes it evident that both $\calG_1$ and $\calG_2$ are groups (and are in fact Abelian groups), see \citep{DEJ_ObjForm} for a more complete derivation of these groups, as well as other types of helical groups not considered in this work.}\footnote{Note that the groups $\calG_1$ and $\calG_2$ contain a group of translations as a normal subgroup if $\alpha$ is a rational number. In certain terminology \citep{Hahn, DEJ_ObjForm}, such cases would be identified as \textit{rod groups} and the term \textit{helical group} would be reserved only for cases for which $\alpha$ is an irrational number (i.e., when the group is not equivalent to a periodic group generated by a single translation.) However, we will not make this distinction here.}:
\begin{proposition}[\textbf{Helical group generated by two elements}]
\label{Prop:two_generator_group}
The set of isometries 
\begin{align}
\label{eq:Group_G2}
\calG_2 = \big\{\Upsilon_\mathsf{h}^m \bullet  \Upsilon_\mathsf{c}^n= (\bfR_{2\pi m \alpha + n\Theta} | m\tau\bfe_3): m\in \gz; n=0,1,\ldots,\mathfrak{N}-1  \big\}\,,
\end{align} forms a {discrete} group under the operation $\bullet$. 
\end{proposition}
Since $\calG_1$ and $\mathfrak{C}$ are generated by single elements, they are Abelian groups. Furthermore, since  $\calG_2$ is generated by two elements (i.e., the generators of $\calG_1$ and $\mathfrak{C}$) which commute among themselves, it is an Abelian group as well. 

The action of the groups $\calG_1$ and $\calG_2$ on points in space are easily described using cylindrical coordinates: if $\bfx\in\rz^3$ is point with cylindrical coordinates $(r,\vartheta,z)$, then the action of the group element $\Upsilon_h^m \in \calG_1$ is to send it to a point with cylindrical coordinates $(r,\vartheta + 2\pi m\alpha,z+m\tau)$, while the action of $\Upsilon_h^m \bullet \Upsilon_c^n \in \calG_2$ is to send it to the point with cylindrical coordinates $(r,\vartheta + 2\pi m\alpha + n\Theta,z+m\tau)$. In what follows, we will use the notation $\widetilde{\Upsilon}$ to denote a generic isometry from  $\calG_1$ or $\calG_2$. 
\subsubsection{Fundamental domains}
\label{subsubsec:Fundamental_Domains}
Given a point $\bfx \in \rz^3$, and a group of isometries $\calG$ (which could be the helical groups $\calG_1$ or $\calG_2$ described above, for instance), the \textit{orbit of} $\bfx$ under the group is the set
\begin{align}
\label{eq:group_orbit_point}
\calG \circ \bfx := \{\widetilde{\Upsilon} \circ \bfx: \widetilde{\Upsilon} \in \calG\}\,.
\end{align}
Given a collection of points $S \subset \rz^3$ and a group of isometries $\calG$,  we will use the notation $\calG \circ S$ to denote the orbits of each of the points in $S$ under the group:
\begin{align}
\label{eq:group_orbit_Set}
\calG \circ S := \bigcup_{\bfx \in S}\calG \circ \bfx \;\;\bigg(= \bigcup_{\widetilde{\Upsilon} \in \calG}\!\widetilde{\Upsilon}\circ S = \bigcup_{\substack{\widetilde{\Upsilon} \in \calG,\\ \bfx \in S}}\!\!\widetilde{\Upsilon} \circ \bfx \bigg)\,.
\end{align}
Let $\calO \subset \rz^3$ be a domain with regular boundary that is invariant under a given helical group\footnote{With these hypotheses, the boundary of $\calO$, denoted $\partial \calO$, can be shown to be invariant under the group as well.} $\calG$, i.e., $\calG \circ \calO = \calO$. The \textit{symmetry cell} or \textit{fundamental domain} of $\calG$ in $\calO$ is a set $D \subset \calO$ such that\footnote{In practice, we will require the fundamental domain to have some regularity properties in addition to the conditions in eq.~\ref{eq:FD_def_1} and \ref{eq:FD_def_2}, e.g. it should be connected and have compact closure.}:
\begin{align}
\label{eq:FD_def_1}
\bigcup_{\widetilde{\Upsilon} \in \calG}\!\widetilde{\Upsilon} \circ D = \calO\,,
\end{align}
and for $\widetilde{\Upsilon}_1, \widetilde{\Upsilon}_2 \in  \calG$:
\begin{align}
\label{eq:FD_def_2}
(\widetilde{\Upsilon}_1 \circ D)\,\bigcap\,(\widetilde{\Upsilon}_2 \circ D) =\,\text{a set of Lebesgue measure}\;0\;\text{for}\;\widetilde{\Upsilon}_1 \neq \widetilde{\Upsilon}_2.
\end{align}

To see concrete examples of the sets $\calO$ and ${D}$, let $\mathscr{D}_R$ denote an open disk of radius $R$ on the $\bfe_1,\bfe_2$ plane, i.e., 
\begin{align} 
\mathscr{D}_R = \{\bfp \in \rz^2: \norm{\bfp}{\rz^2} < R \}\,,
\label{eq:Disk}
\end{align}
and let $\calC$ denote the infinite cylinder obtained by translating $\mathscr{D}_R$ along $\bfe_3$, i.e.:
\begin{align}
\label{eq:Infinite_Cylinder}
\calC = \mathscr{D}_R \times \{x_3\,\bfe_3: x_3\in \rz\}\,.
\end{align}
Then, the cylinder $\calC$ has all the properties required of the domain $\calO$. Furthermore, we observe that the finite cylinder $\calD_{\calG_1} = \mathscr{D}_R \times \{x_3\,\bfe_3: 0 \leq x_3 <  \tau \}$  serves as the fundamental domain of $\calG_1$ in $\calC$. Finally, the sector with slanted walls, described in cylindrical coordinates as:
\begin{align}
\label{eq:Cylindrical_Sector}
\calD_{\calG_2} = \big\{(r,\vartheta,z):0 \leq r<R , \frac{2\pi\alpha z}{\tau} \leq \vartheta < \frac{2\pi}{\mathfrak{N}} + \frac{2\pi\alpha z}{\tau} , 0 \leq z < \tau\big\}\,,
\end{align}
serves as the fundamental domain of $\calG_2$ in $\calC$.
\subsubsection{Helical Structures}
\label{subsubsec:Helical_Structures}
A helical structure i.e., an atomic/molecular structure with helical symmetries is simply the orbit of a set of non-degenerate points under the action of one of the helical groups $\calG_1$ or $\calG_2$. More precisely, let $\calP_{\calG_1} \subset \calD_{\calG_1}$ \big(or $\calP_{\calG_2} \subset \calD_{\calG_2}$ in case of $\calG_2$\big) be a finite collection of distinct points labeled $\big \{\bfx_k \big\}_{k=1}^{M_{\calG_1}}$ \big(or $\big \{\bfx_k \big\}_{k=1}^{M_{\calG_2}}$ in case of $\calG_2$\big). These points are representative of atomic positions within the fundamental domain and we will refer to them as \textit{simulated points} or \textit{simulated atoms}. The (valence) nuclear charges corresponding to these atoms will be denoted as $\big \{Z_k \big\}_{k=1}^{M_{\calG_1}}$ \big(or $\big \{Z_k \big\}_{k=1}^{M_{\calG_2}}$ in case of $\calG_2$\big). A helical structure is simply a set of the form:
\begin{align}
\label{eq:Helical_Structure_1}
\calS_{\calG_1,\calP_{\calG_1}} &= \calG_1 \circ \calP_{\calG_1} =  \bigcup_{\substack{\widetilde{\Upsilon} \in \calG_1,\\ k =1,\ldots, M_{\calG_1}}}\!\!\!\widetilde{\Upsilon}\circ \bfx_k\,,\\
\text{or}\;\;\calS_{\calG_2,\calP_{\calG_2}} &= \calG_2 \circ \calP_{\calG_2} = \bigcup_{\substack{\widetilde{\Upsilon} \in \calG_2,\\  k =1,\ldots, M_{\calG_2}}}\!\!\!\widetilde{\Upsilon} \circ \bfx_k\,.
\label{eq:Helical_Structure_2}
\end{align}
Additionally, for any $\widetilde{\Upsilon} \in \calG_1$, the atom at the location $\widetilde{\Upsilon}\circ\bfx_k$ is taken to be of the same species as the atom at $\bfx_k \in \calP_{\calG_1}$ (similarly also for $\widetilde{\Upsilon} \in \calG_2$ and $\bfx_k \in \calP_{\calG_2}$), and so it is associated with the same (valence) nuclear charge $Z_k$.
\subsection{Kohn-Sham Problem for Helical Structures}
\label{subsec:Kohn_Sham_Helical}
The Kohn-Sham equations, as they apply to finite structures can be found in numerous references \citep{LeBris_ReviewBook, ghosh2019symmetry, ghosh2017sparc_1}.  In order to formulate an appropriate version of the Kohn-Sham equations for helical structures however, we need to keep in mind a few typical features of such a structure. In what follows, for the sake of simplicity, we will consider in detail the case of a structure associated with a helical group generated by a single element (i.e., the group $\calG_1$ described above). {We will comment on modifications to the above case that need to be considered while dealing with a structure associated with a helical group generated two elements (i.e., the group $\calG_2$ described above), and present the final expressions/equations for this case in \ref{appendix:two_element_group_expressions}.} A more detailed discussion of the modifications and the application of resulting equations is the scope of ongoing and future work \citep{Helical_DFT_Paper_2}.

Helical structures are essentially quasi-one-dimensional in nature. This implies that they have limited spatial extent in the $\bfe_1,\bfe_2$  plane, while being infinitely extended along the $\bfe_3$ direction. Consequently, it is appropriate to set up the Kohn-Sham equations for such a structure in a computational domain which is of limited spatial extent in the $\bfe_1,\bfe_2$ plane, while being infinite in extent along $\bfe_3$. This, along with the requirement that a symmetry adapted formulation of the Kohn-Sham equations needs to be solved on a domain that is also invariant with respect to the symmetry operations of the helical structure, suggests the cylinder $\calC$ as being a natural choice for the computational domain (for a helical structure generated by a single element). The radius of this cylinder has to be consistent with the requirements that the all the atoms of the helical structure should be located sufficiently away from the lateral surface of the cylinder so as to allow sufficient decay of various fields that appear in the Kohn-Sham problem.

The quasi-one-dimensional nature of the systems under study results in additional complications. Specifically, due to the infinite extent of the system along the $\bfe_3$ direction, the system is associated with an infinite  number of electronic states\footnote{In general, these states would be expected to be delocalized over the entire volume of the cylinder $\calC$.} as well as an infinite number of nuclei. This potentially poses divergence issues while computing the electrostatics terms in the Kohn-Sham problem \citep{Dumitrical_Ewald, ghosh2019symmetry} and it is dealt with in this work by solving an appropriate symmetry adapted Poisson problem involving a neutral charge distribution --- such a charge distribution arises as a combination of the electron density and the nuclear pseudocharges associated with the structure. Additionally, the infinitely many electronic states have to be incorporated into the Kohn-Sham problem in a manner that is consistent with the Pauli exclusion principle and the Aufbau principle \citep{LeBris_Defranceschi_ReviewBook,LeBris_ReviewBook}. Taking cue from the solid state/condensed matter physics literature --- specifically, ab initio calculations of crystalline solids \citep{Martin_ES, LeBris_Defranceschi_ReviewBook} --- we address this issue here by formulating a band theory of electronic states for helical structures. This allows the Kohn-Sham problem for the entire helical structure, as posed on the cylinder $\calC$, to be reduced to computations on the fundamental domain when augmented with appropriate boundary conditions.\footnote{{Note that  we are not attempting to solve the thermodynamic limit problem associated with the helical structure in this work. Instead, we are postulating the form of the governing equations at the thermodynamic limit (i.e., for the infinite helical structure which is under study) and expressing them in a mathematically rigorous manner. This is necessary so that we can then numerically solve these equations and extract physical properties of systems of interest. In contrast, the thermodynamic limit problem would involve the passage from a finite (truncated) helical structure to the infinite one keeping various energetic contributions in mind, and is far beyond the scope of the current contribution.}}

A key ingredient of the band theory for helical structures is an appropriate version of the Bloch theorem \citep{Bloch, Martin_ES, Ashcroft_Mermin} for such systems. The form of this mathematical result can be guessed by looking at the analogous case of the Bloch Theorem for one dimensional periodic systems\footnote{See e.g.\! equations $25,26$ in \citep{ghosh2019symmetry}.} and the result appears to have been made use of by earlier researchers in various contexts \citep{My_PhD_Thesis, Pekka_1,Pekka_2,Mintmire_White1, Serbian_Group1, Serbian_Group2, CNT_1, CNT_2, CNT_3, CNT_Dumitrica, CNT_4, CNT_5}. However, a rigorous mathematical derivation of the result does not seem to appear anywhere in the literature --- other than in  \citep{My_PhD_Thesis}, where a proof of the existence of Helical Bloch waves was sketched by using tools from the theory of linear elliptic partial differential equations. In what follows, we address this gap in the literature and follow up on \citep{My_PhD_Thesis}, by establishing the existence and completeness of helical Bloch waves, and then use this to gain insight into the spectrum of the single electron Hamiltonian associated with helical systems (i.e., to set up an electronic band theory for such systems). This information is subsequently used to set up the governing equations of the system. Our mathematical treatment closely follows the techniques presented in references \citep{Odeh_Keller, My_PhD_Thesis, dorfler2011photonic, Wilcox, Reed_Simon4}.
\subsubsection{Analysis of the single electron problem for helical structures - helical Bloch waves}
\label{subsubsec:Helical_Bloch_Waves}
As a starting point, we consider the single electron Hamiltonian:
\begin{align}
\hamil = -\half \Delta + V(\bfx) = -\half\bigg(\hpd{}{x_1}{2}+\hpd{}{x_2}{2}+\hpd{}{x_3}{2}\bigg) +V(x_1,x_2,x_3) \,,
\label{eq:single_electron_Hamiltonian}
\end{align}
with the real valued continuous potential $V(\bfx)$ invariant under the helical group $\calG_1$, i.e., 
\begin{align}
V(\widetilde{\Upsilon} \circ \bfx) = V(\bfx),\forall\,\widetilde{\Upsilon} \in \calG_1\,.
\label{eq:potential_invariance}
\end{align}
This operator naturally arises during each self-consistent field iteration cycle in Kohn-Sham calculations\footnote{Within the setting of the local density approximation and the use of local pseudopotentials for example, $V(\bfx)$ can be identified as the total effective potential appearing in the Kohn-Sham equations and can be written as the sum of electrostatic and exchange-correlation terms, i.e., $V(\bfx) = V_{\text{es}}(\rho(\bfx)) + V_{\text{xc}}(\rho(\bfx))$.} and in that scenario, the invariance of the potential automatically follows from the invariance of the electron density \citep{My_PhD_Thesis}. 

We are interested in functions $\psi$ that satisfy the equation $\hamil  \psi = \lambda \psi$ within the region $\calC$ in an appropriate manner. Additionally, to model the decay of the eigenstates as one moves away from the axis of the cylinder to infinity \citep{wavefunc_decay1, wavefunc_decay2}, we will enforce Dirichlet boundary conditions on the lateral surfaces of the cylinder\footnote{This ``wire'' boundary condition is commonly employed in the literature for studying quasi-1D systems \citep{han2008real, ghosh2017sparc_2}. This boundary condition allows the operator $\hamil$ to have some convenient properties without having to enforce any specific decay conditions on $V(\bfx)$ as one moves away from the axis of the cylinder.}, i.e., $\psi(\bfx) = 0$ for $\bfx \in \partial\calC$.\footnote{In what follows, we will use the following notation: if $\mathsf{A}$ is a measure space with measure $\mu$, then for $1 \leq p < \infty$, we will use $\mathsf{L}^{p}(\mathsf{A},\mathsf{B}, \mu)$ to denote Lebesgue measurable functions $f:\mathsf{A} \to \mathsf{B}$ which satisfy $\displaystyle \int_{\mathsf{A}} \norm{f}{\mathsf{B}}^p\,d\mu < \infty$, and we will use $\mathsf{L}^{\infty}(\mathsf{A},\mathsf{B}, \mu)$ to denote functions for which $\displaystyle \textsf{ess.\,sup.}_{x\in \mathsf{A}} \norm{f(x)}{\mathsf{B}} < \infty$. In particular, if $\mathsf{A}$ is a domain in $\rz^3$, we will use $\Lpspc{2}{}{A}$ to denote the usual Hilbert space of complex valued functions on $\mathsf{A}$ which are square integrable (using the Lebsegue measure). The inner product of two functions on this space will be expressed as:
\begin{align}
\innprod{f_1}{f_2}{\Lpspc{2}{}{\mathsf{A}}} = \int_{\mathsf{A}} f_1(\bfx)\overline{f_2(\bfx)}\,d\bfx
\end{align}
Furthermore, $\mathsf{H}^k(\mathsf{A})$ will denote the Sobolev space of tempered distributions whose $k^{\text{th}}$ weak derivative lies in $\Lpspc{2}{}{\mathsf{A}}$, while $\mathsf{H}^1_{0}(\mathsf{A})$ will denote the subspace of functions in $\mathsf{H}^1(\mathsf{A})$ which vanish at the boundary of $\mathsf{A}$ in the trace sense. {Finally, the rank one operator created as the tensor product of two functions $f_1,f_2 \in \Lpspc{2}{}{\mathsf{A}}$, i.e., $f_1 \otimes \overline{f_2}$ will act on a generic function $f \in \Lpspc{2}{}{\mathsf{A}}$ to yield $\innprod{f}{f_2}{\Lpspc{2}{}{\mathsf{A}}}f_1$.}}\footnote{We may view $\hamil$ as an unbounded operator on $\Lpspc{2}{}{\calC}$ with the function space $\text{Dom.}(\hamil)=\mathsf{H}^2_{}(\calC) \cap \mathsf{H}^1_{0}(\calC)$ as the domain of the operator. The operator $\hamil$ is formally symmetric (or, in linear algebra terminology, \textit{Hermitian} since the underlying function spaces are complex): if $f_1,f_2$ are Schwartz functions in $\calC$ which obey the boundary condition $f_1(\bfx) = f_2(\bfx) = 0$ for $\bfx \in \partial\calC$, we have:
\begin{align}
\label{eq:symmetric_hamiltonian_1}
\innprod{\hamil f_1}{f_2}{\Lpspc{2}{}{\calC}} = -\half \int_{\calC} \Delta f_1 \overline{f_2}\,d\bfx + \int_{\calC} V f_1 \overline{f_2}\,d\bfx\,.
\end{align}
 On using integration by parts \citep{Evans_PDE} and the decay of $f_1$ and $f_2$ as $x_3 \to \infty$, we get:
\begin{align}
\label{eq:symmetric_hamiltonian_2}
-\half \int_{\calC} \overline{f_2}\Delta f_1 \,d\bfx = \half \bigg(\int_{\calC} \nabla \overline{f_2}\cdot \nabla f_1 \,d\bfx - \int_{\partial \calC}\overline{f_2} \nabla f_1 \cdot d\bfs\bigg)\,.
\end{align}
Here $d\bfs$ denotes the oriented surface measure. The second term on the right-hand side above vanishes due to the boundary conditions obeyed by $f_1,f_2$ on $\partial\calC$ and so, this leaves us with:
\begin{align}
\label{eq:symmetric_hamiltonian_3}
\innprod{\hamil f_1}{f_2}{\Lpspc{2}{}{\calC}} = \half  \int_{\calC} \nabla \overline{f_2}\cdot \nabla f_1 \,d\bfx + \int_{\calC} V f_1 \overline{f_2}\,d\bfx\,.
\end{align}
In a similar manner, we get:
\begin{align}
\label{eq:symmetric_hamiltonian_4}
\innprod{f_1}{\hamil f_2}{\Lpspc{2}{}{\calC}} &= \half \int_{\calC} \nabla {f_1} \cdot \nabla{ \overline{f_2}} \,d\bfx + \int_{\calC} f_1 \overline{V f_2} \,d\bfx \\
&= \innprod{\hamil f_1}{f_2}{\Lpspc{2}{}{\calC}} \,,
\label{eq:symmetric_hamiltonian_5}
\end{align}
as the potential $V(\bfx)$ is real. Since Schwartz functions are dense in the domain of $\hamil$, the result follows. The direct integral decomposition of $\hamil$ (\ref{appendix:direct_integrals}) makes it easy to appreciate that $\hamil$ is in fact self-adjoint.} 
Helical Bloch waves (or helical Bloch states) are solutions to the above equation which have the ansatz: 
\begin{align}
\label{eq:Bloch_ansatz_1}
\psi(\bfx) = e^{-i2\pi \eta \frac{x_3}{\tau}}\phi(\bfx;\eta),
\end{align}
Here $\phi(\bfx; \eta)$ group invariant i.e., 
\begin{align}
\label{eq:Bloch_ansatz_2}
\phi(\bfx; \eta) = \phi(\widetilde{\Upsilon} \circ \bfx; \eta), \forall\widetilde{\Upsilon} \in \calG_1\,.
\end{align}
and obeys the boundary condition:
\begin{align}
\label{eq:Bloch_ansatz_3}
\phi(\bfx; \eta) = 0\,\text{for}\,\bfx \in \partial\calC,
\end{align}
commensurate with the boundary condition on $\psi$. The parameter $\eta$ serves a role that is analogous to k-points in periodic calculations and as shown later, it can be chosen such that $\eta \in [-\half,\half)$. {In what follows, we first show the existence of such solutions and then demonstrate their completeness. In essence, these results together give us information that  certain special electronic states (i.e., helical Bloch states) can be \textit{always} found to be associated with the single electron Hamiltonian of a helical structure, and they further inform us that such special states can be used to characterize \textit{all} of the possible electronic states of the system (within the single electron model). Therefore, it is sufficient for us to restrict our attention to these states while discussing the spectrum of the single electron Hamiltonian associated with a helical structure. Our derivation of these results follows techniques employed in classic references on the mathematical theory of Bloch waves in crystals \citep{Odeh_Keller, Wilcox} and builds the theory in a ``bottom up'' manner using standard tools from functional analysis and the theory of linear elliptic operators (see \citep{Renardy_Rogers, Evans_PDE, Folland_Real} for relevant background material). In subsequent sections (Section \ref{subsubsec:Governing_Equations}, \ref{appendix:direct_integrals}), we use techniques presented in \citep{Reed_Simon4} to use helical Bloch waves for ``block-diagonalizing'' the single electron Hamiltonian through the apparatus of direct integrals, and then use this formalism to derive governing equations.}

First, to demonstrate the existence of these special solutions, we have:
\begin{theorem}[\textbf{Existence theorem for helical Bloch waves}]
\label{Thm:Existence_Helical_Bloch}
Let $V(\bfx)$ be a real valued continuous potential that is invariant under the helical group $\calG_1$ and let $\hamil$ denote the operator $-\half \Delta + V(\bfx)$. For any $\eta \in \rz$ there exist a countable number of solutions of the equation $\hamil \psi = \lambda \psi$ in $\calC$ which are expressible in terms of the helical Bloch ansatz (eqs.~\ref{eq:Bloch_ansatz_1}-\ref{eq:Bloch_ansatz_3}).
\end{theorem}
\begin{proof}
We fix $\eta \in \rz$ and substitute the helical Bloch wave ansatz in the equation $\hamil \psi = \lambda \psi$ to find that $\phi(\bfx; \eta)$ should obey the following auxiliary equation in the region $\calC$:
\begin{align}
\mathfrak{h}_{\eta}^{\textsf{aux}}\phi =-\half \bigg( \Delta \phi -i \frac{4\pi \eta}{\tau} \pd{\phi}{x_3} - \frac{4\pi^2\eta^2}{\tau^2}{\phi}\bigg) + V\phi=  \lambda \phi\,.
\label{eq:phi_equation_1}
\end{align}
Additionally, $\phi$ should be group invariant and obey the zero Dirichlet boundary condition on $\partial \calC$. Let $\calD$ denote the interior of the fundamental domain $\calD_{\calG_1}$, i.e., it is the open set described in cylindrical coordinates as $\calD = \{(r,\vartheta,z): 0 \leq r < R, 0 < z < \tau\}$. The boundary of $\calD$ includes the lateral surface $\partial\calD^{\,r=R} =  \{(r,\vartheta,z):  r = R, 0 < z < \tau\}$ that is shared with $\calC$, as well as the discs $\partial\calD^{\,z=0} =  \{(r,\vartheta,z):  0 \leq r \leq R,z = 0\}$ and $\partial\calD^{\,z=\tau} =  \{(r,\vartheta,z):  0 \leq r \leq R,z = \tau\}$, which are both parallel to the $\bfe_1-\bfe_2$ plane. The group operation $\Upsilon_\mathsf{h}$ (i.e., the generator of the group $\calG_1$) maps $\partial\calD^{\,z=0}$ to $\partial\calD^{\,z=\tau}$ and conversely, $\Upsilon_\mathsf{h}^{-1}$ maps $\partial\calD^{\,z=\tau}$ to $\partial\calD^{\,z=0}$.

We now restrict the auxiliary eigenvalue problem $\mathfrak{h}_{\eta}^{\textsf{aux}}\phi =  \lambda \phi$ as outlined in eq.~\ref{eq:phi_equation_1}, to the region $\calD$ by imposing the boundary conditions $\phi(\bfx, \eta)=\phi(\Upsilon_\mathsf{h} \circ \bfx, \eta)$, $\nabla\phi(\bfx, \eta)=\bfR_{2\pi\alpha}^{-1}\nabla\phi(\Upsilon_\mathsf{h} \circ \bfx, \eta)$ for $\bfx \in \partial\calD^{\,z=0}$ and (as before), $\phi(\bfx, \eta) = 0$ for $\bfx \in \partial\calD^{\,r=R}$. The operator $\mathfrak{h}_{\eta}^{\textsf{aux}}$ on $\Lpspc{2}{}{\calD}$ is uniformly elliptic, and as shown in \ref{appendix:h_eta_symmetric}, it is also symmetric with the above boundary conditions.  Since $\calD$ is a bounded domain and $V(\bfx) \in \Lpspc{\infty}{}{\calD}$, the operator $\mathfrak{h}_{\eta}^{\textsf{aux}}$ satisfies the conditions of G\aa rding's Inequality (Theorems 9.17, 9.18 in \citep{Renardy_Rogers}; Section 6.2 in \citep{Evans_PDE}). This guarantees that $\mathfrak{h}_{\eta}^{\textsf{aux}}$ has a unique self-adjoint extension in $\Lpspc{2}{}{\calD}$, which we also denote as $\mathfrak{h}_{\eta}^{\textsf{aux}}$ here. Furthermore, as a consequence of the Rellich-Kondrachov Compactness Theorem  (Theorem 7.29 in \citep{Renardy_Rogers};  Section 5.7 in \citep{Evans_PDE}), $\mathfrak{h}_{\eta}^{\textsf{aux}}$ can be shown to have a compact resolvent (Lemma 9.20 in \citep{Renardy_Rogers}). Consequently, $\mathfrak{h}_{\eta}^{\textsf{aux}}$ has a discrete set of eigenvalues $\lambda_j({\eta})$ and corresponding eigenfunctions $\phi_j(\bfx; \eta)$ (Theorem 6.29 in \citep{Kato}; Theorem 9.22 in \citep{Renardy_Rogers}). Each eigenvalue is of finite multiplicity and such that $\lambda_j({\eta}) \to \infty$ as $j \to \infty$. Results from elliptic regularity theory (Sections 9.5, 9.6 in \citep{Renardy_Rogers}; Section 6.3 in \citep{Evans_PDE}) imply that $\phi_j(\cdot; \eta) \in \mathsf{H}^2(\calD)$. We now use the boundary conditions on $\phi$ outlined above to extend the eigenfunctions $\phi_j(\bfx; \eta)$ to all of $\calC$, noting that  these boundary conditions are meaningful in the trace sense since the eigenfunctions are in $\mathsf{H}^2(\calD)$. Thereafter, defining $\psi_j(\bfx, \eta) = e^{-i2\pi \eta \frac{x_3}{\tau}}\phi_j(\bfx,\eta)$, for $j\in \nz$ and $\bfx \in \calC$, establishes the theorem.
\end{proof}
We define $\Lambda = \{\lambda_j({\eta}): \eta \in \rz, j\in \nz\}$ and $\Psi = \{\psi_j(\cdot;\eta): \eta \in \rz, j \in \nz\}$ as the collection of generalized eigenvalues and generalized eigenfunctions\footnote{The real numbers $\lambda_j({\eta})$ are generalized eigenvalues of $\hamil$ since (as discussed later) they are part of the essential spectrum of $\hamil$ and not its discrete spectrum.
On a similar note, the functions $\psi_j(\cdot;\eta)$ do not belong in $\Lpspc{2}{}{\calC}$ and therefore, they are not eigenfunctions of $\hamil$ in the usual sense. However, as discussed above, they do satisfy an equation of the form $\hamil\,\psi_j(\cdot; \eta) = \lambda(\eta)\,\psi_j(\cdot; \eta)$, thus suggesting their similarity to conventional eigenvalues and eigenfunctions.} associated with $\hamil$. The first observation we make is that the sets $\Lambda$ and $\Psi$ are unchanged upon restricting $\eta \in [-\half,\half)$. To see this, we recall that $\lambda_j({\eta})$ and $\psi_j(\cdot;\eta)$ are obtained by computing the spectrum of $\hamil$ when subjected to the conditions\footnote{The Dirichlet boundary condition in eq.~\ref{eq:Bloch_ansatz_3} is also  obeyed equivalently by $\psi$ and does not need to be further considered here.}  in eqs.~\ref{eq:Bloch_ansatz_1}, \ref{eq:Bloch_ansatz_2}. However, these equations can be equivalently recast as the following condition on $\psi$: 
\begin{align}
\label{eq:Bloch_ansatz_4}
\psi(\Upsilon_\mathsf{h} \circ \bfx) = e^{-i2\pi \eta \frac{x_3 + \tau}{\tau}}\phi(\Upsilon_\mathsf{h}  \circ \bfx;\eta) = e^{-i2\pi \eta}e^{-i2\pi \eta \frac{x_3}{\tau}}\phi(\bfx;\eta) = e^{-i2\pi \eta}\psi(\bfx)\,
\end{align}
or more generally, for $m\in \gz$:
\begin{align}
\label{eq:Bloch_ansatz_5}
\psi(\Upsilon^m_\mathsf{h} \circ \bfx) =  e^{-i2\pi m\eta}\psi(\bfx)\,.
\end{align}
In other words, solving $\hamil \psi = \lambda \psi$ while imposing the condition $\psi(\Upsilon_\mathsf{h} \circ \bfx) =e^{-i2\pi \eta}\psi(\bfx)$ on $\psi$ also gives us the sets $\Lambda$ and $\Psi$. Since $e^{-i2\pi \eta} = e^{-i2\pi (\eta + n)}$ for any $n \in \gz$, we see that the boundary conditions on $\psi$ do not change upon translating the value $\eta$ by an integer. Thus, it suffices to restrict $\eta \in [-\half,\half)$.  In what follows, we will denote $\mathfrak{I} = [-\half,\half)$, and we will re-define $\Lambda = \{\lambda_j(\eta): \eta \in \mathfrak{I}, j \in \nz\}$ and $\Psi = \{\psi_j(\cdot;\eta): \eta \in \mathfrak{I}, j \in \nz\}$. In keeping with solid state physics terminology, we will refer to the set $\mathfrak{I}$ as \textit{reciprocal space} (or more specifically, the Brillouin zone of the  reciprocal space). Consequently, the dependence of a quantity on $\eta$ will be termed as reciprocal space dependence while its dependence on usual  physical space will be termed as real space dependence.

For a given $j \in \nz$, we will refer to the set $\Lambda_j =  \{\lambda_j({\eta}): \eta \in \mathfrak{I})\}$ as a \textit{helical band}. Results from the theory of regular perturbations of self-adjoint problems \citep{rellich1969perturbation, Kato} imply that (for a fixed $j  \in \nz$) the map $\eta \mapsto \lambda_j(\eta)$ is analytic. Therefore, the set $\Lambda_j$ is connected and compact.\footnote{\label{footnote:physics_proof} In contrast to the rigorous proof presented above, a formal derivation of the Bloch theorem for a helical structure, inspired by the solid state physics literature \citep{Ashcroft_Mermin, Kittel} is as follows: We observe that since the Laplacian commutes with all isometry operations -- including those that constitute the group $\calG_1$, and further, since the potential  $V(\bfx)$ is group invariant  (eq.~\ref{eq:potential_invariance}), the operator $\hamil$ must commute with the symmetry operations in the group $\calG_1$. Specifically, for any continuous function $f$ defined over $\calC$, we may define the operators:
\begin{align}
\mathcal{T} = \big\{T_{\widetilde{\Upsilon}}: T_{\widetilde{\Upsilon}}f(\bfx) = f(\widetilde{\Upsilon}^{-1}\circ\bfx)\big\}_{\widetilde{\Upsilon} \in \calG_1}\,.
\end{align}
Then, for any function $f$ in the domain of $\hamil$, the relationship $T_{\widetilde{\Upsilon}} \hamil f = \hamil T_{\widetilde{\Upsilon}} f$ holds for any $T_{\widetilde{\Upsilon}} \in \mathcal{T}$. This commutation property can be used to infer that the unitary representations of $\calG_1$ and the operator $\hamil$ can be ``simultaneously diagonalized'' in a suitable basis of common ``eigenstates''. Since $\calG_1$ is an Abelian group, its irreducible representations are all one-dimensional \citep{McWeeny, Folland_Harmonic}.  Furthermore, these irreducible representations can be used to decompose any unitary representation of the group \citep{Folland_Harmonic, Barut_Reps}. This suggests therefore that the eigenstates associated with $\hamil$ transform under the group  in a manner similar to the irreducible representations of $\calG_1$, which then implies the helical Bloch theorem. While the above argument is perhaps  correct in spirit and variants of the argument appear often in the physics literature (in the context of periodic systems) it has a number of technical deficiencies owing to the fact that $\hamil$ is an unbounded operator and the group $\calG_1$ is infinite. These issues prevent heuristic arguments like the one above -- which are more suited to representations of finite groups on finite dimensional spaces -- from being applied in the current context. In particular e.g., Bloch states are not eigenfunctions in the usual sense since they are not square integrable.} 

We will refer to the set $\Psi = \{\psi_j(\cdot;\eta): \eta \in \mathfrak{I}, j \in \nz\}$ as the  collection of \textit{helical Bloch states} corresponding to the helical bands. If we fix $\eta \in \mathfrak{I}$, then the set $\Psi_{\eta} = \{\psi_j(\cdot;\eta):  j \in \nz\}$ has the property that it is orthonormal and complete in $\Lpspc{2}{}{\calD}$. This follows directly from the properties of the group invariant functions $\phi_j(\bfx; \eta)$ defined above. Specifically, for $j,j' \in \nz$:
\begin{align}
\nonumber
&\innprod{\psi_j(\cdot; \eta)}{\psi_{j'}(\cdot; \eta)}{\Lpspc{2}{}{\calD}} = \int_{\calD}\psi_j(\bfx,\eta)\,\overline{\psi_{j'}(\bfx,\eta)}\,d\bfx \\&= \int_{\calD} e^{-i2\pi\eta\frac{x_3}{\tau}}\phi_j(\bfx, \eta)\,e^{i2\pi\eta\frac{x_3}{\tau}}\overline{\phi_{j'}(\bfx, \eta)}\,d\bfx = \int_{\calD}\phi_j(\bfx, \eta)\,\overline{\phi_{j'}(\bfx, \eta)}\,d\bfx = \delta_{j,j'}\,.
\label{eq:innprod_psi_j}
\end{align}
Furthermore, if $h \in \Lpspc{2}{}{\calD}$ such that $\innprod{h(\cdot)}{\psi_{j'}(\cdot; \eta)}{\Lpspc{2}{}{\calD}} = 0$ for every $j\in \nz$, then we must have:
\begin{align}
\int_{\calD} h(\bfx)\,\overline{\psi_{j}(\bfx,\eta)}\,d\bfx = \int_{\calD} h(\bfx)\,e^{i2\pi\eta\frac{x_3}{\tau}}\overline{\phi_{j}(\bfx,\eta)}\,d\bfx = 0,\forall j\in \nz.
\end{align}
Due to the completeness of the functions $\phi_j(\cdot; \eta)$ it then follows that $h(\bfx)\,e^{i2\pi\eta\frac{x_3}{\tau}} = 0$, i.e., $h(\bfx) = 0$ almost everywhere in $\calD$. Thus, the set $\Psi_{\eta}$ is complete in $\Lpspc{2}{}{\calD}$.

Due to the completeness of the set $\Psi_{\eta}$ for each $\eta \in \mathfrak{I}$, it actually follows that the set of helical Bloch states (i.e., the set $\Psi$) is complete in $\Lpspc{2}{}{\calC}$. {To prove this important result, we first need to establish a few preliminaries related to the so-called helical Bloch-Floquet transform, i.e., an analogue of the classical Bloch-Floquet transform \citep{Reed_Simon4, dorfler2011photonic}, as extended to the case of helical symmetries. Specifically, we show that there is a one-to-one correspondence between functions in $\Lpspc{2}{}{\calC}$ and $\Lpspc{2}{}{{\calD} \times \mathfrak{I}}$ (this is the content of Lemmas \ref{Lemma:completeness_1} and  \ref{Lemma:completeness_2}), and we then identify the helical Bloch-Floquet transform as an operator which maps between these spaces\footnote{{Throughout this paper, we will often write functions in $\Lpspc{2}{}{{\calD} \times \mathfrak{I}}$ as $f(\bfx,\eta)$ as well as $f(\bfx;\eta)$. The latter notation is meant to emphasize such functions as being  $\eta$-parametrized members of $\Lpspc{2}{}{{\calD}}$. However, as pointed out by an anonymous reviewer, it is perhaps not always possible to make this distinction  consistently.}}. Thereafter, the completeness of $\Psi_{\eta}$ in $\Lpspc{2}{}{\calD}$ for each $\eta \in \mathfrak{I}$, in conjunction with the use of the helical Bloch-Floquet transform can be used to demonstrate the completeness of helical Bloch states in $\Lpspc{2}{}{\calC}$ (this being the content of Theorem \ref{Thm:Completeness_Helical_Bloch}). The completeness result of the helical Bloch states is intimately connected to the direct integral decomposition of the Hamiltonian, which we use for deriving the governing equations in the next section.} %In what follows, we will denote $\mathfrak{I} = [-\half,\half)$ and we will use $\overline{\calD}$ to denote the closure of ${\calD}$.

We have:
\begin{lemma}
Let ${f} \in \Lpspc{2}{}{\calC}$, $\eta \in \mathfrak{I}$ and $m \in \gz$. We define:
\begin{align}
g(\bfx, \eta) := \sum_{m\in \gz} {f}(\Upsilon^m_\mathsf{h} \circ \bfx)\,e^{i 2 \pi m\eta}\,.
\label{eq:lemma_2_1}
\end{align}
Then $g$ is defined almost everywhere in ${\calD}$ and further, $g \in \Lpspc{2}{}{{\calD} \times \mathfrak{I}}$.
\label{Lemma:completeness_1}
\end{lemma} 
\begin{proof}
We denote $\tilde{f}_m(\bfx) = {f}(\Upsilon^m_\mathsf{h} \circ \bfx)$. Then, $\displaystyle g(\bfx, \eta) = \sum_{m\in \gz}\tilde{f}_m(\bfx)\,e^{i 2 \pi m\eta}$. {By use of the Fubini-Tonelli theorem, we now observe \citep{Folland_Real} that:\footnote{{We would like to thank the anonymous reviewers for their comments which helped clarify and fix certain technical aspects of this proof, including the suggestion that Plancharel's Theorem \citep{Folland_Real} can be used to make certain statements in the above proof more precise.}}
\begin{align}
\nonumber
 \int_{\bfx \in \calD} \sum_{m\in \gz}  \abs{\tilde{f}_m(\bfx)}^2 \,d\bfx  &= \sum_{m\in \gz}  \int_{\bfx \in \calD} \abs{\tilde{f}_m(\bfx)}^2 \,d\bfx = \sum_{m\in \gz} \int_{\bfx \in \calD}\!\abs{f(\Upsilon^m_\mathsf{h} \circ \bfx)}^2 \,d\bfx \\
&= \sum_{m\in \gz} \int_{\bfx\in \Upsilon^m_\mathsf{h} \circ \calD}\!\abs{{f}(\bfx)}^2 d\bfx = \int_{\bfx \in \calC}\!\abs{f(\bfx)}^2 \,d\bfx < \infty\,.
\label{eq:lemma_2_2}
\end{align}
This establishes that the function $h(\bfx) = \displaystyle \sum_{m\in \gz}  \abs{\tilde{f}_m(\bfx)}^2$ is finite for almost every $\bfx \in \calD$, since $h(\bfx) = \infty$ on a set of non-zero measure would violate eq.~\ref{eq:lemma_2_2}. Thus, for almost every $\bfx \in {\calD}$, the sequence $\displaystyle \big\{ \tilde{f}_m(\bfx) \big\}_{m \in \gz}$ is square summable, and the expression for $g$ in eq.~\ref{eq:lemma_2_1} can be interpreted as a Fourier expansion (in the $\eta$ variable). We may now use Parseval's identity  \citep{Folland_Real} and eq.~\ref{eq:lemma_2_1} to obtain:
\begin{align}
\int_{\mathfrak{I}} \abs{g(\bfx, \eta)}^2\,d\eta = \sum_{m\in \gz} \abs{\tilde{f}_m(\bfx)}^2\,.
\label{eq:lemma_2_3}
\end{align}
Integrating both sides of this expression for $\bfx \in \calD$ and using the steps in eq.~\ref{eq:lemma_2_2} establishes that $g \in \Lpspc{2}{}{{\calD} \times \mathfrak{I}}$, as required.
}
\end{proof}
The following result is the converse of Lemma \ref{Lemma:completeness_1} and is established using the same tools as above:
\begin{lemma}
Let $g \in \Lpspc{2}{}{{\calD} \times \mathfrak{I}}$ and for $\bfx \in \calD, m \in \gz$, let:
\begin{align}
\tilde{f}_m(\bfx) = \int_{\mathfrak{I}} g(\bfx,\eta)\,e^{-i2\pi m \eta}\,d\eta\,.
\label{eq:lemma_3_1}
 \end{align}
Furthermore, let the function ${f}$ be an extension of $\tilde{f}$ from the domain $\calD$ to the domain $\calC$ in the sense that for $\bfx\in\calD$,
\begin{align}
{f}(\Upsilon^m_\mathsf{h} \circ \bfx) :=   \tilde{f}_{m}(\bfx)\,.
\label{eq:lemma_3_2}
\end{align}
Then, ${f} \in \Lpspc{2}{}{\calC}$.
\begin{proof}
By Tonelli's theorem \citep{Folland_Real}, since $g \in \Lpspc{2}{}{{\calD} \times \mathfrak{I}}$, it holds that $g(\bfx,\cdot) \in \Lpspc{2}{}{\mathfrak{I}}$ for almost every $\bfx \in \calD$. Then, we may interpret eq.~\ref{eq:lemma_3_1} as a Fourier transform in $\eta$. By Parseval's identity  \citep{Folland_Real}, we have:
\begin{align}
\int_{\mathfrak{I}} \abs{g(\bfx, \eta)}^2\,d\eta = \sum_{m\in \gz} \abs{\tilde{f}_m(\bfx)}^2\,.
\label{eq:lemma_3_3}
\end{align}
Integrating both sides over $\bfx \in \calD$ and using $g \in \Lpspc{2}{}{{\calD} \times \mathfrak{I}}$, we get:
\begin{align}
\nonumber
\infty &> \int_{\bfx\in\calD}\!\bigg(\int_{\mathfrak{I}} \abs{g(\bfx, \eta)}^2\,d\eta\bigg)d\bfx = \int_{\bfx\in\calD}\!\bigg( \sum_{m\in \gz} \abs{\tilde{f}_m(\bfx)}^2 \bigg)d\bfx = \sum_{m\in \gz} \int_{\bfx\in\calD} \sum_{m\in \gz} \abs{\tilde{f}_m(\bfx)}^2 d\bfx\\
&=\sum_{m\in \gz} \int_{\bfx\in\calD} \abs{{f}(\Upsilon^m_\mathsf{h} \circ \bfx)}^2 d\bfx = \sum_{m\in \gz} \int_{\bfx\in \Upsilon^m_\mathsf{h}\circ \calD} \abs{{f}(\bfx)}^2 d\bfx =  \int_{\bfx \in \calC}\abs{f(\bfx)}^2 \,d\bfx\,.
\label{eq:lemma_3_4}
\end{align}
This shows that  ${f} \in \Lpspc{2}{}{\calC}$, as required. Note that the interchange of the summation and the integral in the calculations above can be justified using the Fubini-Tonelli Theorem \citep{Folland_Real}.
\end{proof}
\label{Lemma:completeness_2}
\end{lemma}
Lemma \ref{Lemma:completeness_1} establishes the existence of an operator $\calU: \Lpspc{2}{}{\calC} \to  \Lpspc{2}{}{{\calD} \times \mathfrak{I}}$ defined as:
\begin{align}
\label{eq:operatorU}
(\calU f)(\bfx,\eta)=\sum_{m\in \gz} {f}(\Upsilon^m_\mathsf{h} \circ \bfx)\,e^{i 2 \pi m\eta}\,,
\end{align}
while Lemma \ref{Lemma:completeness_2} establishes the existence of its inverse  $\calU^{-1}:\Lpspc{2}{}{{\calD} \times \mathfrak{I}} \to  \Lpspc{2}{}{\calC}$ defined as:
\begin{align}
\label{eq:operatorU_inverse}
(\calU^{-1}g)(\Upsilon^m_\mathsf{h} \circ \bfx)=\int_{\mathfrak{I}} g(\bfx,\eta)\,e^{-i2\pi m \eta}\,d\eta\,.
\end{align}
To verify that eq.~\ref{eq:operatorU_inverse} indeed defines the inverse of the operator in eq.~\ref{eq:operatorU}, we consider $f \in \Lpspc{2}{}{\calC}$ and $g \in  \Lpspc{2}{}{{\calD} \times \mathfrak{I}}$ such that $g = \calU f$, i.e., 
\begin{align}
\label{eq:operatorU_again_1}
g(\bfx,\eta)=\sum_{m\in \gz} {f}(\Upsilon^m_\mathsf{h} \circ \bfx)\,e^{i 2 \pi m\eta}\,.
\end{align}
We now multiply the above by $e^{-i 2 \pi m'\eta}$ for $m'\in\gz$ and integrate over $\eta$, to arrive at:
\begin{align}
\nonumber
\int_{\mathfrak{I}} g(\bfx,\eta)\,e^{-i2\pi m' \eta}\,d\eta &= \int_{\mathfrak{I}} \sum_{m\in \gz} {f}(\Upsilon^m_\mathsf{h} \circ \bfx)\,e^{i 2 \pi (m - m') \eta}\,d\eta \\
&=  \sum_{m\in \gz} \int_{\mathfrak{I}} {f}(\Upsilon^m_\mathsf{h} \circ \bfx)\,e^{i 2 \pi (m - m') \eta}\,d\eta = {f}(\Upsilon^{m'}_\mathsf{h} \circ \bfx)\,.
\label{eq:operatorU_again_2}
\end{align}
Thus  $f = \calU^{-1}g$ in accordance with eq.~\ref{eq:operatorU_inverse}. 

We also observe, based on the calculations in eq.~\ref{eq:lemma_3_4} that:
\begin{align}
\norm{f}{ \Lpspc{2}{}{\calC}} = \norm{\calU f}{\Lpspc{2}{}{{\calD} \times \mathfrak{I}}}\,,
\label{eq:operatorU_isometry}
\end{align}
and therefore, the operator $\calU$ is an isometric-isomorphism\footnote{Eq.~\ref{eq:operatorU_isometry} shows that $\calU$ is an isometry and Lemma \ref{Lemma:completeness_2} shows that it has a well defined inverse. Therefore, it is a unitary operator \citep{Folland_Real}.} between the spaces $\Lpspc{2}{}{\calC}$ and ${\Lpspc{2}{}{{\calD} \times \mathfrak{I}}}$. In analogy to the  Bloch-Floquet transform in the literature used for studying periodic problems \citep{Reed_Simon4, dorfler2011photonic}, we will refer to the operator $\calU$ as the \textit{helical Bloch-Floquet transform}\footnote{This operator is closely related to the so-called Zak transform \citep{justel2018zak, justel2014radiation} associated with the group.}\footnote{\label{footnote:U_translation} By use of the definition in eq.~\ref{eq:operatorU}, it is easy to observe that the operator behaves in the following manner with respect to the action of the group: 
\begin{align}
(\calU f)(\Upsilon^n_\mathsf{h} \circ \bfx,\eta) = e^{-i2\pi n\eta}(\calU f)(\bfx,\eta)\,,
\end{align} for any $n \in \gz$.} 
This operator allows us to demonstrate the completeness of the helical Bloch waves in  $\Lpspc{2}{}{\calC}$. As mentioned earlier, the basic idea behind  this proof is to map a given $f \in \Lpspc{2}{}{\calC}$ to its counterpart in ${\Lpspc{2}{}{{\calD} \times \mathfrak{I}}}$ and to then use the completeness of the set $\Psi_{\eta}$ for each $\eta \in \mathfrak{I}$.
\begin{theorem}[\textbf{Completeness theorem for helical Bloch waves}]
\label{Thm:Completeness_Helical_Bloch}
Let $f \in \Lpspc{2}{}{\calC}$, and for $\ell \in \nz$, $\bfx \in \calC$, let:
\begin{align}
\label{eq:f_l_expression}
f_{\ell}(\bfx) := \sum_{s=1}^{\ell} \int_{\mathfrak{I}}  \big\langle{(\calU f)(\cdot ;\eta), \psi_{s}(\cdot; \eta)}\big\rangle_{\Lpspc{2}{}{\calD}}\,\psi_{s}(\bfx, \eta)\,d\eta\,.
\end{align}
Then $f_{\ell} \to f$ in $\Lpspc{2}{}{\calC}$ as $\ell \to \infty$.
\end{theorem}
\begin{proof}
Since $\calU f \in {\Lpspc{2}{}{{\calD} \times \mathfrak{I}}}$, it follows from Fubini's theorem that $\calU f(\cdot; \eta) \in \Lpspc{2}{}{{\calD}}$ for almost every $\eta \in  \mathfrak{I}$. Therefore, it can be approximated using the functions in the set $\Psi_{\eta}$. Consequently, if we define:
\begin{align}
g_{\ell}(\bfx,\eta) := \sum_{s=1}^{\ell} \big\langle{(\calU f)(\cdot ;\eta), \psi_{s}(\cdot; \eta)}\big\rangle_{\Lpspc{2}{}{\calD}}\,\psi_{s}(\bfx, \eta)\,,
\label{eq:gl_def}
\end{align}
then $g_{\ell}(\cdot ;\eta) \to (\calU f)(\cdot ;\eta)$ in ${\Lpspc{2}{}{\calD}}$ as $\ell  \to \infty$ for almost every $\eta \in \mathfrak{I}$. In other words, the residual: 
\begin{align}
\label{eq:rl_def}
\scriptr_{\ell}(\eta) = \norm{(\calU f)(\cdot ;\eta) - g_{\ell}(\cdot ;\eta)}{\Lpspc{2}{}{\calD}}^2\,,
\end{align}
has the property that $\scriptr_{\ell}(\eta) \to 0$ for almost every $\eta \in \mathfrak{I}$, as $\ell \to \infty$. Furthermore using the identity $\norm{f_1 + f_2}{}^2 \leq 2(\norm{f_1}{}^2 + \norm{f_2}{}^2)$, as well as Bessel's inequality \citep{Folland_Real} on $g_{\ell} \to  \calU f$, we get: 
\begin{align}
\nonumber
\scriptr_{\ell}(\eta) &=  \norm{(\calU f)(\cdot ;\eta) - g_{\ell}(\cdot ;\eta)}{\Lpspc{2}{}{\calD}}^2
\leq 2\,\bigg(  \norm{(\calU f)(\cdot ;\eta)}{\Lpspc{2}{}{\calD}}^2 +  \norm{g_{\ell}(\cdot ;\eta)}{\Lpspc{2}{}{\calD}}^2 \bigg)\\
&\leq  2\,\bigg(  \norm{(\calU f)(\cdot ;\eta)}{\Lpspc{2}{}{\calD}}^2 +  \norm{(\calU f)(\cdot ;\eta)}{\Lpspc{2}{}{\calD}}^2 \bigg) = 4\,  \norm{(\calU f)(\cdot ;\eta)}{\Lpspc{2}{}{\calD}}^2\,.
\label{eq:rl_property}
\end{align}
However, $\norm{(\calU f)(\cdot ;\eta)}{\Lpspc{2}{}{\calD}}^2$ is in $\Lpspc{1}{}{\mathfrak{I}}$ based on the calculations in Lemma \ref{eq:lemma_2_1}. Therefore, by the Dominated Convergence Theorem \citep{Folland_Real}:
\begin{align}
\label{eq:rl_dominated_convergence}
\int_{\mathfrak{I}} \scriptr_{\ell}(\eta)\,d\eta \to 0\;\text{as}\; \ell \to \infty\,,
\end{align}
and consequently:
\begin{align}
\label{eq:gl_convergence}
\norm{\calU f - g_{\ell}}{\Lpspc{2}{}{{\calD} \times \mathfrak{I}}} \to 0\;\text{as}\; \ell \to \infty\,.
\end{align}
Since $\calU$ is an isometric isomorphism, this implies that $\calU^{-1}g_{\ell} \to f$ in $\Lpspc{2}{}{\calC}$ as $\ell \to \infty$. Now, using eq.~\ref{eq:operatorU_inverse}, we see that:
\begin{align}
\calU^{-1}g_{\ell}(\Upsilon^m_\mathsf{h}  \circ \bfx) = \int_{\mathfrak{I}}\sum_{s=1}^{\ell} \big\langle{(\calU f)(\cdot ;\eta), \psi_{s}(\cdot; \eta)}\big\rangle_{\Lpspc{2}{}{\calD}}\,\psi_{s}(\bfx, \eta)\,e^{-i2\pi m \eta}\,d\eta\,.
\label{eq:Uinv_gl}
\end{align}
On the other hand, evaluating eq.~\ref{eq:f_l_expression} at $\bfx = \Upsilon^{m}_\mathsf{h} \circ \bfy$, and using eq.~\ref{eq:Bloch_ansatz_5}, we see that:
\begin{align}
f_{\ell}(\Upsilon^{m}_\mathsf{h} \circ \bfy) &=  \sum_{s=1}^{\ell} \int_{\mathfrak{I}}  \big\langle{(\calU f)(\cdot ;\eta), \psi_{s}(\cdot; \eta)}\big\rangle_{\Lpspc{2}{}{\calD}}\,\psi_{s}(\Upsilon^{m}_\mathsf{h} \circ \bfy, \eta)\,d\eta\,\\
&=\sum_{s=1}^{\ell} \int_{\mathfrak{I}}  \big\langle{(\calU f)(\cdot ;\eta), \psi_{s}(\cdot; \eta)}\big\rangle_{\Lpspc{2}{}{\calD}}\,e^{-i2\pi m\eta}\,\psi_{s}(\bfy, \eta)\,d\eta\,.
\label{eq:gl_fl}
\end{align}
Comparing eqs.~\ref{eq:Uinv_gl} and \ref{eq:gl_fl}, it follows that $\calU^{-1}g_{\ell}  = f_{\ell}$ since $\bfy$ and $m$ are generic, and therefore, $f_{\ell} \to f$ in $\Lpspc{2}{}{\calC}$ when $\ell \to \infty$, as required.
\end{proof}
As mentioned earlier, the above results imply in essence that the spectral properties of $\hamil$ can be described completely in terms of helical bands and helical Bloch states (refer to \ref{appendix:direct_integrals} for further discussion along these lines).\footnote{An immediate consequence of the completeness theorem for Bloch states is that the spectrum of $\hamil$ is completely contained in the set of helical bands, i.e., more precisely, $\mathsf{spec.}(\hamil) \subseteq \mathsf{clos.}(\Lambda)$, with $\mathsf{clos.}(\cdot)$ denoting the (topological) closure. This is because, if $\kappa \in \rz$ is such that $\kappa \notin \mathsf{clos.}(\Lambda)$, then the action of $(\hamil - \kappa)^{-1}$ on $f \in \Lpspc{2}{}{\calC}$ can be computed formally using eq.~\ref{eq:f_l_expression} in Theorem \ref{Thm:Completeness_Helical_Bloch} as:
\begin{align}
\label{eq:formal_theorem_1}
(\hamil - \kappa)^{-1}f =  \sum_{s=1}^{\infty} \int_{\mathfrak{I}}  \big\langle{(\calU f)(\cdot ;\eta), \psi_{s}(\cdot; \eta)}\big\rangle_{\Lpspc{2}{}{\calD}}\,(\hamil - \kappa)^{-1}\psi_{s}(\cdot; \eta)\,d\eta\,.
\end{align}
Now, using $\hamil\,\psi_s(\cdot;\eta) = \lambda_s(\eta)\,\psi_s(\cdot;\eta)$, we have: 
\begin{align}
\label{eq:formal_theorem_2}
(\hamil - \kappa)^{-1}\psi_{s}(\cdot; \eta) = \frac{1}{\lambda_s(\eta) - \kappa}\psi_{s}(\cdot; \eta)\,,
\end{align}
so that:
\begin{align}
\label{eq:formal_theorem_3}
(\hamil - \kappa)^{-1}f =  \sum_{s=1}^{\infty} \int_{\mathfrak{I}}  \big\langle{(\calU f)(\cdot ;\eta), \psi_{s}(\cdot; \eta)}\big\rangle_{\Lpspc{2}{}{\calD}}\, \frac{1}{\lambda_s(\eta) - \kappa}\psi_{s}(\cdot; \eta)\,d\eta\,.
\end{align}
Since $\kappa \notin \mathsf{clos.}(\Lambda)$, the term $\frac{1}{\lambda_s(\eta) - \kappa}$ remains bounded even as $s \to \infty$. Therefore, the right-hand side of eq.~\ref{eq:formal_theorem_3} can be interpreted as a bounded operator on $f$, and so, $\kappa$ must belong to the resolvent set of $\hamil$. Conversely, based on the techniques presented in \citep{Odeh_Keller, dorfler2011photonic} it is also possible to directly demonstrate that $\Lambda \subseteq \mathsf{spec.}(\hamil)$, by constructing a suitable singular sequence of the form $\displaystyle u_l(\bfx) = \zeta \bigg(\frac{{\bfx}}{l}\bigg)\psi_j(\bfx,\eta)$, and using Weyl's criterion \citep{hislop2012introduction, Kato, rellich1969perturbation}. Here $\zeta(\cdot)$ is a carefully chosen smooth cutoff function.   Since, $\mathsf{spec.}(\hamil)$ is always a closed set \citep{Kato}, and by definition, $\mathsf{clos.}(\Lambda)$ is the smallest closed set containing $\Lambda$, it follows that $\mathsf{clos.}(\Lambda) = \mathsf{spec.}(\hamil)$.  Furthermore, if $\lambda \in \mathsf{spec.}(\hamil)$, it can be immediately seen to be part of the essential spectrum of $\hamil$. This is because if it were part of the point spectrum, then $\lambda$ would be associated with eigenfunctions of finite multiplicity. Due to the fact that $\hamil$ commutes with the operators in $\mathcal{T} = \big\{T_{\widetilde{\Upsilon}}: T_{\widetilde{\Upsilon}}f(\bfx) = f(\widetilde{\Upsilon}^{-1}\circ\bfx)\big\}_{\widetilde{\Upsilon} \in \calG_1}$, these eigenfunctions would be left invariant by the operators in $\mathcal{T}$ as well (also see footnote \ref{footnote:physics_proof}). However, this would contradict the requirement that these eigenfunctions belong to $\Lpspc{2}{}{\calC}$.
Note that these above results also follow from the direct integral decomposition of the Hamiltonian discussed in Section \ref{subsubsec:Governing_Equations} and \ref{appendix:direct_integrals}.} Additionally, since the behavior of any helical Bloch state over all of $\calC$ is completely specified based on its behavior over $\calD$ (once a value of $\eta \in \mathfrak{I}$ is chosen), the single electron problem posed on all of $\calC$ can be reduced to a set of problems (indexed by $\eta$) posed on the fundamental domain (illustrated in Figure \ref{fig:Helical_Bloch_Theorem}).  Consequently, by appropriate use of the helical Bloch states and the helical bands, quantities of interest in Kohn-Sham theory (which can be described using the solutions to the single electron problem), can be formulated entirely in terms of quantities specified on the fundamental domain. We now look at this procedure in more detail.
\begin{figure}[ht]
\centering
{\includegraphics[trim={8cm 5cm 2.5cm 4cm}, clip, width=0.9\textwidth]{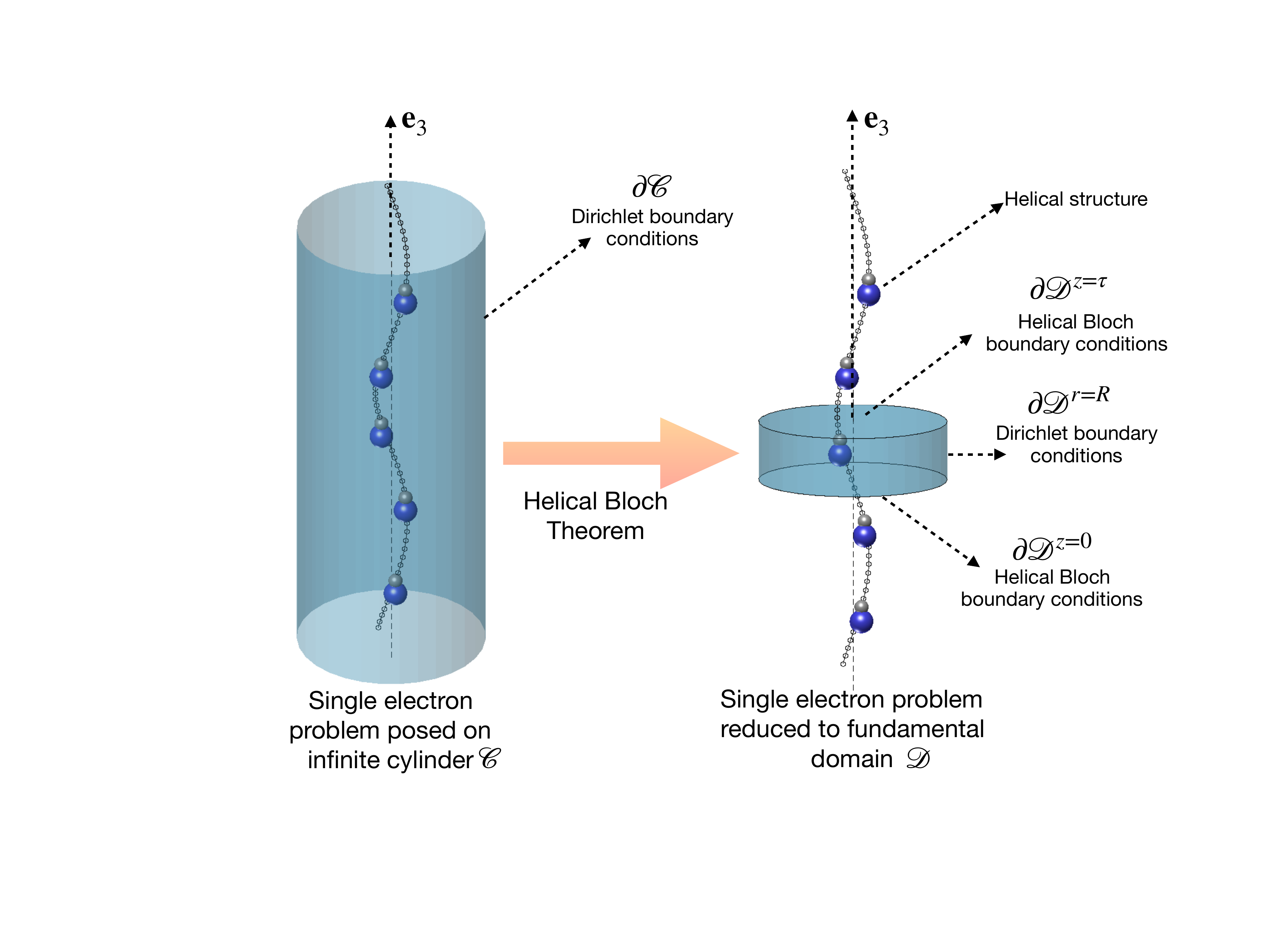}}
\caption{Illustration of the Helical Bloch Theorem (Theorem \ref{Thm:Existence_Helical_Bloch}) for the case when the single electron Hamiltonian obeys the symmetries of the helical group $\calG_1$.}
\label{fig:Helical_Bloch_Theorem}
\end{figure}
\subsubsection{Formulation of governing equations}
\label{subsubsec:Governing_Equations}
In what follows, we will consider the helical structure to be at finite electronic temperature $T_{\text{e}}$ Kelvin and we will ignore spin polarization effects. For the sake of clarity of presentation, we will itemize the formulation/derivation of the various terms and equations, as we go along.

\underline{Electron Density and Density Matrix:} A quantity of key importance in Kohn-Sham theory is the electron density. For a finite structure, such as a molecule or a cluster, this can be expressed in a straightforward manner in terms of the associated Kohn-Sham eigenstates and electronic occupations \citep{ghosh2017sparc_1, ghosh2019symmetry}. For a helical structure however,  care has to be taken to express this quantity due to the fact that there are effectively an infinite number of electrons associated with the structure. {In what follows, motivated by rigorous mathematical results related to the description of electronic states in crystalline systems \citep{catto2001thermodynamic, cances2020numerical, cances2008new}, we address this issue by
defining the single particle density operator \citep{goedecker1999linear, anantharaman2009existence, wang2016variational} in terms of the single electron Hamiltonian, and then expressing the electron density in terms of the diagonal of the density operator.}

{To clarify the above procedure, let us first consider a finite system (i.e., an isolated molecule or an atomic cluster, for example) in $\rz^3$, and let the single electron Hamiltonian, the single particle density operator (or density matrix), and the electron density for the system be denoted as $H$, $\mathfrak{D}$ and $\varrho(\bfx)$ respectively. Then, $H$ and $\mathfrak{D}$ are related as \citep{goedecker1999linear, anantharaman2009existence, wang2016variational}:}
{
\begin{align}
\label{eq:DM_Hamiltonian}
\mathfrak{D} =  f_{T_\text{e}}(H)\,,
\end{align}}{with $f_{T_e}(\cdot)$ denoting the Fermi-Dirac distribution function at electronic temperature $T_\text{e}$, i.e.:
\begin{align}
\label{eq:Fermi_Dirac}
f_{T_e}(y) = \frac{1}{1 + \exp\big({\frac{y\,-\,\lambda_{\text{F}}}{k_{\text{B}} T_{\text{e}}}}\big)}\,.
\end{align}
Here, $\lambda_{\text{F}}$ and $k_{\text{B}}$ denote the Fermi level and the Boltzmann constant respectively. Due to this definition, $\mathfrak{D}$ turns out to be a trace-class operator on $\Lpspc{2}{}{\rz^3}$, even though $H$ is (generically) an unbounded self-adjoint operator on the same space. Assuming $H$ has a pure point spectrum, denoting the eigenvalues of $H$ as $\omega_{j} \in \rz$, and the corresponding eigenvectors as $\upsilon_j \in \Lpspc{2}{}{\rz^3}$, we may express $H$ using its spectral representation as:}
{
\begin{align}
H = \sum_{j=1}^{\infty} \omega_{j}\,\upsilon_j \otimes \overline{\upsilon_j}\,.
\label{eq:finite_hamil_spectral}
\end{align}
}
{Using this form we see that the action of $H$ on any $f \in \text{Dom.}(H)$ is expressible as:
\begin{align}
Hf = \sum_{j=1}^{\infty} \omega_{j}\, \innprod{f}{\upsilon_j }{\Lpspc{2}{}{\rz^3}}\,\upsilon_j \,,
\end{align}}{and that $\mathfrak{D}$ can be expressed by means of spectral mapping \citep{Kato} as:}
{\begin{align}
\mathfrak{D} = \sum_{j=1}^{\infty} f_{T_e}(\omega_{j})\,\upsilon_j \otimes \overline{\upsilon_j}
\end{align}
}{Due to this definition, the action of $\mathfrak{D}$ on any $f \in \Lpspc{2}{}{\rz^3}$ can be expressed as:
\begin{align}
\mathfrak{D} f = \sum_{j=1}^{\infty} f_{T_e}(\omega_{j})\, \innprod{f}{\upsilon_j }{\Lpspc{2}{}{\rz^3}}\,\upsilon_j \,,
\end{align}
and it makes sense to write $\mathfrak{D}$ in coordinate form as:
\begin{align}
\mathfrak{D}(\bfx,\bfy) = \sum_{j=1}^{\infty} f_{T_e}(\omega_{j})\,\upsilon_j(\bfx) \otimes \overline{\upsilon_j(\bfy)}\,,
\end{align}
for $\bfx,\bfy \in \rz^3$. In this setting, the electron density $\varrho(\bfx)$ can be identified in terms of the diagonal of $\mathfrak{D}$, i.e.,
\begin{align}
\varrho(\bfx) = 2\,\mathfrak{D}(\bfx,\bfx)\,,
\label{eq:rho_DM_finite}
\end{align}
which leads to the well-known expression from Kohn-Sham theory \citep{KohnSham_DFT} (also see footnote \ref{footnote:diagonal_trace_class}):
\begin{align}
\varrho(\bfx) = 2\,\sum_{j=1}^{\infty} f_{T_e}(\omega_{j})\,\abs{\upsilon_j(\bfx)}^2\,.
\label{eq:electron_density_finite}
\end{align}
}
{Now, coming back to the case of the helical structure, we would analogously like to connect the single electron Hamiltonian $\hamil$, the single particle density operator ${\Gamma}$ and the electron density $\rho(\bfx)$. Accordingly, we define:
\begin{align}
{\Gamma} = f_{T_\text{e}}(\hamil)\,,
\label{eq:Gamma_def}
\end{align}
as an operator on $\Lpspc{2}{}{\calC}$. The issue however, is that $\hamil$ does not admit a representation similar to eq.~\ref{eq:finite_hamil_spectral}, and so the above definition does not immediately lead to transparent expressions for the electron density or the density operator in coordinate representation. To adress this, it is useful to first recast eq.~\ref{eq:Gamma_def} in terms of helical Bloch states.  The apparatus of \textit{direct integrals} \citep{Reed_Simon4, Garrett_Spectral_Theorem}, discussed in \ref{appendix:direct_integrals}, allows us to do this in a mathematically rigorous manner. The key result from the appendix is that the helical Bloch-Floquet transform allows the single electron Hamiltonian to be ``block-diagonalized'' into a set of problems associated with the helical Bloch states that are posed over the fundamental domain, i.e.:
\begin{align}
\label{eq:Bloch_Floquet_Block_Diagonal}
\calU\,\hamil\,\calU^{-1} =  \int_{\mathfrak{I}}^{\oplus} \hamil_{\eta}\,d\eta\,.
\end{align}
Here, as before, $\hamil$ represents the operator $-\half \Delta + V(\bfx)$ over the cylinder $\calC$ along with the boundary condition $\psi(\bfx) = 0$ for $\bfx \in \partial \calC$. The potential $V(\bfx)$ is group invariant (eq.~\ref{eq:potential_invariance}) and the unitary operator $\calU: \Lpspc{2}{}{{\calC}} \to \Lpspc{2}{}{{\calD} \times \mathfrak{I}}$ represents the helical Bloch-Floquet transform (eq.~\ref{eq:operatorU}). The operators $\{\hamil_{\eta}\}_{\eta \in \mathfrak{I}}$ represent the \textit{fibers} of $\hamil$ (in the sense of direct integrals) and are closely related to the operators $\mathfrak{h}_{\eta}^{\textsf{aux}}$ introduced in the proof\footnote{The key difference is that the operators $\mathfrak{h}_{\eta}^{\textsf{aux}}$ include $\eta$ dependence in the operators themselves and have group invariant solutions, whereas the operators $\hamil_{\eta}$ include  $\eta$ dependence in the boundary conditions and have helical Bloch solutions (i.e., solutions which are group invariant up to an $\eta$ dependent phase).} of Theorem \ref{Thm:Existence_Helical_Bloch} (eq.~\ref{eq:phi_equation_1}). Specifically, for each $\eta \in \mathfrak{I}$, the operator $\hamil_{\eta}$ represents the operator $-\half \Delta + V(\bfx)$ over the interior of the fundamental domain (i.e., the set $\calD$) along with the boundary conditions $\psi(\Upsilon_{\mathsf{h}} \circ \bfx)=e^{-i2\pi\eta}\,\psi(\bfx)$, $\bfR_{2\pi\alpha}^{-1}\nabla\psi(\Upsilon_{\mathsf{h}} \circ \bfx)=e^{-i2\pi\eta}\,\nabla\psi(\bfx)$ for $\bfx \in \partial\calD^{\,z=0}$ and, $\psi(\bfx) = 0$ for $\bfx \in \partial\calD^{\,r=R}$. The eigenstates of the operators $\{\hamil_{\eta}\}_{\eta \in \mathfrak{I}}$ are precisely the helical bands $\Lambda = \{\lambda_j(\eta): \eta \in \mathfrak{I}, j \in \nz\}$ and the helical Bloch states $\Psi = \{\psi_j(\cdot;\eta): \eta \in \mathfrak{I}, j \in \nz\}$ restricted to the region $\calD$.} 

{Above, eq.~\ref{eq:Bloch_Floquet_Block_Diagonal} expresses that $\hamil$ is unitarily equivalent to a ``block-diagonal'' operator whose ``blocks'' $\hamil_{\eta}$ are indexed by $\eta \in \mathfrak{I}$. Therefore, upon computing $\Gamma$ using eq.~\ref{eq:Gamma_def}, we can expect to obtain another operator which is unitarily equivalent to a block-diagonal operator with blocks $\big\{\Gamma_{\eta} =  f_{T_\text{e}}(\hamil_{\eta})\big\}_{\eta \in \mathfrak{I}}$. Since the function $f_{T_e}(\cdot)$ is analytic \citep{benzi2013decay}, these statements can be made mathematically precise by making use of the properties of the direct integral representation \citep[Theorem XIII.85]{Reed_Simon4}. Thus, we may write for the density matrix $\Gamma$ (as an operator on $\Lpspc{2}{}{{\calC}}$):
\begin{align}
\calU\,\Gamma\,\calU^{-1} =   \int_{\mathfrak{I}}^{\oplus} \Gamma_{\eta}\,d\eta\, = \int_{\mathfrak{I}}^{\oplus} f_{T_\text{e}}(\hamil_{\eta})\,d\eta\,.
\label{eq:Gamma_Direct_Integral}
\end{align}
Next, if we are able to express  $f_{T_\text{e}}(\hamil_{\eta})$ in a transparent form, we may be able to further simplify the expression for $\Gamma$. Accordingly, we write the operators $\hamil_{\eta}$ using spectral representation \citep{schmudgen2012unbounded, Teschl_QM, Kato} as:
\begin{align}
\hamil_{\eta} = \sum_{j=1}^{\infty} \lambda_j(\eta)\,\psi_j(\cdot; \eta) \otimes \overline{ \psi_j(\cdot; \eta)}\,,
\label{eq:spectral_hamil_eta}
\end{align}
and obtain:
\begin{align}
\label{eq:fibers_of_DM}
\Gamma_{\eta} =  f_{T_\text{e}}(\hamil_{\eta}) = \sum_{j=1}^{\infty}  f_{T_\text{e}}\big(\lambda_j(\eta)\big)\,\psi_j(\cdot; \eta) \otimes \overline{ \psi_j(\cdot; \eta)}\,.
\end{align}
Thus, as an operator on $\Lpspc{2}{}{{\calC}}$, the density matrix admits the representation:
\begin{align}
{\Gamma} = \calU^{-1}\,\bigg(\int_{\mathfrak{I}}^{\oplus}\!f_{T_{\text{e}}}(\hamil_{\eta})\,d\eta\bigg)\,\calU = \calU^{-1}\,\bigg( \int_{\mathfrak{I}}^{\oplus}\! \sum_{j=1}^{\infty} f_{T_{\text{e}}}\big(\lambda_j(\eta)\big)\,\psi_j(\cdot; \eta)\otimes\overline{ \psi_j(\cdot; \eta)}\,d\eta\bigg)\,\calU\,.
\label{eq:spectral_DM_calC}
\end{align}
While describing quantities over the fundamental domain, it is more appropriate and easier to deal with the density matrix as expressed as an operator on $\Lpspc{2}{}{{\calD} \times \mathfrak{I}}$. 
This can be written as $\tilde{\Gamma} = \calU\,\Gamma\,\calU^{-1}$ (i.e., the right hand side of eq.~\ref{eq:Gamma_Direct_Integral}), and it admits the following expression in coordinate representation (with $\bfx,\bfy \in \calD$):
\begin{align}
\tilde{\Gamma}(\bfx,\bfy) =  \int_{\mathfrak{I}}^{\oplus}\!\sum_{j=1}^{\infty}\,f_{T_{\text{e}}}\big(\lambda_j(\eta)\big)\,\psi_j(\bfx; \eta)\otimes\overline{\psi_j(\bfy; \eta)}\,d\eta\,.
\label{eq:spectral_DM_calD}
\end{align}
Now writing:
\begin{align}
\rho(\bfx) = 2\,\tilde{\Gamma}(\bfx,\bfx)\,,
\label{eq:rho_Gamma_def}
\end{align}
we see that the electron density can be expressed\footnote{\label{footnote:diagonal_trace_class}{As pointed out by an anonymous reviewer, eq.~\ref{eq:rho_Gamma_def} or (eq.~\ref{eq:rho_DM_finite} for the finite system case) can be somewhat subtle to interpret because the} {diagonal set $(\bfx, \bfx)$ is of measure zero in $\rz^3 \times \rz^3$. As far as we can tell, it is quite common in the electronic structure calculations literature to define the electron density as the diagonal of the density matrix (see e.g. equation 22 in \citep{wang2016variational}), though this particular issue is never really addressed. The validity of this definition actually follows from the properties of nuclear operators (see \citep{Math_overflow_Nuclear_operator_Diagonal, Math_overflow_trace_HS} and references therein) and therefore, the case of the density matrices discussed here is covered. Alternately, if one were to accept the expression for the electron density as given in eq.~\ref{eq:electron_density} (or eq.~\ref{eq:electron_density_finite} for the finite system case), then eq.~\ref{eq:rho_Gamma_def} (or correspondingly eq.~\ref{eq:rho_DM_finite} ) can be seen as meaningful.}} as (for $\bfx \in \calD$):
\begin{align}
\rho(\bfx) = 2\int_{\mathfrak{I}}\sum_{j=1}^{\infty}\,f_{T_e}\big(\lambda_j(\eta)\big)\,\lvert \psi_j(\bfx; \eta)\rvert^2\,d\eta\,.
\label{eq:electron_density}
\end{align}}
It is easy  to see from the above expression\footnote{We would like to thank Eric Cances (Ecole des Ponts ParisTech) and Carlos Garcia Cervera (Univ. of California, Santa Barbara) for email communication related to technicalities of the above derivation of eq.~\ref{eq:electron_density} and also for providing useful references.}  that the electron density is group invariant and also obeys a zero-Dirichlet boundary condition on the lateral surface of  $\calD$ (i.e., for $\bfx \in \partial\calD^{\,r=R}$). It is also apparent from the expression\footnote{{Note that the action of the operator $\tilde{\Gamma}$ on functions in $\Lpspc{2}{}{{\calD} \times \mathfrak{I}}$ follows from the definition of the direct integral, specifically, eq.~\ref{eq:operatorA} in \ref{appendix:direct_integrals}. Specifically, for $f(\bfx,\eta) \in \Lpspc{2}{}{{\calD} \times \mathfrak{I}}$:
\begin{align}
\big(\tilde{\Gamma}f\big)(\bfx,\eta) = \Gamma_{\eta} f(\bfx, \eta) = \sum_{j=1}^{\infty}\,f_{T_{\text{e}}}\big(\lambda_j(\eta)\big)\innprod{f(\cdot; \eta)}{\psi_j(\cdot; \eta)}{\Lpspc{2}{}{{\calD}}}\,\psi_j(\bfx,  \eta)\,.
\end{align}}} for the density matrix (eq.~\ref{eq:spectral_DM_calD}) that the following invariance relationship holds for any $m \in \gz$ and $\bfx,\bfy \in \calD$:
\begin{align}
\tilde{\Gamma}(\bfx,\bfy) = \tilde{\Gamma}(\Upsilon^m_{\mathsf{h}} \circ \bfx,\Upsilon^m_{\mathsf{h}} \circ \bfy)\,.
\label{eq:DM_invariance}
\end{align}
For notational simplicity, it is convenient to introduce the scalars $g_j(\eta) = f_{T_e}\big(\lambda_j(\eta)\big)$, i.e. the thermalized occupation numbers of the electronic states of the system, that appear in eqs.~\ref{eq:electron_density}, \ref{eq:spectral_DM_calD} and \ref{eq:spectral_DM_calC}. We will denote the collection of occupation numbers as $\mathfrak{G} = \big\{g_j(\eta) = f_{T_e}\big(\lambda_j(\eta)\big):\eta\in\mathfrak{I}, j\in \nz\}$.

The electron density for an extended system is expected to obey the constraint of having a fixed number of electrons per unit fundamental domain of the system, even though the electronic states themselves are delocalized over the entire structure \citep{Defranceschi_LeBris, LeBris_ReviewBook}. Denoting the number of electrons per unit cell as $N_{\text{e}}$, in our case, this leads to:
\begin{align}
\int_{\calD} \rho(\bfx)\,d\bfx = N_{\text{e}}\,,
\label{eq:rho_integral}
\end{align}
from which, using the orthonormality of the Bloch states over $\calD$, follows the constraint:
\begin{align}
2\int_{\mathfrak{I}}\sum_{j=1}^{\infty}\,f_{T_e}\big(\lambda_j(\eta)\big)\,d\eta =  2\int_{\mathfrak{I}}\sum_{j=1}^{\infty} g_j(\eta) = N_{\text{e}}\,.
\label{eq:rho_constraint}
\end{align}
In practice, the above equation can be used to compute the Fermi-level ($\lambda_{\text{F}}$) of the system. It also follows from this discussion that the density matrix $\Gamma$ on $\Lpspc{2}{}{\calC}$ is locally trace class.

\underline{Electronic Free Energy (per unit fundamental domain):} With the above expressions in place, we can now use an energy-minimization formalism to deduce the governing equations of Kohn-Sham theory for the helical structure. Since the structure is infinite, the quantity of primary importance in this regard is the \textit{electronic free energy per unit fundamental domain}, denoted here as $\calF(\Lambda, \Psi, \calP_{\calG_1}, \calD, \calG_1)$. This notation emphasizes the dependence of this quantity on the helical Bloch states, the helical Bloch bands, the positions of the simulated atoms $\calP_{\calG_1} = \big \{\bfx_k \big\}_{k=1}^{M_{\calG_1}}$ within the fundamental domain, (the interior of) the fundamental domain and the helical group itself. Following \citep{ghosh2019symmetry}, we express this quantity within the pseudopotential \citep{troullier1991efficient, chelikowsky2019introductory} and Local Density Approximations \citep{KohnSham_DFT} as:
\begin{align}
\nonumber
\calF(\Lambda, \Psi, \calP_{\calG_1}, \calD, \calG_1) = T_{\text{s}}&(\Lambda, \Psi,\calD, \calG_1) \,+\,E_{\text{xc}}(\rho, \calD)\,+\,K(\Lambda, \Psi, \calP_{\calG_1}, \calD, \calG_1)\\
&+ E_{\text{el}}(\rho, \calP_{\calG_1}, \calD, \calG_1)\,-\,T_{\text{e}}\,S(\Lambda)\,,
\label{eq:electronic_free_energy}
\end{align}
with the terms on the right-hand side representing the kinetic energy of the electrons, the exchange correlation energy, the nonlocal pseudopotential energy, the electrostatic energy and the electronic entropy contribution, respectively. We now discuss each of the terms in the above equation in detail.

\underline{Kinetic Energy Term:} The first term on the right-hand side of the above expression represents the kinetic energy of the electrons per unit fundamental domain. To motivate this term, we recall that for a finite system (in ${\rz^3}$) with single particle density matrix $\mathfrak{D}$, the kinetic energy can be expressed as \citep{catto2001thermodynamic, wang2016variational, anantharaman2009existence}:
\begin{align}
T_{\text{s}}^{\text{finite}} = 2\,\text{Tr.}\big[-\half\Delta\,\mathfrak{D}\big]\,,
\label{eq:KE_finite}
\end{align}
with $\text{Tr.}[\cdot]$ denoting the operator trace (of a trace-class operator on $\Lpspc{2}{}{\rz^3}$).  Analogously, it would make sense to consider the \textit{trace per unit fundamental domain} in case of the helical structure. As described in \ref{appendix:direct_integrals}, for an operator $\calA$ which is invariant under the group $\calG_1$ and which is locally trace-class, it is possible to assign meaning to the trace per unit fundamental domain by means of the direct integral decomposition. Specifically, if the helical Bloch-Floquet transform block-diagonalizes the operator into its fibers as:
\begin{align}
\calU\,\calA\,\calU^{-1} =  \int_{\mathfrak{I}}^{\oplus} \calA_{\eta}\,d\eta\,,
\label{eq:direct_integral_operatorA}
\end{align}
then the trace per unit cell (denoted $\underline{\text{Tr.}}[\cdot]$ henceforth) can be expressed as:
\begin{align}
\underline{\text{Tr.}}[\calA] = \int_{\mathfrak{I}}\!\text{Tr.}[\calA_{\eta}]\,d\eta\,,
\label{eq:direct_integral_trace_formula}
\end{align}
with the trace inside the integral signifying the usual operator trace\footnote{The operator trace for any trace-class operator $\mathfrak{O}$ on $\Lpspc{2}{}{\calD}$ can be computed as:
\begin{align}
\text{Tr.}[\mathfrak{O}] = \sum_{j=1}^{\infty}\innprod{\mathfrak{O}\,f_j}{f_j}{\Lpspc{2}{}{\calD}}\,,
\label{eq:trace_formula}
\end{align}
where $\{f_j\}_{j=1}^{\infty}$ can be any orthonormal basis of $\Lpspc{2}{}{\calD}$. Refer e.g.~to \citep{catto2001thermodynamic, anantharaman2009existence} for broader discussions of trace-class and locally trace-class operators in the context of electronic structure models.} in $\Lpspc{2}{}{\calD}$.
The expression for the kinetic energy per unit fundamental domain for the helical structure therefore boils down to:\footnote{Since the helical Bloch states $\psi_j(\cdot;\eta)$ belong to the domain of the operator $\hamil_{\eta}$, it follows that the traces in eq.~\ref{eq:KE_trace_fiber} are finite, making the expressions in that equation well defined. See \citep{catto2001thermodynamic, anantharaman2009existence} for further mathematical details along these lines.}
\begin{align}
T_{\text{s}}&(\Lambda, \Psi,\calD, \calG_1) = 2\,\underline{\text{Tr.}}\big[-\half\Delta\,{\Gamma}\big] =  2\int_{\mathfrak{I}}\!\text{Tr.}\bigg[\big(-\half\Delta\,{\Gamma}\big)_{\eta}\bigg]\,d\eta\,.
\label{eq:KE_trace_fiber}
\end{align}
We now write $\big(-\half\Delta{\Gamma}\big)_{\eta}$ as $(-\half\Delta)_{\eta}{\Gamma}_{\eta}$, observe that ${\Gamma}_{\eta}$ is already available from eq.~\ref{eq:fibers_of_DM}. Next, based on the discussion in \ref{appendix:direct_integrals}, we note that the fibers of the Laplacian on $\Lpspc{2}{}{\calC}$ are simply the Laplacian operators on $\Lpspc{2}{}{\calD}$ with ($\eta$ -dependent) helical Bloch boundary conditions. Since the states $\psi_j(\bfx; \eta)$ already satisfy these boundary conditions, and they form a basis of $\Lpspc{2}{}{\calD}$, it follows that:
\begin{align}
\nonumber
T_{\text{s}}(\Lambda, \Psi,\calD, \calG_1) &=  2\int_{\mathfrak{I}}\bigg[\sum_{j=1}^{\infty}f_{T_{\text{e}}}\big(\lambda_j(\eta)\big)\int_{\calD}-\half\Delta \psi_j(\bfx; \eta)\,\overline{\psi_j({\bfx}; \eta)}\,d\bfx \bigg]d\eta\\\nonumber
&=  \int_{\mathfrak{I}}\sum_{j=1}^{\infty}f_{T_{\text{e}}}\big(\lambda_j(\eta)\big)\big\langle{\Delta \psi_j(\cdot; \eta)},{\psi_j(\cdot; \eta)}\big\rangle_{\Lpspc{2}{}{\calD}}\,d\eta\\
&= \int_{\mathfrak{I}}\sum_{j=1}^{\infty}g_j(\eta)\big\langle{\Delta \psi_j(\cdot; \eta)},{\psi_j(\cdot; \eta)}\big\rangle_{\Lpspc{2}{}{\calD}}\,d\eta\,.
\label{eq:KE_Expression}
\end{align}

\underline{Exchange-Correlation Term:} The second term on the right-hand side of eq.~\ref{eq:electronic_free_energy} represents the exchange correlation energy of the electrons per unit fundamental domain. Within the Local Density Approximation (LDA) \citep{KohnSham_DFT}, it can be written as:
\begin{align}
E_{\text{xc}}(\rho, \calD) = \int_{\calD}\varepsilon_{\text{xc}}[\rho(\bfx)]\,\rho(\bfx)\,d\bfx\,.
\label{eq:Exc_Expression}
\end{align}
Note that it is also possible to modify this expression to use more sophisticated exchange correlation functionals such as the Generalized Gradient Approximation \citep{GGA_made_simple_Perdew} and this will have little bearing on our subsequent discussion.

\underline{Nonlocal Pseudopotential Energy Term:} The third term on the right-hand side of eq.~\ref{eq:electronic_free_energy} represents the energetic contribution from the nonlocal part of the pseudopotential and models the effect of electronic core states. For a finite system of $N_{\text{at}}$ atoms located at the points $\{\bfp_k \in \rz^3\}_{k=1}^{N_{\text{at}}}$, if the single particle electron density is denoted as $\mathfrak{D}$, then this term has the following form:
\begin{align}
K^{\text{finite}} = 2\,\text{Tr.}\big[\calV^{\text{nl}}\,\mathfrak{D}\big]\,.
\label{eq:Finite_Vnl_Expression}
\end{align}
The  operator $\calV^{\text{nl}}$ in the above equation is expressible in Kleinman-Bylander form \citep{kleinman1982efficacious} as :
\begin{align}
\calV^{\text{nl}} = \sum_{k=1}^{N_{\text{at}}}\sum_{p \in \calN_k}\gamma_{k,p}\,\chi_{k,p}(\cdot;\bfp_k)\,{\otimes}\,\overline{\chi_{k,p}(\cdot;\bfp_k)}\,.
\label{eq:Kleinman_Bylander_form}
\end{align}
Here, $\calN_k$ denotes the collection of projectors associated with the atom at $\bfp_k$, $\chi_{k,p}$ are the projection functions, and $\gamma_{k,p}$ are the corresponding normalization constants. The functions $\chi_{k,p}$ are themselves expressible in terms of atomic orbitals and are usually supported in a small region of space by design \citep{Troullier_Martins_pseudo}. To obtain the correct analog of this expression for the helical structure i.e., \textit{the nonlocal pseudopotential energy per unit fundamental domain}, it is useful to recall that this contribution to the energy is tied to the atoms in the fundamental domain as well as the electronic states in the system.  It is of a somewhat different nature as compared to the kinetic energy term for instance, in which case the electrons are the only contributing source. Since the electrons in the extended structure are delocalized, the trace per unit fundamental domain leads to the appropriate expression in that case. In case of the nonlocal pseudopotential energy term however, the contribution from the electrons is delocalized, while those from the atoms are not. With this in mind,\footnote{We would like to thank Phanish Suryanarayana, Georgia Institute of Technology, for discussions which helped clarify some of the properties of the nonlocal pseudoptential operator for the case of extended/condensed matter systems.} we now focus on the atoms in the fundamental domain, and denote the non-local pseudoptential operator associated with these atoms as:
\begin{align}
\calV^{\text{nl}}_{\calD} = \sum_{k=1}^{M_{\calG_1}}\sum_{p \in \calN_k}\gamma_{k,p}\,\chi_{k,p}(\cdot;\bfx_k)\,{\otimes}\,\overline{\chi_{k,p}(\cdot;\bfx_k)}\,.
\label{eq:VNL_fundamental_domain}
\end{align}
Then, in analogy with eq.~\ref{eq:Finite_Vnl_Expression}, the nonlocal pseudopotential energy per unit fundamental domain in case of the helical structure can be written by considering the action of $\calV^{\text{nl}}_{\calD} $ to the density matrix operator defined in eq.~\ref{eq:spectral_DM_calC}, i.e.:
\begin{align}
K = 2\,\text{Tr.}[\,\calV^{\text{nl}}_{\calD}\,{\Gamma}\,]\,.
\label{eq:K_helical_expression_1}
\end{align}
To simplify this expression\footnote{Note that ${\Gamma}$ is locally trace class, while $\calV^{\text{nl}}_{\calD}$ is a finite rank (and hence bounded) operator with a limited spatial extent. This makes eq.~\ref{eq:K_helical_expression_1} well defined.}, we employ eqs.~\ref{eq:Gamma_Direct_Integral}, \ref{eq:spectral_DM_calC}, the unitarity of the operator $\calU$, as well as the invariance of the trace under unitary transformations to obtain:
\begin{align}
K = 2\,\text{Tr.}[\,\calV^{\text{nl}}_{\calD}\,\calU^{-1}\,\tilde{\Gamma}\,\calU\,] =   2\,\text{Tr.}[\,\calU^{-1}\,\calU\,\calV^{\text{nl}}_{\calD}\,\calU^{-1}\,\tilde{\Gamma}\,\calU\,] = 2\,\text{Tr.}[\,\calU\,\calV^{\text{nl}}_{\calD}\,\calU^{-1}\,\tilde{\Gamma}\,]\,.
\label{eq:K_helical_expression_2}
\end{align}
{Next, we observe that the operator $\widehat{\calV}^{\text{nl}}_{{\calD}} = \calU\,\calV^{\text{nl}}_{\calD}\,\calU^{-1}$ acts on $\Lpspc{2}{}{\calD \times \mathfrak{I}}$, and it admits a direct integral representation (i.e., $\widehat{\calV}^{\text{nl}}_{{\calD}} =  \displaystyle \int_{\mathfrak{I}}^{\oplus}(\widehat{\calV}^{\text{nl}}_{{\calD}})_{\eta}\,d\eta$ ). The fibers of $\widehat{\calV}^{\text{nl}}_{{\calD}}$ can be written as : 
\begin{align}
\nonumber
(\widehat{\calV}^{\text{nl}}_{{\calD}})_{\eta} = (\calU\,\calV^{\text{nl}}_{\calD}\,\calU^{-1})_{\eta} &= \calU\,\bigg(\sum_{k=1}^{M_{\calG_1}}\sum_{p \in \calN_k}\gamma_{k,p}\,\chi_{k,p}(\cdot;\bfx_k)\,{\otimes}\,\overline{\chi_{k,p}(\cdot;\bfx_k)}\bigg)\,\calU^{-1}\\
&= \sum_{k=1}^{M_{\calG_1}}\sum_{p \in \calN_k}\gamma_{k,p}\,\,\calU\chi_{k,p}(\cdot;\eta;\bfx_k)\,{\otimes}\,\overline{\calU\chi_{k,p}(\cdot;\eta;\bfx_k)}\,.\label{eq:U_Vnl_Uinv}
\end{align}
In what follows, for the sake of brevity, we will denote the helical Bloch-Floquet transform of the projection functions, i.e., $\calU\chi_{k,p}(\cdot;\eta;\bfx_k)$ as $\hat{\chi}_{k,p}(\cdot;\eta;\bfx_k)$ and note that they can be represented\footnote{\label{footnote:chimod} 
In practice, since the projection functions often have small support (centered about atomic positions), it is possible to truncate the above summation to just a few terms. Under certain circumstances, a somewhat more computationally convenient form for $\hat{\chi}_{k,p}(\bfx;\eta;\bfx_k)$ may be obtained by making use of the specific form of the projection functions ${\chi}_{k;p}$. The functions ${\chi}_{k;p}$ are related to atomic orbitals and are therefore expressible as the product of a spherical harmonic with a compactly supported radially symmetric function. In the particular case that the projection functions arise from s-orbitals, it in fact follows that ${\chi}_{k;p}(\bfx; \bfx_k) =  {\chi}_{k;p}(\norm{\bfx - \bfx_k}{\rz^3})$. Then, we have:
\begin{align}
\nonumber
\hat{\chi}_{k,p}(\bfx;\eta;\bfx_k) &= \sum_{m\in \gz} {\chi}_{k;p}(\Upsilon^m_{\mathsf{h}} \circ \bfx; \bfx_k)\,e^{i 2 \pi m\eta} = \sum_{m\in \gz} {\chi}_{k;p}(\norm{ \Upsilon^m_{\mathsf{h}} \circ \bfx - \bfx_k}{\rz^3})\,e^{i 2 \pi m\eta}\\\nonumber
&= \sum_{m\in \gz} {\chi}_{k;p}\big(\norm{ \Upsilon^m_{\mathsf{h}} \circ \big(\bfx - \Upsilon^{-m}_{\mathsf{h}} \circ \bfx_k \big)}{\rz^3}\big)\,e^{i 2 \pi m\eta} \\
&= \sum_{m\in \gz} {\chi}_{k;p}(\norm{\bfx - \Upsilon^{-m}_{\mathsf{h}} \circ \bfx_k }{\rz^3})\,e^{i 2 \pi m\eta} = \sum_{m\in \gz} {\chi}_{k;p}(\norm{\bfx - \Upsilon^{m}_{\mathsf{h}} \circ \bfx_k }{\rz^3})\,e^{-i 2 \pi m\eta}\,.
\end{align}
Thus, under these circumstances, the group action in the formula for  $\hat{\chi}_{k,p}(\bfx;\eta;\bfx_k) $ has been shifted from $\bfx$, to $\bfx_k$, which is easier to deal with computationally.} via eq.~\ref{eq:operatorU} (for $\bfx \in \calD$ and $\eta \in \mathfrak{I}$ as):
\begin{align}
\label{eq:chi_chi_hat}
\hat{\chi}_{k,p}(\bfx;\eta;\bfx_k) = \calU\chi_{k,p}(\bfx;\eta;\bfx_k) = \sum_{m\in \gz} {\chi}_{k;p}(\Upsilon^m_{\mathsf{h}} \circ \bfx; \bfx_k)\,e^{i 2 \pi m\eta}\,.
\end{align}
With this notation, using properties tensor products, as well as eqs.~\ref{eq:chi_chi_hat} and \ref{eq:spectral_DM_calD}, we may rewrite eq.~\ref{eq:K_helical_expression_2} as:
\begin{align}
\nonumber
K(\Lambda, \Psi, & \calP_{\calG_1}, \calD, \calG_1) \\\nonumber 
&=  2\,\text{Tr.}[\,\calU\,\calV^{\text{nl}}_{\calD}\,\calU^{-1}\,\tilde{\Gamma}\,] = 2\,\text{Tr.}[\,\widehat{\calV}^{\text{nl}}_{{\calD}}\,\tilde{\Gamma}\,] \\\nonumber &= 2\,\text{Tr.}\Bigg[\,\bigg(\int_{\mathfrak{I}}^{\oplus}(\widehat{\calV}^{\text{nl}}_{{\calD}})_{\eta}\,d\eta\bigg)\,\bigg(  \int_{\mathfrak{I}}^{\oplus} \Gamma_{\eta}\,d\eta \bigg)\,\Bigg]  \\\nonumber 
&= 2\,\text{Tr.}\Bigg[ \bigg(\int_{\mathfrak{I}}^{\oplus}\sum_{k=1}^{M_{\calG_1}}\sum_{p \in \calN_k}\gamma_{k,p}\,\hat{\chi}_{k,p}(\cdot;\eta;\bfx_k)\,{\otimes}\,\overline{\hat{\chi}_{k,p}(\cdot;\eta;\bfx_k)}\,d\eta\bigg) \\\nonumber &\quad\quad\quad\quad\bigg(  \int_{\mathfrak{I}}^{\oplus}\!\sum_{j=1}^{\infty}\,f_{T_{\text{e}}}\big(\lambda_j(\eta)\big)\,\psi_j(\cdot; \eta)\,{\otimes}\,\overline{\psi_j(\cdot; \eta)}\,d\eta\,\bigg)\Bigg]\,.\\\nonumber
&= 2\,\text{Tr.}\Bigg[ \int_{\mathfrak{I}}^{\oplus} \bigg(\sum_{k=1}^{M_{\calG_1}}\sum_{p \in \calN_k}\gamma_{k,p}\sum_{j=1}^{\infty} f_{T_{\text{e}}}\big(\lambda_j(\eta)\big)  \big(\hat{\chi}_{k,p}(\cdot;\eta;\bfx_k) \psi_j(\cdot; \eta)\big) \\
&\quad\quad\quad\quad\quad\quad\quad\quad\quad\quad\quad\quad\quad\quad\quad\quad\quad\;\otimes \overline{\big(\hat{\chi}_{k,p}(\cdot;\eta;\bfx_k)\,\psi_j(\cdot; \eta)\big)}\bigg)\,d\eta\,\Bigg]
\label{eq:K_helical_expression_3}
\end{align}}
Now using the definition of the trace and that the {helical Bloch states are a basis\footnote{{Note that for any fixed $\eta \in \mathfrak{I}$, $\{\psi_j(\cdot, \eta)\}_{j=1}^{\infty}$ are a basis of $\Lpspc{2}{}{\calD}$. Additionally, as the proof of Theorem \ref{Thm:Completeness_Helical_Bloch} shows,  the entire set of helical Bloch states, i.e., $\Psi = \{\psi_j(\cdot;\eta): \eta \in \mathfrak{I}, j \in \nz\}$ forms a basis of $\Lpspc{2}{}{\calD \times \mathfrak{I}}$. Thus, the trace of an operator $\mathfrak{O}$ on $\Lpspc{2}{}{\calD \times \mathfrak{I}}$ may be computed as:
\begin{align}
\text{Tr.}[\mathfrak{O}]  = \int_{\mathfrak{I}} \sum_{j=1}^{\infty} \innprod{\mathfrak{O}\,\psi_j(\cdot;\eta)}{\psi_j(\cdot;\eta)}{\Lpspc{2}{}{\calD}}\,d\eta\,.
\end{align}}
{This is essentially the calculation described in eqs.~\ref{eq:K_helical_expression_3}, \ref{eq:K_helical_expression_4} above, with the operator $\mathfrak{O} = \mathcal{U} \mathcal{V}_{\mathcal{D}}^{n l} \mathcal{U}^{-1} \tilde{\Gamma}$ on $\Lpspc{2}{}{\calD \times \mathfrak{I}}$. Note that the above results also directly follow from the properties of direct integral decomposition discussed in \ref{appendix:direct_integrals}.}}}, this reduces to:
\begin{align}
\nonumber
&K(\Lambda, \Psi, \calP_{\calG_1}, \calD, \calG_1)\\\nonumber
&= 2\, \sum_{k=1}^{M_{\calG_1}}\sum_{p \in \calN_k}\gamma_{k,p}\int_{\mathfrak{I}}\bigg(\sum_{j=1}^{\infty}\,f_{T_{\text{e}}}\big(\lambda_j(\eta)\big)\big\langle{\hat{\chi}_{k,p}(\cdot;\eta;\bfx_k)},{\psi_j(\cdot; \eta)}\big\rangle_{\Lpspc{2}{}{\calD}}\\\nonumber
&\quad\quad\quad\quad\quad\quad\quad\quad\quad\quad\quad\quad\times\overline{\big\langle{\hat{\chi}_{k,p}(\cdot;\eta;\bfx_k)},{\psi_j(\cdot; \eta)}\big\rangle}_{\Lpspc{2}{}{\calD}}\bigg)\,d\eta\\\nonumber
&= 2\,\sum_{k=1}^{M_{\calG_1}}\sum_{p \in \calN_k}\sum_{j=1}^{\infty}\gamma_{k,p}\int_{\mathfrak{I}}\bigg( f_{T_{\text{e}}}\big(\lambda_j(\eta)\big)\bigg\lvert{\big\langle{\hat{\chi}_{k,p}(\cdot;\eta;\bfx_k)},{\psi_j(\cdot; \eta)}\big\rangle}_{\Lpspc{2}{}{\calD}}\bigg\rvert^2\bigg)\,d\eta\\
&= 2\,\sum_{k=1}^{M_{\calG_1}}\sum_{p \in \calN_k}\sum_{j=1}^{\infty}\gamma_{k,p}\int_{\mathfrak{I}}\bigg(g_j(\eta)\bigg\lvert{\big\langle{\hat{\chi}_{k,p}(\cdot;\eta;\bfx_k)},{\psi_j(\cdot; \eta)}\big\rangle}_{\Lpspc{2}{}{\calD}}\bigg\rvert^2\bigg)\,d\eta
\label{eq:K_helical_expression_4}
\end{align}

\underline{Electrostatic Energy Term:} We now discuss the contribution of the electrostatic interaction energy to the free energy per unit fundamental domain. This is the fourth term on the right-hand side of eq.~\ref{eq:electronic_free_energy}. Often, it is computationally advantageous to express this term using a so-called \textit{local formulation} \citep{Pask2005, suryanarayana2014augmented, motamarri2012higher, ghosh2019symmetry}, as we now do\footnote{The term \textit{local formulation} is associated with the fact that the electrostatic potential $\Phi$ can be solved through a Poisson equation, which avoids evaluation of the non-local integrals in eq.~\ref{eq:ES_phi_formula_finite}.}. For a finite system with atomic nuclei located at the points $\{\bfp_k \in \rz^3\}_{k=1}^{N_{\text{at}}}$ and electron density $\rho$, this term takes the form of the following optimization problem in the total electrostatic potential:
\begin{align} 
\label{eq:ES_finite} 
{E}_{\text{el}}^{\text{finite}} = \max_{\Phi}\bigg\{-\frac{1}{8\pi}\int_{\rz^3}\abs{\nabla \Phi}^2\,d\bfx + \int_{\rz^3}(\rho + b^{\text{finite}})\Phi\,d\bfx\bigg\} + E_{\text{sc}}^{\text{finite}}(\bfp_1,\bfp_2,\ldots,\bfp_k)\,.
\end{align}
Here $b^{\text{finite}}$ represents the total nuclear pseudocharge for the finite set of nuclei, and can be expressed in terms of the individual nuclear pseudocharges $\big\{b_k(\bfx;\bfp_k)\big\}_{k=1}^{{N_{\text{at}}}}$ as: 
\begin{align}
b^{\text{finite}}(\bfx) = \sum_{k=1}^{N_{\text{at}}}b_k(\bfx;\bfp_k)\,, 
\end{align} 
and the term $E_{\text{sc}}^{\text{finite}}(\bfp_1,\bfp_2,\ldots,\bfp_k)$ corrects for self-interactions and overlaps of the nuclear pseudocharges \citep{suryanarayana2014augmented}. Note that by design, the individual nuclear pseudocharges are usually smooth, radially symmetric functions centered at the nuclear positions, they have compact support, and they integrate to the (valence) nuclear charge of the nucleus in question. The electrostatic potential $\Phi$ that solves the maximization problem in eq.~\ref{eq:ES_finite} is the Newtonian potential associated with the net charge in the system:\footnote{Note that we have made a minor abuse of notation and used $\Phi$ to denote the ``trial''  electrostatic potentials involved in the maximization problems in eqs.~\ref{eq:ES_finite} and \ref{eq:ES_helical}, as well as the actual potentials that achieve these maxima (i.e., the \textit{arg max} of the functionals listed in eqs.~\ref{eq:ES_finite} and \ref{eq:ES_helical}.). The latter are expressible succinctly as the corresponding Newtonian potentials in eqs.~\ref{eq:ES_phi_formula_finite} and \ref{eq:ES_phi_formula_helical}.}
\begin{align}
\label{eq:ES_phi_formula_finite} 
\Phi(\bfx) = \int_{\rz^3} \frac{\rho(\bfy) + b^{\text{finite}}(\bfy)}{\norm{\bfx - \bfy}{\rz^3}}\,d\bfy
\end{align}

To extend the above formulation to a helical structure, we first write the total nuclear  pseudocharge at any point $\bfx \in \calC$ in terms of the pseudocharges of the atoms in the fundamental domain as:
\begin{align}
b(\bfx,\calP_{\calG_1},\calG_1) = \sum_{m \in \gz}\sum_{k = 1}^{M_{\calG_1}}b_{k}(\bfx; \Upsilon^m_{\mathsf{h}}\circ \bfx_k)\,,
\label{eq:Net_pseudocharge}
\end{align}
and observe that this quantity is group invariant owing to the aforementioned properties of the individual nuclear pseudocharges \citep{banerjee2016cyclic}. Since the electron density is group invariant as well, it follows that the net electrostatic potential $\Phi$ expressed as\footnote{Using a Fourier series expansion, it can be shown (see e.g. \citep{Defranceschi_LeBris} for similar arguments) that eq.~\ref{eq:ES_phi_formula_helical} is well defined whenever the system is charge neutral i.e., when $\displaystyle \int_{\calD} \rho(\bfx) +b(\bfx,\calP_{\calG_1},\calG_1)\,d\bfx = 0$.}:
\begin{align}
\Phi(\bfx) = \int_{\calC} \frac{\rho(\bfy) +b(\bfy,\calP_{\calG_1},\calG_1)}{\norm{\bfx - \bfy}{\rz^3}}\,d\bfy\,,
\label{eq:ES_phi_formula_helical} 
\end{align}
is also group invariant \citep{My_PhD_Thesis}. Therefore, we may use $\Phi, b$ and $\rho$ as defined over the fundamental domain, to define the analog of eq.~\ref{eq:ES_finite} for the helical structure as:
\begin{align}
\nonumber
E_{\text{el}}(\rho, \calP_{\calG_1}, \calD, \calG_1) =  \max_{\Phi}\bigg\{-&\frac{1}{8\pi}\int_{\calD}\abs{\nabla \Phi}^2\,d\bfx + \int_{\calD}\bigg(\rho(\bfx) + b(\bfx, ,\calP_{\calG_1},\calG_1)\bigg)\Phi(\bfx)\,d\bfx\bigg\}\\ 
&+ E_{\text{sc}}(\calP_{\calG_1},\calG_1, \calD)\,.
\label{eq:ES_helical} 
\end{align}
For the sake of brevity, we omit the details of the form of the corrections due to self interactions and overlaps of the nuclear pseudocharges, as reduced to the fundamental domain (i.e., the term $E_{\text{sc}}(\calP_{\calG_1},\calG_1, \calD)$ above) and instead point to \citep{banerjee2016cyclic, suryanarayana2014augmented} for  relevant information.

\underline{Electronic Entropy Term:} Finally, the last term on the right-hand side of \ref{eq:electronic_free_energy} represents the electronic entropy contribution to the free energy. Following \citep{motamarri2018configurational}, this term can be expressed for a finite system with density matrix $\mathfrak{D}$ as:
\begin{align}
S^{\text{finite}} = -2\,k_{\text{B}}\text{Tr.}[\,\mathfrak{D}\,\log(\mathfrak{D}) + (\calI - \mathfrak{D})\,\log(\calI - \mathfrak{D})]\,,
\end{align}
with $\calI$ denoting the identity operator. To obtain the analogous expression for the case of the helical structure, we work with the trace per unit fundamental domain instead. This gives us:
\begin{align}
\nonumber
S &= -2\,k_{\text{B}}\,\underline{\text{Tr.}}\bigg[\,{\Gamma}\,\log({\Gamma}) + (\calI - {\Gamma})\,\log(\calI - {\Gamma})\,\bigg]\\
&= -2\,k_{\text{B}}\int_{\mathfrak{I}}\text{Tr.}[{\Gamma}_{\eta}\,\log({\Gamma}_{\eta}) + (\calI - {\Gamma}_{\eta})\,\log(\calI - {\Gamma}_{\eta})]\,d\eta\,.
\label{eq:S_per_unit_cell}
\end{align}
By means of spectral mapping \citep{Kato}, the use of eq.~\ref{eq:fibers_of_DM}, and by noting again that the helical Bloch states $\psi_j(\bfx; \eta)$ are a basis of $\Lpspc{2}{}{\calD}$, we readily obtain:
\begin{align}\nonumber
&S(\Lambda) =\\\nonumber &-2\,k_{\text{B}}\int_{\mathfrak{I}}\bigg[\sum_{j=1}^{\infty}f_{T_{\text{e}}}\big(\lambda_j(\eta)\big)\,\log\big(f_{T_{\text{e}}}\big(\lambda_j(\eta)\big)\big) + \big(1 - f_{T_{\text{e}}}\big(\lambda_j(\eta)\big)\big)\,\log\big(1 - f_{T_{\text{e}}}\big(\lambda_j(\eta)\big)\big)\bigg]\,d\eta\\
&= -2\,k_{\text{B}}\int_{\mathfrak{I}}\bigg[\sum_{j=1}^{\infty}g_j(\eta)\,\log\big(g_j(\eta)\big) + \big(1 - g_j(\eta)\big)\,\log\big(1 - g_j(\eta)\big)\bigg]\,d\eta
\label{eq:S_expression}
\end{align}
With the above terms explicitly defined, we now turn to the variational problem for deducing the governing equations.

\underline{Variational Problem and Kohn-Sham Equations:} The variational problem for Kohn-Sham ground state of a given helical structure (i.e., the atomic coordinates in the fundamental domain, the fundamental domain  and the helical group are held fixed) consists of minimizing the electronic free energy $\calF(\Lambda, \Psi, \calP_{\calG_1}, \calD, \calG_1)$ with respect to the helical Bloch states and the helical bands, subject to the constraint in eq.~\ref{eq:rho_constraint}. In the literature, this minimization is often stated in terms of the helical Bloch states and the electronic occupation numbers instead \citep{ghosh2019symmetry, ghosh2017sparc_2}. Along those lines, we may define $\widetilde{\calF}(\mathfrak{G}, \Psi, \calP_{\calG_1}, \calD, \calG_1) = \calF(\Lambda, \Psi, \calP_{\calG_1}, \calD, \calG_1)$ to write the variational problem as:
\begin{align}
\widetilde{\calF}_{0}(\calP_{\calG_1}, \calD, \calG_1) = \text{inf.}_{{\Psi,\mathfrak{G}}}\,\widetilde{\calF}(\mathfrak{G}, \Psi, \calP_{\calG_1}, \calD, \calG_1) 
\end{align}
subject to: 
\begin{align}
2\int_{\mathfrak{I}}\sum_{j=1}^{\infty} g_j(\eta) = N_{\text{e}}\,,
\end{align}
and the requirement that the states in $\Psi$ be helical Bloch states. This requires that for any $\psi_j(\cdot;\eta) \in \Psi$ and $m\in \gz$, we have $\psi_j(\Upsilon^m_{\mathsf{h}} \circ \bfx; \eta ) =  e^{-i2\pi m\eta}\psi_j(\bfx; \eta)\,$, as well as the orthonormality condition between two helical Bloch states $\psi_i(\cdot;\eta), \psi_j(\cdot;\eta) \in \Psi$:
\begin{align}
\innprod{\psi_i(\cdot;\eta)}{\psi_j(\cdot;\eta)}{\Lpspc{2}{}{\calD}} = \delta_{i,j}\,.	
\end{align}
We take variations of the above constrained minimization problem, and obtain the Euler-Lagrange equations as the following helical symmetry adapted Kohn-Sham equations over the fundamental domain:
\begin{align}
\hamil^{\text{KS}}_{[\mathfrak{G}, \Psi, \calP_{\calG_1}, \calD, \calG_1]}\,\psi_j(\cdot;\eta) = \lambda_j(\eta)\,\psi_j(\cdot;\eta)\,\,\text{for}\, j\in\nz,\eta \in\mathfrak{I}\,.
\label{eq:Euler_Lagrange}
\end{align}
Here, the helical symmetry adapted Kohn-Sham Hamiltonian operator (with its dependence on the helical Bloch sates, the occupation numbers, etc., made explicit)\footnote{Up to a notational change, this operator is essentially the same as $\hamil^{\text{KS}}_{[\Lambda, \Psi, \calP_{\calG_1}, \calD, \calG_1]}$, when the dependence on helical bands (instead of the occupation numbers) is highlighted.} is:
\begin{align}
\label{eq:KS_Operator}
\hamil^{\text{KS}}_{[\mathfrak{G}, \Psi, \calP_{\calG_1}, \calD, \calG_1]} \equiv -\half \Delta + V_{\text{xc}} + \Phi + \widetilde{\calV}^{\text{nl}}_{\calD}\,,
\end{align}
in which $\displaystyle V_{\text{xc}} = \frac{\delta E_{\text{xc}}(\rho, \calD)}{\delta \rho}$ is the exchange correlation potential, $\Phi$ is the net electrostatic potential and satisfies the following symmetry adapted Poisson problem over the fundamental domain:
\begin{align}
\nonumber
-\Delta \Phi &= 4\pi\,\big(\rho + b(\cdot,\calP_{\calG_1},\calG_1)\big)\,,\\
\Phi(\Upsilon_\mathsf{h} \circ\bfx) &= \Phi(\bfx)\,,
\label{eq:Poisson_Equation}
\end{align}
and the operator $\widetilde{\calV}^{\text{nl}}_{\calD} = \calU\,{\calV}^{\text{nl}}_{\calD}\,\calU^{-1}$ is as defined\footnote{Due to the property $\hat{\chi}_{k,p}(\Upsilon^{n}_{\mathsf{h}}\circ\bfx;\eta;\bfx_k) =  e^{-i2\pi n\eta} \hat{\chi}_{k,p}(\bfx;\eta;\bfx_k)$ as observed in footnote \ref{footnote:U_translation}, it follows that $\widetilde{\calV}^{\text{nl}}_{\calD}$ commutes with the group action in an appropriate sense (when viewed as an operator on $\Lpspc{2}{}{{\calD} \times \mathfrak{I}}$). This implies that the Kohn-Sham operator $\hamil^{\text{KS}}$ commutes with the group action as well.} in eqs.~\ref{eq:U_Vnl_Uinv} and \ref{eq:chi_chi_hat}.

\underline{Harris-Foulkes Functional:} The above set of expressions represent a set of coupled nonlinear partial differential equations in the fields $\psi_j(\cdot;\eta)$ and the scalars $g_j(\eta)$. Once they have been solved self-consistently, the ground state electronic free energy $\widetilde{\calF}_{0}(\calP_{\calG_1}, \calD, \calG_1)$ per unit fundamental domain can be computed through eq.~\ref{eq:electronic_free_energy}. {In practical calculations, since self-consistency is never achieved perfectly, a better estimate of the ground state electronic free energy may be found using the so-called Harris-Foulkes functional \citep{harris1985simplified, foulkes1989tight}}. This can be written in helical symmetry-adapted form, using quantities expressed over the fundamental domain as:
\begin{align}
\nonumber
{\calF}^{\text{HF}}(\Lambda, \Psi,& \calP_{\calG_1}, \calD, \calG_1) = E_{\text{band}}(\Lambda) + E_{\text{xc}}(\rho, \calD) - \int_{\calD}V_{\text{xc}}(\rho(\bfx))\rho(\bfx)\,d\bfx\\
&+\half\int_{\calD}\bigg(b(\bfx, ,\calP_{\calG_1},\calG_1) - \rho(\bfx) \bigg)\Phi(\bfx)\,d\bfx + E_{\text{sc}}(\calP_{\calG_1},\calG_1,\calD) - T_{\text{e}}\,S(\Lambda)\,.
\label{eq:Harris_Foulkes}
\end{align}
All the quantities on the right-hand side of the above equation are easily interpreted based on earlier discussion, except the first one, i.e., $E_{\text{band}}(\Lambda)$, which represents the electronic band energy per unit fundamental domain. For a finite system with a single electron Hamiltonian $H$ and single particle density matrix $\mathfrak{D}$, this quantity is expressed as \citep{wang2016variational}: 
\begin{align}
E_{\text{band}}^{\text{finite}} = 2\,\text{Tr.}[H\mathfrak{D}]\,.
\end{align}
Analogously, for the helical structure, we use the trace per unit fundamental domain to write:
\begin{align}
E_{\text{band}} =2\underline{\text{Tr.}}[\,\hamil{\Gamma}\,] =2 \int_{\mathfrak{I}}\text{Tr.}[\,\hamil_{\eta}{\Gamma}_{\eta}\,]\,d\eta\,.
\end{align}
Using eqs.~\ref{eq:spectral_hamil_eta} and \ref{eq:fibers_of_DM} and using the completeness of the helical Bloch waves, we see that this is be expressible as:
\begin{align}
E_{\text{band}}(\Lambda) = 2\int_{\mathfrak{I}}\sum_{j=1}^{\infty}\lambda_j(\eta) f_{T_{\text{e}}}\big(\lambda_j(\eta)\big)\,d\eta = 2\int_{\mathfrak{I}}\sum_{j=1}^{\infty}\lambda_j(\eta)\,g_j(\eta)\,d\eta
\end{align}

\underline{Atomic Forces:} The Hellmann-Feynman forces on the atoms in the fundamental domain are (in Cartesian coordinates):
\begin{align}
\mathbf{f}_k = - \pd{\widetilde{\calF}_{0}(\calP_{\calG_1}, \calD, \calG_1)}{\bfx_k}\,.
\label{eq:force_expression_1}
\end{align}
By directly differentiating the various terms involved (eqs.~\ref{eq:K_helical_expression_3}, \ref{eq:ES_helical}) we arrive at the following expression for  $\mathbf{f}_k$ using quantities specified over the fundamental domain\footnote{Motivated by \citep{kikuji2005first, ghosh2019symmetry} we may use integration by parts to modify the last term on the right-hand side of eq.~\ref{eq:force_expression_2}, so that the derivatives of the projectors with respect to atomic coordinates can be eliminated in favor of  Cartesian gradients of the wavefunctions instead. This tends to improve the accuracy of the computed forces in practical calculations --  the orbitals are more smoothly varying than the projectors and therefore they tend to behave better upon taking derivatives. With this change as well as making use of the discussion in Footnote \ref{footnote:chimod}, it is possible to rewrite:
\begin{align}
\int_{\calD} \psi_j(\bfx;\eta)\,\overline{\pd{\hat{\chi}_{k,p}(\bfx;\eta;\bfx_k)}{\bfx_k}}\,d\bfx = \sum_{m\in \gz} (\bfR_{2\pi m \alpha})^{-1} \int_{\calD} \nabla \psi_{j}(\bfx;\eta)\,\overline{\hat{\chi}_{k,p}(\bfx;\eta;\Upsilon^{m}_{\mathsf{h}}\circ\bfx_k)}\,e^{i 2 \pi m\eta}\,d\bfx\,,
\end{align} 
under specific circumstances.}:
\begin{align}
\nonumber
\mathbf{f}_k &=\sum_{m \in \gz} (\bfR_{2\pi m \alpha})^{-1} \int_{\calD}\nabla b_{k}(\bfx; \Upsilon^m_{\mathsf{h}}\circ \bfx_k)\Phi(\bfx)\,d\bfx - \pd{E_{\text{sc}}(\calP_{\calG_1},\calG_1, \calD)}{\bfx_k}\\\nonumber
&-4\sum_{j=1}^{\infty}\Bigg(\int_{\mathfrak{I}}g_j(\eta)\sum_{p \in \mathcal{N}_k}\gamma_{k;p}\,\text{Re.}\Bigg\{\bigg[\int_{\calD} \hat{\chi}_{k,p}(\bfx;\eta;\bfx_k)\,\overline{\psi_j(\bfx;\eta)}\,d\bfx\bigg]\\
&\quad\quad\quad\quad\quad\quad\quad\quad\quad\quad\quad\quad\quad \times\bigg[\int_{\calD} \psi_j(\bfx;\eta)\,\overline{\pd{\hat{\chi}_{k,p}(\bfx;\eta;\bfx_k)}{\bfx_k}}\,d\bfx\bigg]
\Bigg\}\Bigg)\,d\eta
\label{eq:force_expression_2}
\end{align}
Here, $\text{Re.}$ denotes the real part of the quantity in braces.

This completes a discussion of the derivation of the various physically relevant terms, as well as the form of the equations of Kohn-Sham theory, as applied to a helical structure associated with a helical group generated by a single element. Comments on modifications of the above equations while dealing with a structure associated with a helical group generated two elements, and a presentation of the final expressions/equations for that case in appear in \ref{appendix:two_element_group_expressions}.
\section{Numerical Implementation}
\label{sec:Numerical_Implementation}
The Kohn-Sham equations for a helical structure (i.e., eq.~\ref{eq:Euler_Lagrange} or eq.~\ref{eq:Euler_Lagrange_G2}) are a set of non-linear eigenvalue problems indexed by $\eta$ (as well as $\nu$ in case of the group $\calG_2$) that are coupled to each other through the electron density $\rho$. The standard procedure for solving the equations of Kohn-Sham theory is through self-consistent field (SCF) iterations \citep{KohnSham_DFT}. This amounts to starting from a reasonable guess of the electron density in the fundamental domain (e.g. superpositions of individual atomic densities, as is used in our simulations) and an appropriate set of trial orthonormal wavefunctions (randomly chosen in our simulations), and then evaluating the eigenstates of the Kohn-Sham operator with these guesses. Thus, a set of linear eigenvalue problems (i.e., those associated with the linearized Kohn-Sham operator evaluated at the given electron density) indexed by $\eta$ (as well as $\nu$ in case of the group $\calG_2$) have to be solved. From this, the (trial) Fermi-level of the system and the (trial) occupation numbers maybe computed. The eigenfunctions and the occupation numbers may be then combined (in accordance with eq.~\ref{eq:electron_density} or eq.~\ref{eq:electron_density_G2}) to yield the trial electron density for the next step of the iterations. The above procedure has to be repeated till the difference in the electron density (or the effective potential) between successive iterations reaches below the desired convergence threshold. We will now discuss several features of this self-consistent solution process as implemented in the Helical DFT code.
\subsection{Discretization of reciprocal space}
\label{subsec:reciprocal_space_discretization}
Many quantities described in Section \ref{subsubsec:Governing_Equations} and \ref{appendix:two_element_group_expressions} involve integrals over $\eta \in \mathfrak{I} = [-\half,\half)$ (as well as normalized summations over $\nu \in \{0, {1},2, \ldots,{\mathfrak{N}-1}\}$ for the group $\calG_2$). To evaluate such integrals numerically, we employ quadrature based on the Monkhorst-Pack scheme \citep{monkhorst1976special}. Specifically, we sample the interval $\mathfrak{I}$ using a grid of $N_{\eta}$ points, and write:
\begin{align}
\int_{\mathfrak{I}} f(\eta)\,d\eta \approx \sum_{{b=1}}^{N_{\eta}}\,w_b\,f(\eta_b)\,.
\end{align}
Here, $w_b$ and $\eta_b$ denote the integration weights and integration nodes respectively. Summations over $\nu$ are left unchanged. The total number of points used for discretizing the reciprocal space (i.e., set $\mathfrak{B} = \mathfrak{I} \times \{0,1,2,\ldots,\mathfrak{N}-1\}$), therefore, is $N_{\calK} = N_{\eta} \times \mathfrak{N}$ (with $\mathfrak{N} = 1$ for the group $\calG_1$). Based on considerations of time-reversal symmetry (which apply as long as e.g.~magnetic fields are absent) \citep{geru2018time, ghosh2019symmetry}, it follows that for $\eta \in \mathfrak{I}$ and $\nu = 1,2,\mathfrak{N}-1$:
\begin{align}
\label{eq:time_rev_1}
\lambda_j(\eta,\nu) = \lambda_j(-\eta,\mathfrak{N}-\nu)\,,\,\psi_j(\bfx; \eta,\nu) = \overline{\psi_j(\bfx;-\eta,\mathfrak{N}-\nu)}\,,
\end{align}
while for $\nu = 0$:
\begin{align}
\label{eq:time_rev_2}
\lambda_j(\eta,0) = \lambda_j(-\eta,0)\,,\,\psi_j(\bfx; \eta,0) = \overline{\psi_j(\bfx;-\eta,0)}\,.
\end{align}
Effectively, the above considerations reduce the number of quadrature points over reciprocal space by half (i.e., $\displaystyle N_{\calK} \approx \frac{N_{\eta} \times \mathfrak{N}}{2}$).

With the above discretization choices, the self-consistent field iterations for the Kohn-Sham problem amount to solving a series of $N_{\calK}$ linear eigenvalue problems on every iteration step. Based on the mathematical treatment presented earlier as well as in \citep{banerjee2016cyclic}, it follows that  eigenvalue problems associated with distinct values of $\eta$ (i.e. $\eta_b$ in discretized form) and/or $\nu$ are disjoint from each other.  This implies that these (linear) eigenvalue problems can be solved independently of each other, in an embarrassingly parallel manner, regardless of how the discretization in real space is carried out. We make use of this feature of the equations to assign these distinct eigenvalue problems to different computational cores. This serves as a natural parallelization scheme and helps in drastically reducing the wall time associated with the most computationally intensive part of the SCF iterations.
\subsection{Truncation of infinite sums}
\label{subsec:Truncation}
Primarily, there are two distinct sources of infinite sums in the equations presented in Section \ref{subsubsec:Governing_Equations} and \ref{appendix:two_element_group_expressions}. The first arises due to summing over an infinite number of helical bands (e.g., eqs.~\ref{eq:electron_density} and \ref{eq:electron_density_G2}). To truncate such sums we assume that the electronic occupation numbers reduce to zero beyond the lowest $N_{\text{s}}$ bands and therefore, only $N_{\text{s}}$ eigenstates for each value of $\eta_b$ (and also each $\nu$ for $\calG_2$) need to be computed during the self-consistent field iterations. In effect, this is also an enforcement of the Aufbau principle for the system \citep{LeBris_ReviewBook, Defranceschi_LeBris}. Depending on the size of the discretized reciprocal space (i.e., the value of  the number $N_{\calK}$) we have found that including just a few extra bands beyond the minimum number required for holding the $N_{\text{e}}$ electrons per unit fundamental domain, suffices.\footnote{This is a well used approximation strategy in the literature (see e.g. \citep{ghosh2017sparc_1, ghosh2017sparc_2, Gavini_higher_order}). For finite systems at electronic temperatures that are less than a few thousand Kelvin, it suffices to choose the number of states to be equal to a multiple of half the number of electrons, with the multiplication factor being between $1.05$ and $1.10$ \citep{banerjee2018two}. For an extended system like a helical structure, often a just a few extra bands beyond half the number of electrons is sufficient since this actually amounts to these few extra states being available for every value of $\eta_b$ or $\nu$. As a result of this, tens or even hundreds of extra states (with occupation numbers approaching zero) get effectively included in the calculations.}

The second source of infinite sums arises from considering terms associated with group orbits (e.g., eqs.~\ref{eq:Net_pseudocharge} and \ref{eq:Net_pseudocharge_G2}), since helical groups by definition are infinite. However, these sums are also always associated with functions that are supported in a small ball around an atom of the structure (e.g. the nuclear pseudocharge in eq.~\ref{eq:Net_pseudocharge} and the nonlocal pseudopotential projection function in eq.~\ref{eq:chi_chi_hat}). Therefore, the influence of such sums on points in the fundamental domain is only dependent on terms of the summation that result in a nonzero overlap between the function support and the fundamental domain. This allows such infinite sums to be truncated as well.
\subsection{Helical coordinate system}
\label{subsec:helical_coordinates}
To carry out a discretization of the governing equations in a manner that is naturally adapted to the underlying helical symmetries of the structures being studied, it is useful to employ helical coordinates, as introduced in \citep{My_PhD_Thesis}. In order to have this coordinate system be commensurate with the helical groups $\calG_1$ or $\calG_2$, the coordinate transformation formulae are as follows. For $\bfx \in \calC$, if the Cartesian coordinates are $(x_1,x_2,x_3)$, then the corresponding helical coordinates $(r,\theta_1,\theta_2)$ are\footnote{For computational purposes, it is more apt to use $\arctantwo({x_2},{x_1})$, along with a suitable modification for the negative Y quadrants, instead of $\arctan{(\frac{x_2}{x_1})}$ in eq.~\ref{eq:helical_coordinates}, such that a value between $0$ and $2\pi$ is obtained.}:
\begin{align}
r=\sqrt{x_1^2+x_2^2}\,,\,\theta_1=\frac{x_3}{\tau}\,,\,
\theta_2=\frac{1}{2\pi}\arctan{\bigg(\frac{x_2}{x_1}\bigg)}-\alpha \frac{x_3}{\tau}\,.
\label{eq:helical_coordinates}
\end{align}
The coordinates $(r,\theta_1,\theta_2)$ are a natural generalization of cylindrical coordinates in the sense that they effectively reduce to cylindrical coordinates when the twist angle parameter $\alpha$ of the  system is set to zero. We may verify that these relations are onto and globally invertible on $\calC\backslash\{t\,\bfe_3:t \in \rz\}$. Furthermore, the inverse relations:
\begin{align}
(r,\theta_1,\theta_2) \mapsto (x_1,x_2,x_3)=\big(r\cos(2\pi(\alpha\theta_1+\theta_2)),
r\sin(2\pi(\alpha\theta_1+\theta_2)),\theta_1\tau\big)
\label{eq:inverse_helical_coordinates}
\end{align} 
map the open cuboid $(0,R)\times(0,1)\times(0,1)$ to the interior of the fundamental domain of $\calG_1$, (i.e., to the set $\calD$) and the open cuboid $(0,R)\times(0,1)\times(0,1/\mathfrak{N})$ to the interior of the fundamental domain of $\calG_2$ (i.e., to the set $\widetilde{\calD}$).
 
The action of a helical group can be easily computed in these coordinates as follows. Let $\bfx =(x_1,x_2,x_3) \in \calC$ have helical coordinates $(r,\theta_1,\theta_2)$. The action of the isometry $\Upsilon_\mathsf{h}$ that generates $\calG_1$ (and also $\calG_2$) is to map $\bfx$ to the point $\bfx'=\Upsilon_\mathsf{h} \circ \bfx = \bfR_{2\pi\alpha}\bfx+\tau\bfe_3$. Denoting the Cartesian and helical coordinates of this new point as $(x_1',x_2',x_3')$ and $(r',\theta_1',\theta_2')$, respectively, we see that $x_1' = x_1\cos(2\pi\alpha)-x_2\sin(2\pi\alpha)$, $x_2' = x_2\cos(2\pi\alpha)+x_1\sin(2\pi\alpha)$ and $x_3' = x_3+\tau$. Now, using eqs.~\ref{eq:inverse_helical_coordinates} and \ref{eq:helical_coordinates}, we get $r'=r,\theta_1'=\theta_1+1,\theta_2'=\theta_2$. By a similar calculation, we see that the action of the second generator  $\Upsilon_\mathsf{c}$ of the group $\calG_2$ (i.e., the pure rotation by $\displaystyle\frac{2\pi}{\mathfrak{N}}$ about axis $\bfe_3$), is to map the point with helical coordinates $(r,\theta_1,\theta_2)$ to the point $(r,\theta_1,\theta_2 + \frac{1}{\mathfrak{N}})$. These calculations imply in particular that if a function is invariant under the group $\calG_1$, then it is periodic in $\theta_1$, with period $1$, when expressed in helical coordinates. Similarly, invariance of a function under the group $\calG_2$, implies periodicity in $\theta_1$ (with period $1$), as well as periodicity in $\theta_2$, (with period $1/\mathfrak{N}$) when expressed in helical coordinates. These observations make it easy to enforce helical Bloch boundary conditions on the wavefunctions, as well as the group invariance of the electrostatic potential, in simulations. Derivation of the Cartesian gradient operator, the Laplacian operator and the volume integral in helical coordinates appears in \ref{appendix:grad_laplace_integral_helical}.
\subsection{Real space discretization: Finite difference scheme}
\label{subsec:finite_differences}
We employ a finite difference strategy for discretizing the governing equations in real space. We choose an annular cylindrical region $\Omega$ with axis along $\bfe_3$ as the simulation domain. This allows us to avoid the singularity associated with the helical coordinate system along the axis $\bfe_3$ and does not present any issues as long as the atoms of the structure are located well within  the annular region \citep{banerjee2016cyclic, ghosh2019symmetry}. This latter condition is well satisfied by the nanotube structures simulated in this work. The set $\Omega$ can also serve adequately as a suitable fundamental domain for either group $\calG_1$ or $\calG_2$ in simulations (i.e., it can replace $\calD$ or $\widetilde{\calD}$ in the formulae presented in Section \ref{subsubsec:Governing_Equations} and \ref{appendix:two_element_group_expressions}).  In cylindrical coordinates $(r,\vartheta,z)$ we have:
\begin{align}
\Omega = \big\{(r,\vartheta,z) \in \rz^3:R_{\text{in}} \leq r\leq R_{\text{out}},  \frac{2\pi\alpha z}{\tau}  \leq \vartheta \leq  \frac{2\pi\alpha z}{\tau} + \Theta, 0 \leq z \leq \tau\big\}\,.
\label{eq:Omega_domain}
\end{align}
The boundary of $\Omega$ can be expressed as:
\begin{align}
\partial\Omega = \partial R_{\text{in}} \bigcup \partial R_{\text{out}} \bigcup \partial\vartheta_{0} \bigcup \partial\vartheta_{\Theta}\bigcup\partial\calZ_0\bigcup\partial\calZ_{\tau}\,,
\label{eq:Omega_boundary}
\end{align}
with $\partial R_{\text{in}}$ and  $\partial R_{\text{out}}$ denoting the surfaces $r = R_{\text{in}}$ and $r = R_{\text{out}}$ respectively, $\partial\vartheta_{0}$ and $\partial\vartheta_{\Theta}$ denoting the surfaces $\vartheta =  \frac{2\pi\alpha z}{\tau} $ and $\vartheta =  \frac{2\pi\alpha z}{\tau} + \Theta$ respectively, and finally, $\partial\calZ_0$ and $\partial\calZ_{\tau}$ denoting the surfaces $z = 0$ and $z = \tau$ respectively. Figure \ref{fig:FD_Domain} illustrates the simulation cell as well as the boundaries of this domain for an untwisted helical structure.
\begin{figure}[ht]
\centering
{\includegraphics[trim={2cm 1cm 1cm 0cm}, clip, width=0.8\textwidth]{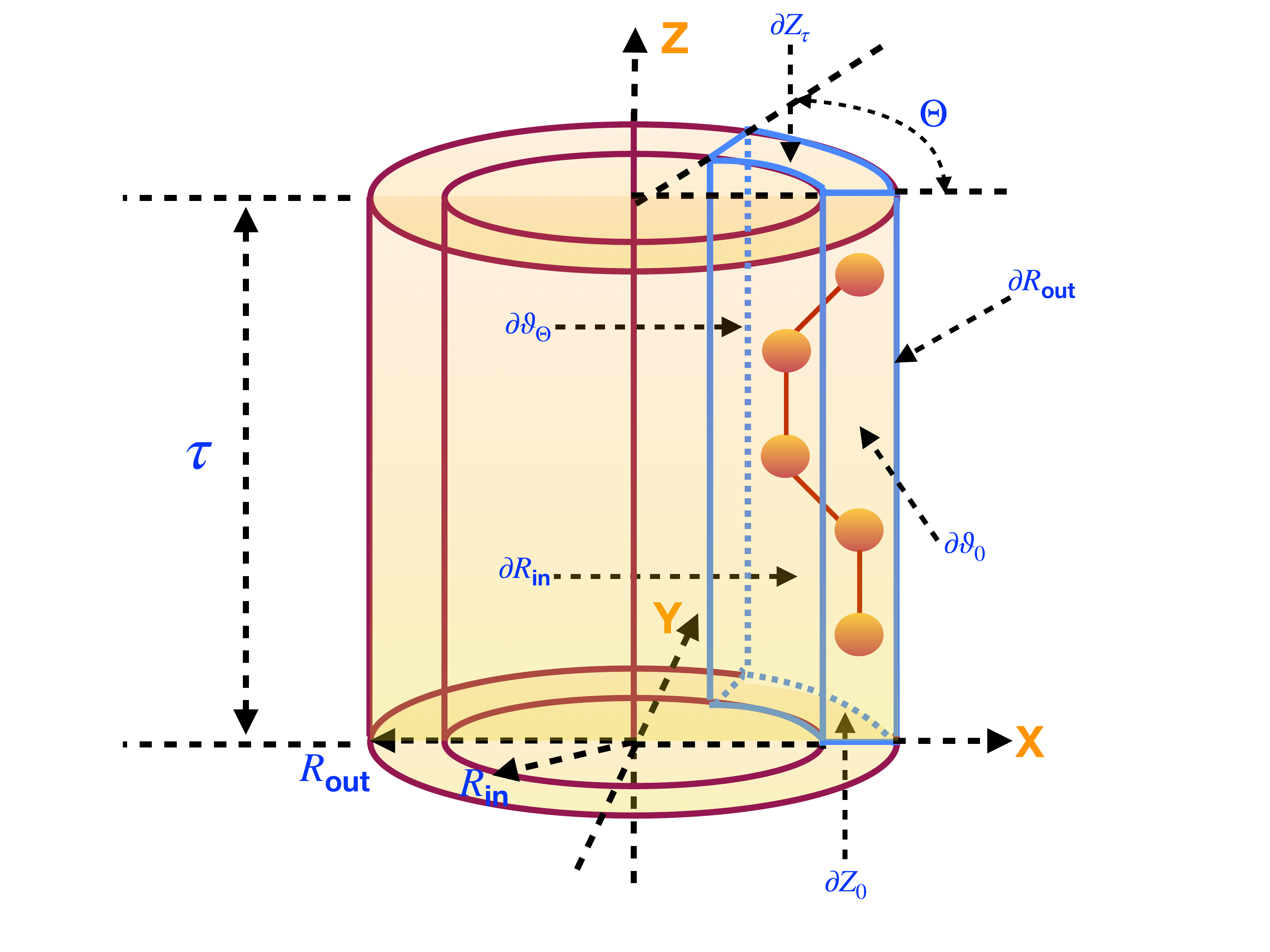}}
\caption{Illustration of the simulation domain $\Omega$ (domain boundary lines in blue) for an untwisted helical structure (i.e., $\alpha = 0$). A few atoms contained in the domain, as well as the various bounding surfaces of the domain are illustrated. For a structure associated with a two-generator helical group, the parameter $\Theta = 2\pi/\mathfrak{N}$ relates to the cyclic group order, while $\tau$ is related to the pitch of the screw transformation.}
\label{fig:FD_Domain}
\end{figure}

We set up a finite difference grid over $\Omega$ using helical coordinates, with spacing $h_{r}$, $h_{\theta_1}$ and $h_{\theta_2}$ along the $r$, $\theta_1$ and $\theta_2$ directions respectively. Since the simulation domain can be represented in helical coordinates with $r$ ranging from $R_{\text{in}}$ to $R_{\text{out}}$, $\theta_1$ ranging from $0$ to $1$, and $\theta_2$ ranging from $0$ to $\frac{1}{\mathfrak{N}}$, it follows that $R_{\text{out}} - R_{\text{in}} = \mathscr{N}_{r}h_{r}$, $1 =  \mathscr{N}_{\theta_1}h_{\theta_1}$, and $\frac{1}{\mathfrak{N}} =  \mathscr{N}_{\theta_2}h_{\theta_2}$, for natural numbers $\mathscr{N}_{r}$, $\mathscr{N}_{\theta_1}$ and $\mathscr{N}_{\theta_2}$. We index the finite difference nodes using a triplet of natural numbers $(\mathsf{i}, \mathsf{j}, \mathsf{k})$, for $\mathsf{i} = 1,2,\ldots,\mathscr{N}_{r}$,  $\mathsf{j} = 1,2,\ldots,\mathscr{N}_{\theta_1}$ and  $\mathsf{k} = 1,2,\ldots,\mathscr{N}_{\theta_2}$. We denote the value of a function $f$ at the grid point $(\mathsf{i}, \mathsf{j}, \mathsf{k})$ as $f^{(\mathsf{i}, \mathsf{j}, \mathsf{k})}$. We will denote $h = \text{Max.}\bigg(h_r, \tau h_{\theta_1}, \big(\frac{R_{\text{in}} + R_{\text{out}}}{2}\big)h_{\theta_2}\bigg)$ as the mesh spacing.

The formulae presented  in Section \ref{subsubsec:Governing_Equations} and \ref{appendix:two_element_group_expressions} require Cartesian gradients, the Laplacian operator and integrals over $\Omega$ to be evaluated using the finite difference scheme. The formulae for these quantities, as expressed in helical coordinates appears in \ref{appendix:grad_laplace_integral_helical}. With these in hand, we approximate the first order partial derivatives using central differences as:
\begin{align}
\frac{\partial f}{\partial r} \bigg|^{(\mathsf{i},\mathsf{j},\mathsf{k})} &\approx  \sum_{p=1}^{n_o} \bigg({w}_{p,r}^{\text{first}}  \big( f^{(\mathsf{i}+p,\mathsf{j},\mathsf{k})} - f^{(\mathsf{i}-p,\mathsf{j},\mathsf{k})}\big) \bigg) \,,  \nonumber \\
\frac{\partial f}{\partial \theta_1} \bigg|^{(\mathsf{i},\mathsf{j},\mathsf{k})} &\approx  \sum_{p=1}^{n_o} \bigg(  {w}_{p,\theta_1}^{\text{first}} \big( f^{(\mathsf{i},\mathsf{j}+p,\mathsf{k})} - f^{(\mathsf{i},\mathsf{j}-p,\mathsf{k})}\big) \bigg) \,, \nonumber \\
\frac{\partial f}{\partial \theta_2} \bigg|^{(\mathsf{i},\mathsf{j},\mathsf{k})} &\approx \sum_{p=1}^{n_o} \bigg({w}_{p,\theta_2}^{\text{first}} \big( f^{(\mathsf{i},\mathsf{j},\mathsf{k}+p)} - f^{(\mathsf{i},\mathsf{j},\mathsf{k}-p)}\big) \bigg) \,.
\label{eq:FD_gradient}
\end{align}
We approximate the pure (i.e., non-mixed) second order partial derivatives as:
\begin{align}
\frac{\partial^2 f}{\partial r^2} \bigg|^{(\mathsf{i},\mathsf{j},\mathsf{k})} &\approx  \sum_{p=0}^{n_o} \bigg({w}_{p,r}^{\text{second}}  \big( f^{(\mathsf{i}+p,\mathsf{j},\mathsf{k})} + f^{(\mathsf{i}-p,\mathsf{j},\mathsf{k})}\big) \bigg) \,,  \nonumber \\
\frac{\partial^2 f}{\partial \theta_1^2} \bigg|^{(\mathsf{i},\mathsf{j},\mathsf{k})} &\approx  \sum_{p=0}^{n_o} \bigg(  {w}_{p,\theta_1}^{\text{second}} \big( f^{(\mathsf{i},\mathsf{j}+p,\mathsf{k})} + f^{(\mathsf{i},\mathsf{j}-p,\mathsf{k})}\big) \bigg) \,, \nonumber \\
\frac{\partial^2 f}{\partial \theta_2^2} \bigg|^{(\mathsf{i},\mathsf{j},\mathsf{k})} &\approx \sum_{p=0}^{n_o} \bigg({w}_{p,\theta_2}^{\text{second}} \big( f^{(\mathsf{i},\mathsf{j},\mathsf{k}+p)} + f^{(\mathsf{i},\mathsf{j},\mathsf{k}-p)}\big) \bigg) \,,
\label{eq:FD_second_pure_derivatives}
\end{align}
while the mixed second order partial derivative $\displaystyle \frac{\partial^2 f}{\partial \theta_1\partial\theta_2}$ is obtained by applying the above first order derivative formula first in $\theta_1$ and then in $\theta_2$, i.e.,
\begin{align}
\nonumber
\frac{\partial^2 f}{\partial \theta_1\partial\theta_2} \bigg|^{(\mathsf{i},\mathsf{j},\mathsf{k})} &\approx \sum_{p=1}^{n_o} {w}_{p,\theta_2}^{\text{first}} \bigg[\bigg( \frac{\partial f}{\partial \theta_1}\bigg)^{(\mathsf{i},\mathsf{j},\mathsf{k}+p)} - \bigg(\frac{\partial f}{\partial \theta_1}\bigg)^{(\mathsf{i},\mathsf{j},\mathsf{k}-p)}\bigg]\,\\\nonumber
&\approx \sum_{p=1}^{n_o} {w}_{p,\theta_2}^{\text{first}} \bigg[ 
\bigg\{ \sum_{p'=1}^{n_o} {w}_{p',\theta_1}^{\text{first}} \big( f^{(\mathsf{i},\mathsf{j}+p',\mathsf{k}+p)} - f^{(\mathsf{i},\mathsf{j}-p',\mathsf{k}+p)}\big) \bigg\}\\
&\quad\quad\quad\quad\;\,- \bigg\{ \sum_{p'=1}^{n_o} {w}_{p',\theta_1}^{\text{first}} \big( f^{(\mathsf{i},\mathsf{j}+p',\mathsf{k}-p)} - f^{(\mathsf{i},\mathsf{j}-p',\mathsf{k}-p)}\big)\bigg\}\bigg]
\label{eq:FD_second_mixed_derivative}
\end{align}
In the above formulae, $n_o$ denotes half the finite difference order and is set to $6$ for all our simulations (i.e., $12^{\text{th}}$ order finite differences are used). This choice has also been employed elsewhere \citep{ghosh2017sparc_1, ghosh2017sparc_2, banerjee2016cyclic, ghosh2019symmetry} and is found to be adequate for attaining chemical accuracy. Letting $s$ denote $r$, $\theta_1$ or $\theta_2$, the weights that appear in the above formulae can be expressed as \citep{mazziotti1999spectral}:
\begin{align}
w_{0,s}^{\text{second}} & =  - \frac{1}{h_s^2} \sum_{q=1}^{n_o} \frac{1}{q^2} \,, \,\,
\nonumber \\\nonumber
w_{p,s}^{\text{second}}  & =  \frac{2 (-1)^{p+1}}{h_s^2\,p^2} \frac{(n_o!)^2}{(n_o-p)! (n_o+p)!} \,\,\text{for}\,\,p=1, 2, \ldots, n_o\,,\\
{w}_{p,s}^{\text{first}} & =  \frac{(-1)^{p+1}}{h_s\,p} \frac{(n_o!)^2}{(n_o-p)! (n_o+p)!}\,\,\text{for}\,\,p=1, 2, \ldots, n_o\,.
\label{eq:FD_weights}
\end{align}
We employ the following quadrature rule for approximating integrals over $\Omega$:
\begin{align}
\int_{\Omega}f(\bfx)\,d\bfx \approx  h_r h_{\theta_1} h_{\theta_2} \sum_{\mathsf{i}=1}^{\mathscr{N}_r} \sum_{\mathsf{j} =1}^{\mathscr{N}_{\theta_1}} \sum_{\mathsf{k}=1}^{\mathscr{N}_{\theta_2}} 2\pi\tau r_{\mathsf{i}}\,f^{(\mathsf{i}, \mathsf{j}, \mathsf{k})}\,,
\end{align}
with $r_{\mathsf{i}}$ denoting the radial coordinate of the finite difference node $(\mathsf{i}, \mathsf{j}, \mathsf{k})$.
\subsection{Boundary conditions}
\label{subsec:boundary_conditions}
We need to specify the boundary conditions on the various fields that are being solved for in the governing equations. These are the helical Bloch states $\psi_j(\bfx;\eta,\nu)$ (for $\eta \in \mathfrak{I}$ and $\nu = \{0,1,2,\ldots,\mathfrak{N}-1\}$) which satisfy the Kohn-Sham equations (eq.~\ref{eq:Euler_Lagrange} or eq.~\ref{eq:Euler_Lagrange_G2}), and the total electrostatic potential $\Phi$ which satisfies Poisson's equation (eq.~\ref{eq:Poisson_Equation} or eq.~\ref{eq:Poisson_Equation_G2}). In helical coordinates, we may interpret the helical Bloch boundary conditions (eq.~\ref{eq:helical_Bloch_two_generator}) as the following conditions:
\begin{align}
\psi_j(r,\theta_1 = 1,\theta_2;\eta,\nu) = e^{-i2\pi\eta}\psi_j(r,\theta_1 = 0,\theta_2;\eta,\nu)\,,
\label{eq:Bloch_BC_1}
\end{align}
which applies to the surfaces $\partial Z_0 \bigcup \partial Z_{\tau}$; as well as the condition:
\begin{align}
\psi_j(r,\theta_1,\theta_2 = \frac{1}{\mathfrak{N}};\eta,\nu) = e^{-i2\pi\frac{\nu}{\mathfrak{N}}}\psi_j(r,\theta_1,\theta_2 = 0;\eta,\nu)\,,
\label{eq:Bloch_BC_2}
\end{align}
which applies to the surfaces $\partial \vartheta_0 \bigcup \partial \vartheta_{\Theta}$.  We assume that the atoms within $\Omega$ are sufficiently far away from the boundary surfaces $\partial R_{\text{in}}$ and $\partial R_{\text{out}}$, so that the electron density decays to zero at these surfaces and zero Dirichlet boundary conditions on the wavefunctions (Section \ref{subsubsec:Helical_Bloch_Waves}) can be applied. Thus, the conditions to be applied on the surfaces $\partial R_{\text{in}} \bigcup \partial R_{\text{out}}$ are:
\begin{align}
\psi_j(r = R_{\text{in}},\theta_1,\theta_2 ;\eta,\nu) = \psi_j(r = R_{\text{out}},\theta_1,\theta_2;\eta,\nu) = 0\,.
\end{align}

As already discussed, the electrostatic potential $\Phi$ inherits the symmetry of the helical structure (i.e., it is invariant under $\calG_1$ or $\calG_2$). Thus, it follows that the boundary conditions on this quantity on the surfaces $\partial Z_0 \bigcup \partial Z_{\tau}$ and $\partial \vartheta_0 \bigcup \partial \vartheta_{\Theta}$ are respectively:
\begin{align}
\nonumber
\Phi(r,\theta_1 = 1,\theta_2) = \Phi(r,\theta_1 = 0,\theta_2)\,,\\
\Phi(r,\theta_1,\theta_2 =  \frac{1}{\mathfrak{N}}) = \Phi(r,\theta_1,\theta_2 = 0)\,.
\end{align}
To apply the right boundary conditions on $\Phi$  on the surfaces $\partial R_{\text{in}} \bigcup \partial R_{\text{out}}$, we may evaluate eq.~\ref{eq:ES_phi_formula_finite} or eq.~\ref{eq:ES_phi_formula_finite_G2} directly using Ewald summation \citep{Dumitrical_Ewald} or multipole expansion \citep{han2008real} techniques.\footnote{As remarked in \citep{ghosh2017sparc_1}, it is often adequate to set the net potential to zero on the bounding surfaces, instead of applying the above procedures, provided the boundaries are sufficiently distant. The net potential arises from a net neutral charge distribution, and therefore  it usually decays to zero faster. This is the procedure adopted here for the simulations presented in Section \ref{sec:Simulation_Results_Discussion}.} 
\subsection{Numerical linear algebra issues}
\label{subsec:numerical_linear_algebra}
On every SCF iteration step, the lowest $N_{\text{s}}$ eigenstates of a set of $N_{\mathcal{K}}$ Kohn-Sham operators (indexed by $\eta_b$  for the case of $\calG_1$, and $\eta_b,\nu$ for the case of $\calG_2$) have to be computed. We employ iterative diagonalization based on Chebyshev polynomial filtered subspace iterations (CheFSI) \citep{zhou2014chebyshev,Serial_Chebyshev,Parallel_Chebyshev, My_DG_Cheby_paper} for this purpose. Due to the fact that the helical coordinate system is curvilinear, the finite difference discretization of the Laplacian operator results in a discretized operator that is non-Hermitian (even though the Laplacian as an operator on an infinite dimensional Hilbert space \textit{is} Hermitian). This results in the discretized Hamiltonian operator also being non-Hermitian, which can lead to non-real eigenvalues of the matrix. However, as the discretization is made finer (i.e., $h_{r},h_{\theta_1},h_{\theta_2}$ become smaller) the discretized Laplacian and Hamiltonian matrices approach Hermitian matrices and so, the eigenvalues resulting from these discretized operators tend to have vanishingly small imaginary parts \citep{gygi1995real, banerjee2016cyclic}. In practice, for the mesh spacings that have been considered in this work, the imaginary parts are small enough that they can be safely ignored without affecting the quality of the simulations (also see \citep{banerjee2016cyclic} and \citep{ghosh2019symmetry} where a similar situation was encountered).

We use Chebyshev polynomial filter orders in the range $55$ to $80$ for our simulations. This is somewhat higher than what is commonly employed in finite difference DFT calculations in affine coordinate systems with comparable mesh spacing \citep{Serial_Chebyshev, ghosh2017sparc_1, ghosh2017sparc_2}.  We adopt it here to mitigate the effect of the larger spectral width of the discretized Hamiltonian that arises due to crowding of  grid points as one approaches the origin in helical coordinates. We have also observed that the time for computing matrix-vector products using the discretized Hamiltonian in helical coordinates is larger compared to the time required for the case of a finite difference Hamitonian (of the same size) arising from a cylindrical system (obtained by setting the twist angle parameter $\alpha$ to $0$ in the helical case). This is certainly due to the presence of cross derivatives in the helical case (eq.~\ref{eq:FD_second_mixed_derivative}), which makes the discretized Hamiltonian somewhat less sparse when compared with the cylindrical one. Within the CheFSI method, we use Arnoldi iterations for computing the extremal eigenvalues of the Hamiltonian and a direct diagonalization method for computing the projected subspace problem. 

We use the Generalized Minimal Residual method (GMRES) \citep{saad1986gmres} to solve the Poisson problem associated with the total electrostatic potential $\Phi$. To accelerate convergence, we use an incomplete LU factorization based preconditioner \citep{saad2003iterative}.
\subsection{Matlab implementation: The Helical DFT code}
\label{subsec:Matlab}
We have implemented the above computational strategies using the MATLAB \citep{MATLAB:2019} software package  into a code called Helical DFT.\footnote{An early version of this MATLAB code was developed in collaboration with Phanish Suryanaryana, Georgia Institute of Technology. Details on a more efficient {C/C++} implementation by Suryanarayana and collaborators appears in their forthcoming work.} The code parallelizes computation over the different $\eta$ and $\nu$ values using MATLAB's Parallel Computing Toolbox (the \textsf{parfor} function). For maximum efficiency of the MATLAB implementation, code vectorization has been used as much as possible. However, for computing certain quantities (such as the atomic forces and the net nuclear pseudocharge), multiple levels of nested loops were found unavoidable. These routines were converted into machine code by use of the MATLAB Coder framework, which helped alleviate performance issues. In order to reduce the memory footprint associated with the storage of the different Hamiltonian matrices arising from the $N_{\calK}$ different values of $\eta_b$ (and also $\nu$ for $\calG_2$), we avoid computing matrix vector products through MATLAB's internal sparse matrix framework since that requires these matrices to be available explicitly. Instead, we store only the Laplacian part of the Hamiltonian matrix in compressed sparse row (CSR) format and apply the helical Bloch boundary conditions associated with the different values of $\eta_b$ and/or $\nu$ on the fly, while computing matrix vector products. The actual task of computing these matrix vector products is carried out using a {C} language routine which has been compiled and interfaced with our MATLAB code.

Some other relevant details related to the implementation are as follows. We use the periodic variant of Pulay's scheme \citep{banerjee2016periodic, pulay_mixing} in the total potential to accelerate the convergence of the SCF iterations. The Fermi energy is calculated using a nonlinear equation root finder (MATLAB's \textsf{fzero} function). When required, structural relaxation is achieved by an implementation of the Fast Intertial Relaxation Engine (FIRE) algorithm \citep{bitzek2006structural}. 
\section{Simulation Results and Discussion}
\label{sec:Simulation_Results_Discussion}
We now turn to a discussion of numerical simulations and results. All simulations were run using a single node of the Mesabi cluster at the Minnesota Supercomputing Institute, or a single node of the Hoffman2 cluster at UCLA's Institute for Digital Research and Education. Each compute node of Mesabi has $24$ Intel Haswell E5-2680v3 processors operating at $2.50$ GHz, and $64$ GB to $1$ TB of RAM. Each compute node of the Hoffman2 cluster has two $18$-core Intel Xeon Gold 6140 processors (with $24.75$ MB cache, and running at $2.3$ GHz), and $192$ GB of RAM.

All calculations presented here use Troullier-Martins norm conserving pseudopotentials \citep{troullier1991efficient}. The Local Density Approximation \citep{KohnSham_DFT} was used for modeling the exchange correlation energy, and the Perdew-Wang parametrization \citep{Perdew_Wang} of the correlation energy was employed. An electronic temperature of $T_{\text{e}} = 315.77$ Kelvin was used for Fermi-Dirac smearing to help accelerate SCF convergence.

The large majority of the simulations here have focused on the study of single wall nanotubes. Starting from the sheet of an elemental two-dimensional material, nanotubes of any chirality can be formed using the so-called ``roll-up'' construction \citep{evarestov2015theoretical}. This procedure                    also allows us to see \citep{James_OS, Dumitrica_James_OMD} that such nanotubes can be adequately represented using helical groups generated by two elements with just $4$ atoms in the fundamental domain \citep{ghosh2019symmetry}. In this representation, the twist angle parameter $\alpha$ becomes related to the chirality of the tubes ($\alpha = 0$ for achiral tubes), while the parameter $\mathfrak{N}$, associated with the cyclic group order, is related to the tube radius  \citep{Dumitrica_James_OMD}. In our nanotube simulations, we have ensured that the atoms within the fundamental domain are always located $10$ to $12$ Bohrs away from the boundary surfaces $\partial R_{\text{in}}$ and  $\partial R_{\text{out}}$ so as to allow sufficient decay of the electron density and the wavefunctions in the radial direction. 
\subsection{Materials system: Single layer black phosphorus nanotubes}
\label{subsec:materials_system}
Single-layer black phosphorus, or phosporene, is a two-dimensional nanomaterial that has been the object of intense investigation in recent years due its association with a number of unusual and fascinating material properties \citep{reich2014phosphorene, das2014tunable, carvalho2016phosphorene, liu2014phosphorene, kou2015phosphorene, rodin2014strain, guan2014phase}. Nanotubes of this material, as formed by the roll-up construction have also received attention in the literature \citep{guo2014phosphorene, nguyen2018atomistic, cao2017lithium, sorkin2016mechanical, ansari2017density, zhang2017strain,pan2017self, allec2016inconsistencies, li2018tunable, liao2016effects, guan2014high,liu2018strength,  fernandez2019structural, sorkin2017recent}, due to their interesting optical and electronic properties, and the coupling of these properties to mechanical strains. This motivates our choice in  selecting this material for the simulations presented in this work.

As a starting point, we obtained the ground state structure of a single layer of phosphorene (Figure \ref{fig:phosphorene_sheet}) using the same pseudopotential, exchange correlation functional and electronic temperature as the Helical DFT simulations subsequently described. We used the plane-wave DFT code ABINIT \citep{Gonze_ABINIT_1, gonze2016recent} to perform the geometry relaxation calculation. The periodic unit cell for this simulation contained $4$ atoms. An energy cutoff of $40$ Ha, along with $30 \times 30 \times 1$ k-points and a cell vacuum of $25$ Bohr in the Z-direction was used.  At the end of this very refined calculation the atomic forces were all less than $10^{-5}$ Ha/Bohr, while the cell stresses were of the order of $10^{-8}\;\text{Ha/Bohr}^3$ or lower. Some of the structural parameters obtained from the calculation are shown in Table \ref{Table:Phosphorene_Parameters}. There appears to be generally good agreement with the literature\footnote{The minor differences are possibly due to our use of LDA exchange correlation, which tends to predict overbinding and shortened bond lengths \citep{van1999correcting}. } thus giving us confidence in the reliability of the subsequent simulations with regard to materials physics.
\begin{table}[ht]
\centering
\scalebox{0.8}{
\begin{tabular}{|c|c|c|c|c|}
\hline
Lattice constant  $a_1$ & Lattice constant $a_2$ & Bond length $\delta_1$ & Bond length $\delta_2$ & Bond angle $\gamma$\\
along X axis (\AA) & along Y axis  (\AA) &   (\AA)  & (\AA) & (degrees)\\\hline
$3.26$ & $4.36$ & $2.20$  & $2.193$   & $102.42$ \\
($3.35$ \citep{liu2014phosphorene}, $3.3$ \citep{behera2018paw}) &  ($4.62$ \citep{liu2014phosphorene}, $4.5$ \citep{behera2018paw}) & ($2.29$ \citep{behera2018paw}) & ($2.25$ \citep{behera2018paw}) & ($103.3$ \citep{behera2018paw}) \\ \hline
\end{tabular} 
}
\caption{Structural parameters of phosphorene computed using a periodic DFT calculation. Quantities in parentheses refer to values in the literature along with references. Explanation of the parameters is available from Figure \ref{fig:phosphorene_sheet}}.
\label{Table:Phosphorene_Parameters}
\end{table}
\begin{figure}[ht]
\centering
\subfloat[Top view of a phosphorene sheet. Atoms in the shaded region are placed in the Helical DFT simulation cell (fundamental domain).]{\includegraphics[trim={12cm 10cm 7cm 4.25cm}, clip, width=0.45\textwidth]{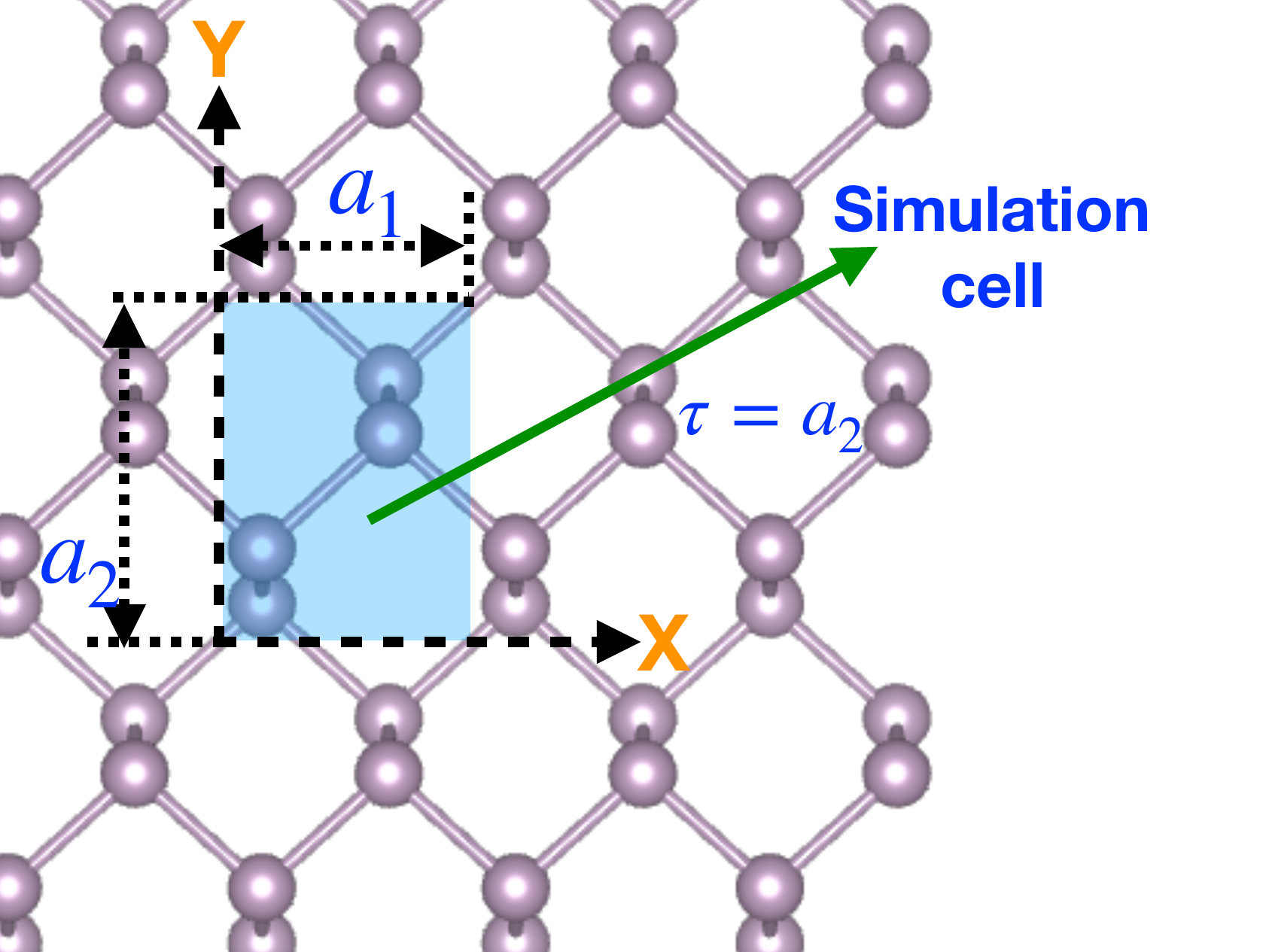}}\quad%\label{fig:phosphorene_sheet_top_view}\quad
\subfloat[Side view of a phosphorene sheet. Different structural parameters are labeled. Values available in Table \ref{Table:Phosphorene_Parameters}.]{\includegraphics[trim={0.5cm 5cm 12.5cm 5cm}, clip, width=0.45\textwidth]{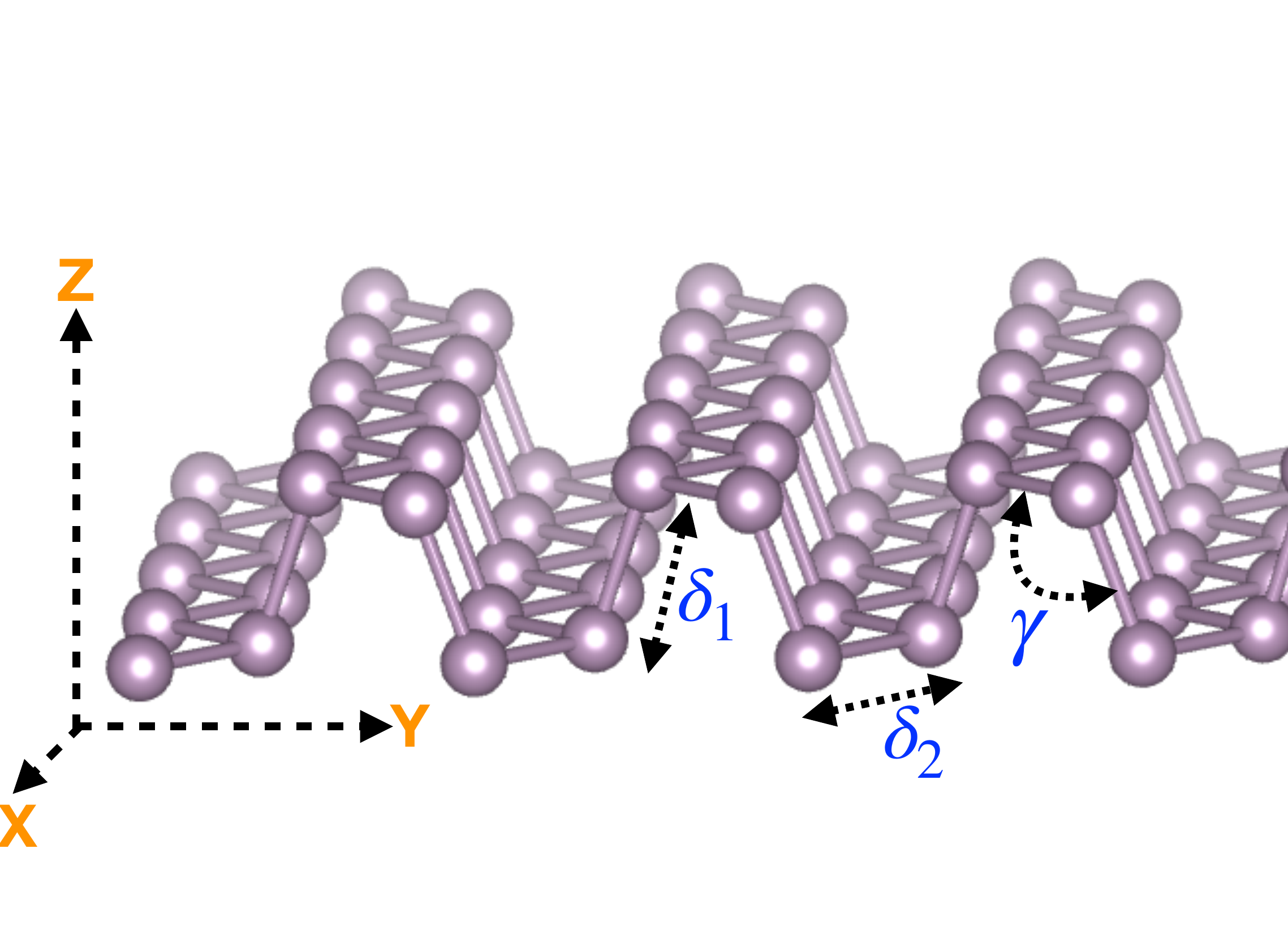}}%\label{fig:phosphorene_sheet_side_view}
\caption{Phosphorene  lattice as computed using a periodic DFT calculation. This provides the starting point for Helical DFT calculations involving phosphorene nanotubes. For zigzag tubes, the parameter $\tau$ in the two-generator helical group is set to be equal to the lattice constant $a_2$ along the Y-axis.}
\label{fig:phosphorene_sheet}
\end{figure}

To represent the zigzag phosphorene nanotubes in the Helical DFT simulations described next, we roll up the phosphorene sheet along the X-axis and place the atoms from the aforementioned  periodic unit cell (the shaded region in Figure \ref{fig:phosphorene_sheet}) into the Helical DFT simulation cell ($\Omega$). In the absence of relaxation effects, the pitch $\tau$ associated with the two-generator helical group, is equal to the lattice constant along the Y axis. Furthermore, the angle $\Theta$ associated with the cyclic group order $\mathfrak{N}$ is related to the radius of the nanotube via the relation $\displaystyle\Theta = \frac{2\pi}{\mathfrak{N}} = \frac{a_1}{R_{\text{avg.}}}$. Here  $a_1$ denotes the lattice vector along the X-axis in the phosphorene sheet, and $R_{\text{avg.}}$ denotes the average radial coordinate of the atoms in the fundamental domain of the nanotube. For chiral nanotubes, we additionally include a no-zero twist angle parameter $\alpha$. Figure \ref{fig:PNTs} shows examples of two phosphorene nanotubes studied in this work using Helical DFT.
\begin{figure}[ht]
\centering
\subfloat[Zigzag phosphorene nanotube.]{\includegraphics[trim={12cm 4cm 5cm 1.5cm}, clip, width=0.48\textwidth]{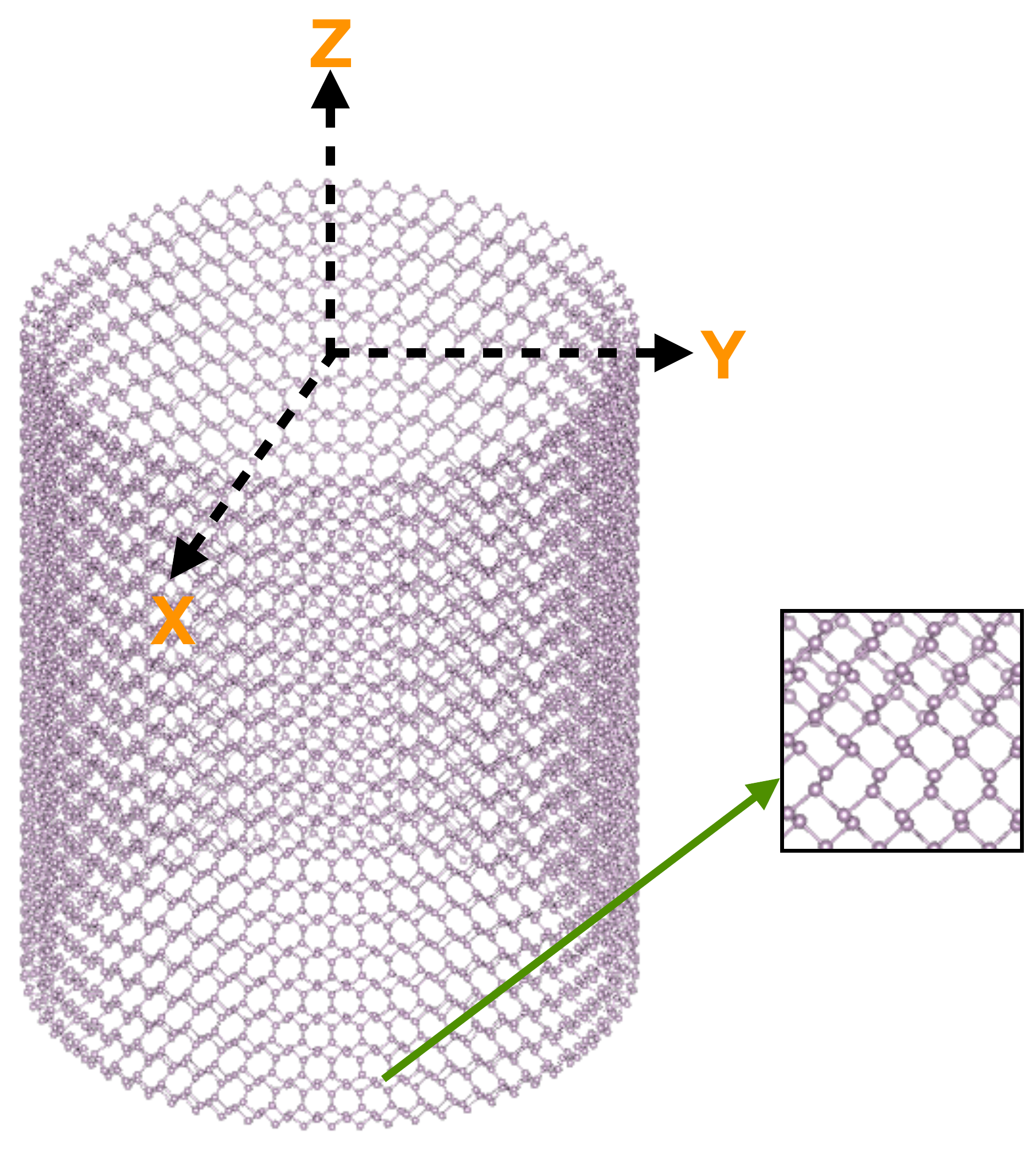}}\quad
\subfloat[Chiral phosphorene nanotube obtained by setting  the angle of twist parameter $\alpha = 0.005$.]{\includegraphics[trim={12cm 4cm 5cm 1.5cm}, clip, width=0.48\textwidth]{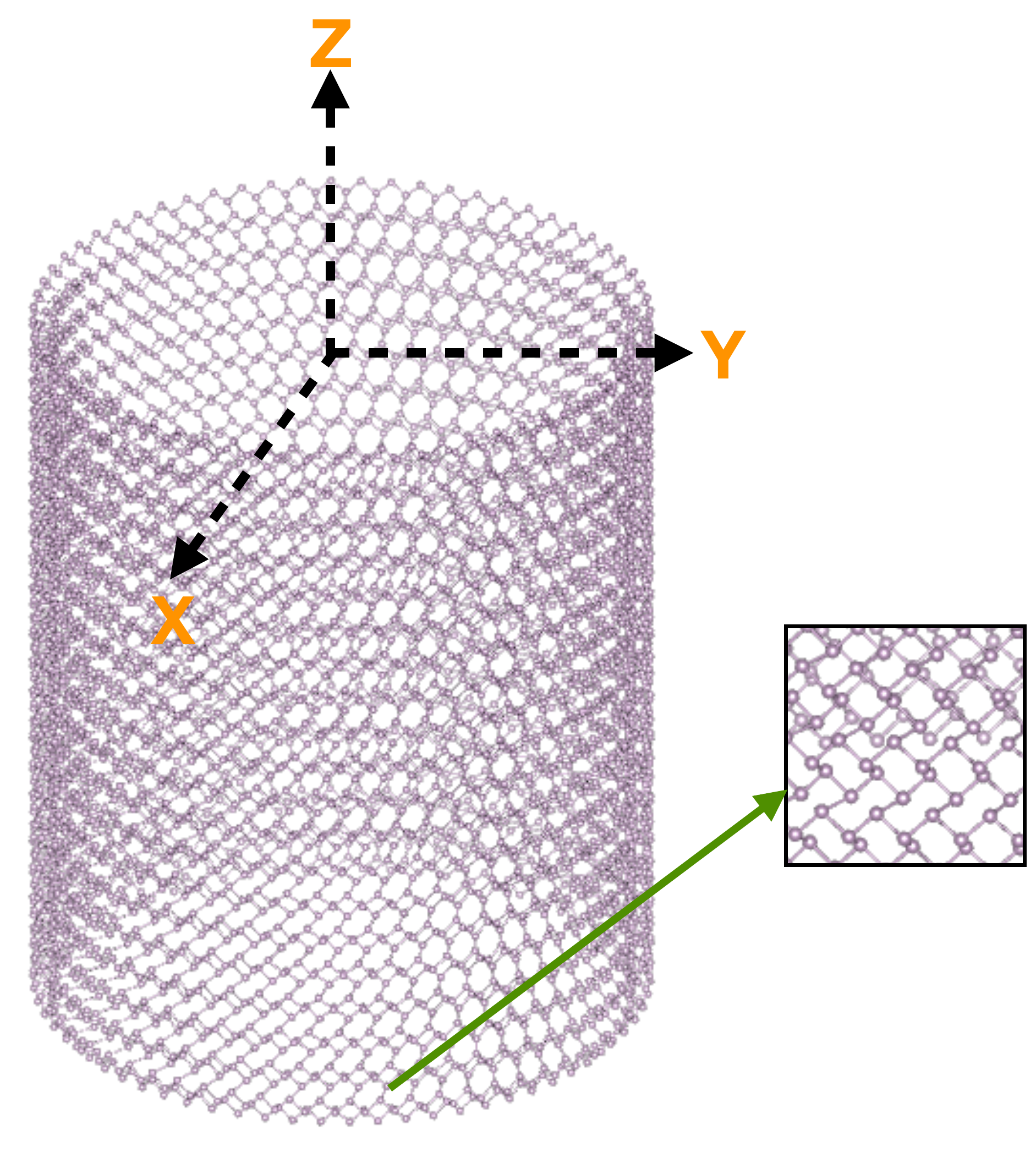}}
\caption{Representative phosphorene nanotubes (radius = $3.4$ nanometers) studied in this work using Helical DFT. Insets show zoomed-in views of the atomic arrangement to highlight differences between the two cases.}
\label{fig:PNTs}
\end{figure}
\subsection{Convergence and accuracy}
\label{subsec:convergence_accuracy}
To study the convergence properties of our numerical implementation, we consider a chiral phosphorene nanotube of radius $1.76$ nanometers. The cyclic group order is $\mathfrak{N} = 32$ and the twist angle parameter is $\alpha = 0.0025$. For a given structure and a fixed simulation domain, the two convergence parameters under study are the mesh spacing $h = \text{Max.}\bigg(h_r, \tau h_{\theta_1}, \big(\frac{R_{\text{in}} + R_{\text{out}}}{2}\big)h_{\theta_2}\bigg)$, and the number of points $N_{\eta}$ used for discretizing the set $\mathfrak{I}$ in reciprocal space. As a reference calculation, we computed the electronic structure of the system with the finest mesh ($h = 0.23$ Bohr) and the highest value of $N_{\eta}$ ($=23$) that we could afford under computational resource constraints. First, for the mesh convergence study, we fix $N_{\eta} = 23$ and carry out a series of calculations with $h \approx 0.30, 0.40, 0.50, 0.60, 0.65, 0.70$ Bohr.  Next, for studying convergence with respect to $N_{\eta}$, we fix $h = 0.23$ Bohr and carry out a series of calculations with  $N_{\eta} = 1, 3, 7, 11, 15, 19$. We plot the errors in the energy per atom and the atomic forces in each of the above cases in  Figure \ref{fig:convergence}.

From the figures, it is clear that the code converges to the reference calculations systematically. Using straight-line fits to the convergence data with respect to $h$, we find slopes of $5.8$ and $7.6$ for the energies and forces, respectively. These numbers are very nearly identical to the convergence rates observed in finite difference calculations involving cylindrical coordinates \citep{ghosh2019symmetry} and are also comparable to finite difference calculations in affine coordinate systems \citep{ghosh2017sparc_1, ghosh2017sparc_2}. From the data, we are also able to estimate that the parameters $h = 0.40$ Bohr and $N_{\eta} = 11$ are more than sufficient to reach chemically accurate energies and forces (the exact errors for these choices with respect to the reference calculation were $2\times 10^{-4}$ Ha/atom and $8\times10^{-5}$ Ha/Bohr in the energies and forces, respectively). For computational efficiency therefore, we use this set of parameters for all relaxation calculations described subsequently, and switch to the reference calculation parameters (i.e., $h = 0.23$ Bohr, $N_{\eta} = 23$)  only when more accurate energies / band gaps are required at the end of a relaxation procedure.\footnote{The convergence thresholds (in terms of relative residuals) for the SCF iterations and the Poisson problem had been set to $10^{-6}$ and $5\times10^{-9}$ for the above calculations, and we continue to use these values for all subsequent simulations.}

Verifying the accuracy of the Helical DFT code with respect to standard plane-wave codes can be challenging since the plane-wave codes may be required to include an enormous number of atoms in the periodic unit cell in order to replicate the exact system being studied by Helical DFT. In case of the above chiral nanotube system for example, over $50,000$ atoms would be needed, making the planewave calculation unfeasible. Therefore, we carry out this accuracy check in two steps. First, we set up a Helical DFT calculation for a zigzag nanotube (i.e., $\alpha = 0$) of radius $0.94$ nanometers. For this tube, the cyclic group order $\mathfrak{N} = 16$. Then, we simulate this tube using the ABINIT code by employing a $64$ atom unit cell (periodicity was enforced along the Z axis, Dirichlet boundary conditions were enforced along X and Y axes by padding with a large amount of vacuum). We converged both codes to the extent allowed by computational resources, and observed that the energies (in Ha/atom) and the forces (in Ha/Bohr)  from these two calculations agreed with each other to $1 \times 10^{-4}$  or better. Next to study a case for which $\alpha \neq 0$, we set up an artificial system consisting of atoms along a single helix (similar to the configuration in Figure \ref{fig:Helical_Bloch_Theorem}). This system was generated using a single generator helical group ($\mathfrak{N} = 1$) and used $\alpha = 0.01$. The helical unit cell had $2$ atoms, while  the periodic unit cell in ABINIT involved $200$. Once again, upon convergence with respect to their respective discretization parameters, the codes produced results that differed from each other by about $1 \times 10^{-4}$ Ha/atom or Ha/Bohr. This completes the accuracy tests.\footnote{It is likely that Helical DFT can be made to agree with ABINIT results to finer levels of accuracy. However, the quasi-one-dimensional nature of the systems being studied, and the slow convergence of the electrostatics  requires the use of large amounts of vacuum padding in the ABINIT supercell, and this tends to cause serious convergence issues as the energy cutoff is increased.}

Through the above examples, we were also able to observe that for realistic helical nanostructure simulations,  the wall time for Helical DFT can be up to orders of magnitude smaller compared to a well optimized plane-wave code like ABINIT, making it a powerful first principles tool in the study of such systems.
\begin{figure}[ht]
\centering
\subfloat[Convergence with respect to discretization in real space.]{\scalebox{0.9}
{\includegraphics[width=0.49\textwidth]{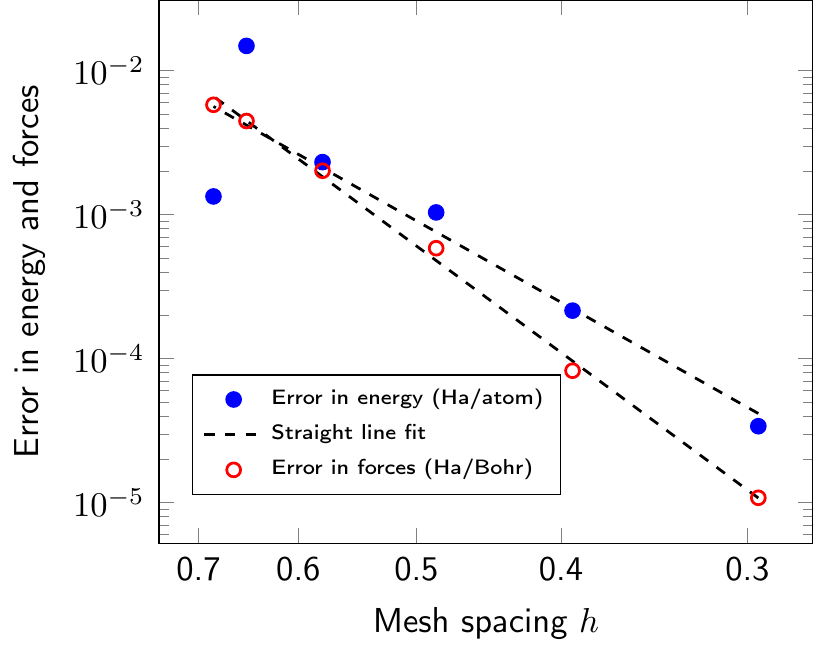}}
} % End subfloat
\quad$\;$
\subfloat[Convergence with respect to number of points used to discretize $\mathfrak{I}$.]{\scalebox{0.9}{
\includegraphics[width=0.49\textwidth]{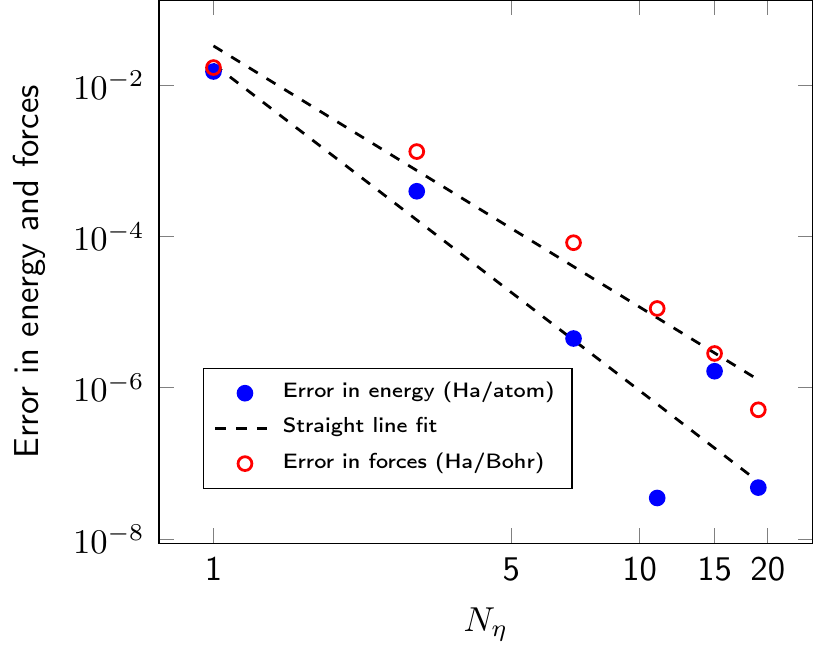}}
}
\caption{{Convergence properties of the Helical DFT code with respect mesh size $h$ and the number of points $N_{\eta}$ used to discretize the set $\mathfrak{I}$. The error in the forces is computed by considering the magnitude of the maximum difference in all the force components of all the atoms.}}
\label{fig:convergence}
\end{figure}
\subsection{Simulation of twisting: Ab initio computation of torsional stiffness}
\label{subsec:twisting}
By using the two-generator helical group $\calG_2, $ it is possible to describe torsional deformations in nanostructures \citep{Dumitrica_James_OMD, Dumitrica_Tight_Binding1, Dumitrica_Tight_Binding2, Pekka_Twisted_GNR, Pekka_Efficient_Approach}. We use this procedure here\footnote{Due to the use of helical symmetries, the structures being modeled in the simulations are infinite. In practice, nanotube structures have finite extent, though some can have lengths of the order of macroscopic sizes \citep{zheng2004ultralong}. The edge effects in these materials are expected to decay as one moves towards the interior of the material, both in the continuum elasticity sense \citep{toupin1965saint, mielke1988saint}  and at the level of the electronic structure \citep{prodan2005nearsightedness}. This offers a justification for the conceptual correctness of the simulations.} to illustrate the utility of Helical DFT in extracting mechanical properties of nanomaterials \textit{ab initio}. Specifically, we investigate the behavior of a zigzag phosphorene nanotube (nanotube radius $3.4$ nanometer, cyclic group order $\mathfrak{N} = 64$.) under twisting deformations. We begin with the untwisted structure ($\alpha = 0$) and relax the atoms in the simulation cell till all force components on all atoms are below $10^{-3}$ Ha/Bohr.\footnote{Due to the relatively large nanotube radius, the atoms in the untwisted nanotube are in an environment similar to that in the phosphorene sheet. Since the phosphorene sheet itself had been relaxed well, the forces on the atoms in the nanotube were relatively small to begin with, and structural relaxation was usually completed in just a few steps.} Starting from this relaxed structure, we apply twisting deformations to the nanotube by prescribing non-zero values of the twist angle parameter $\alpha$ in the two-generator helical group $\calG_2$ used for describing the nanotubes. We varied $\alpha$ in steps of $0.001$ and relaxed the resulting nanotube structure in each case, till all components of the forces on all the atoms in the simulation cell dropped below $10^{-3}$ Ha/Bohr once again. The largest twist we considered is about $5$ degrees per nanometer length of the tube, which corresponds to $\alpha = 0.006$. Anticipating a quadratic dependence of the nanotube twist energy for small values of alpha \citep{Dumitrica_James_OMD}, we plot the results as shown in Figure \ref{fig:nanotube_twist_energy}. In that figure, we have also plotted the twisting energies obtained when the atomic  relaxation effects are not considered.

We write the twist energy per unit length of the tube as $U_{\text{twist}} = \half k_{\text{twist}}\,\beta^2 $, with $\beta$ denoting the twist per unit length of the tube ($\beta = \frac{2\pi\alpha}{\tau}$) and $k_{\text{twist}}$ denoting the torsional stiffness. From the figure, we use the straight line fits near zero twist to estimate the value of $k_{\text{twist}}$ as $13.66$ eV nm $\text{(deg)}^{-2}$ and $27.33$ eV nm $\text{(deg)}^{-2}$, for the relaxed and unrelaxed cases respectively\footnote{These values implicitly use $\beta$ in degrees per nanometer. They need to be multiplied with a factor of $(180/\pi)^2$ to be stated in units conventionally used in other work (e.g. \citep{Dumitrica_James_OMD, Helical_DFT_Paper_2}) i.e.,  eV nm.}. It is also evident from the figure that non-linear effects start to play a noticeable role in this nanotube at around $4$ degrees of twist per nanometer.  These simulations serve as an example of determining constitutive parameters directly from quantum mechanics using Helical DFT. %Can the factor of 2 be explained ? Also, cubic behavior of torsional stiffness \citep{jeong2007elastic, vercosa2010torsional, Dumitrica_James_OMD}

\begin{figure}[ht]
\centering
\subfloat{\scalebox{1.0}
{\includegraphics[width=0.50\textwidth]{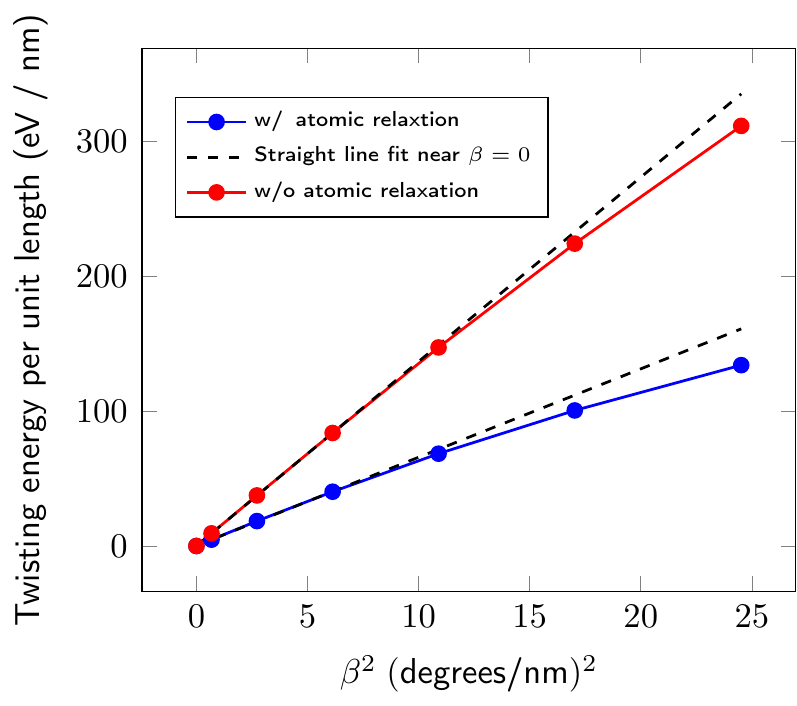}}}
\caption{Twist energy of a zigzag phosphorene nanotube as obtained via Helical DFT.}
\label{fig:nanotube_twist_energy}
\end{figure}
\subsection{Electronic properties: Behavior of zigzag tubes under torsional deformations}
\label{subsec:Electronic_Properties_Twisting}
Since Helical DFT is an \textit{ab initio} simulation tool, it can give us insights into the electronic, optical and transport properties of materials under study. In particular, it may be used to shed light into the coupling of mechanical strains with these properties. To illustrate these points, we consider again the case of the zigzag phosphorene nanotube discussed above (nanotube radius $3.4$ nanometer, cyclic group order $\mathfrak{N} = 64$.). Figure \ref{fig:Helical_Band_Diagram} displays examples of \textit{helical band structure diagrams}, which like their periodic counterparts, can be used to illustrate the electronic levels in the system. A unique feature of these diagrams however, is that they can be used to display the variation of Kohn-Sham eigenvalues with respect to $\eta$ \textit{and} $\nu$. This makes them significantly easier to interpret than traditional periodic band diagrams for quasi-one-dimesional (nanotube-like) structures, which have only one index (i.e., the wave vector $k_z$ in the axial direction) labeling the  Kohn-Sham states.\footnote{Similar band diagrams have been considered in \citep{aghaei2012symmetry, aghaei2011symmetry} in the context of phonon calculations of carbon nanotubes.} We anticipate that helical band diagrams like the ones shown are likely to emerge as powerful tools in understanding instabilities in optical and electronic materials.
\begin{figure}[ht]
\centering
\subfloat[Helical band diagram in $\eta$, along $\nu = 0$.]{\scalebox{0.9}
{\includegraphics[width=0.49\textwidth]{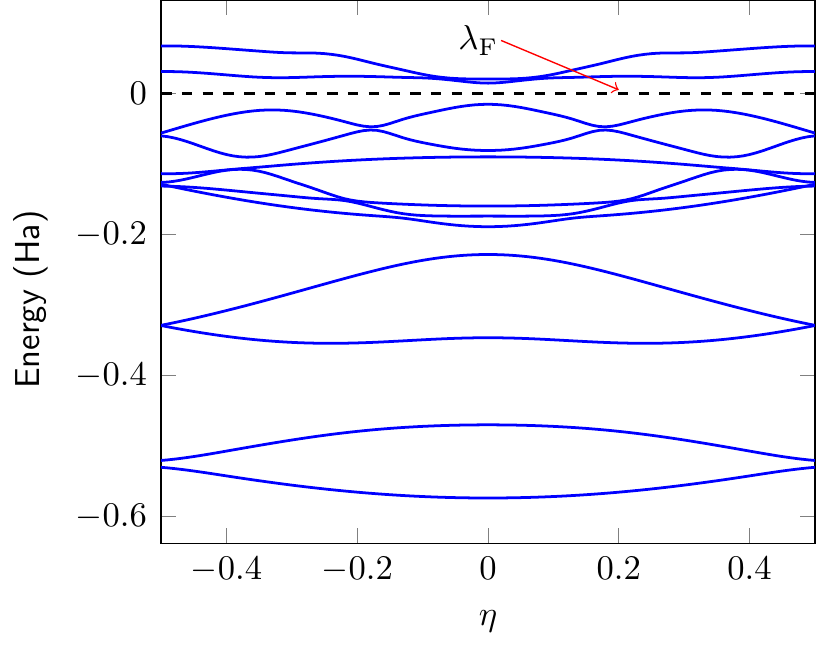}}}
\quad$\;$
\subfloat[Helical band diagram in $\nu$, along $\eta = 0$.]{\scalebox{0.9}
{\includegraphics[width=0.49\textwidth]{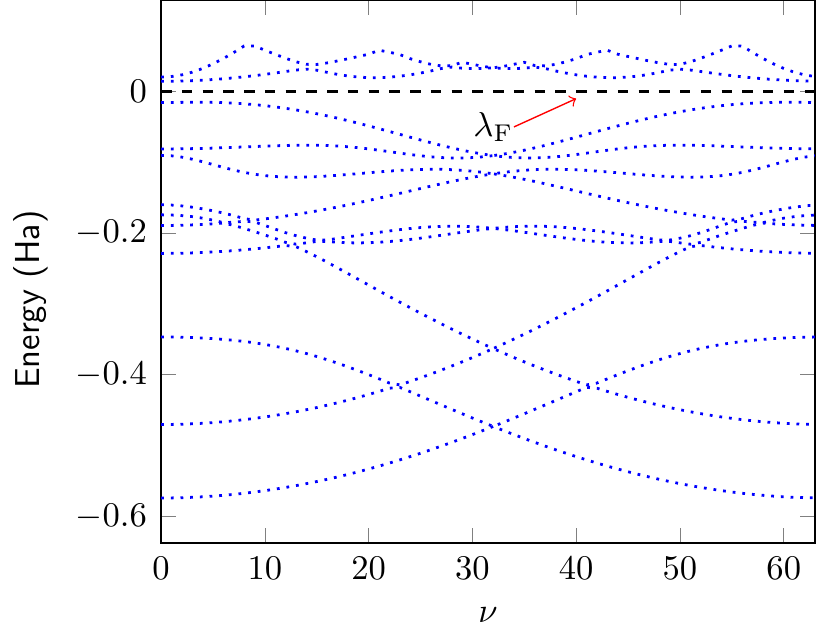}} 
}
\caption{Helical band structure diagrams for an untwisted zigzag phosphorene nanotube (radius $=3.4$ nanometers). $\lambda_{\text{F}}$ denotes the Fermi level (shifted to $0$).}
\label{fig:Helical_Band_Diagram}
\end{figure}
The helical band diagrams described above can be used to compute the size of the band gap of the  system and infer whether it is conducting, semiconducting or insulating.\footnote{For an extended system, such as a helical structure, the band gap is defined as the difference between the conduction band minimum (CBM) and the valence band maximum (VBM). Within Helical DFT, this can be computed as the difference between the smallest eigenvalue above the Fermi level and the largest eigenvalue below the Fermi level, as $\eta$ is varied in $\mathfrak{I}$ and $\nu$ in $0,1,2,\ldots\mathfrak{N}-1$.} We plot in Figure \ref{fig:bandgap_vs_alpha}  the variation in the band gap of the above nanotube, as it undergoes torsional deformations. The variation with and without atomic relaxation effects are both displayed. From the figure, we observe that the nanotube has a semiconducting behavior overall, and can be made to go from an insulating state at no twist (direct band gap of about $0.81$ eV), to a practically  conducting one, once the twist reaches about $5$ degrees per nanometer. This is a rather significant change in the electronic properties of the material, although in a mechanical sense, its deviation from simple linear elastic behavior (Figure \ref{fig:nanotube_twist_energy}) is still fairly modest at this level of twist.\footnote{Although it is well known that LDA is often unable to predict quantitatively accurate band gaps, the qualitative trends observed here are likely to be representative of actual physical behavior \citep{sham1985density, van2006quasiparticle, perdew1983physical, hybertsen1986electron, hybertsen1985first} in these nanotube systems. In any case, the formulation presented here does not have issues with regard to the use of more quantitatively accurate hybrid exchange correlation functionals \citep{marsman2008hybrid, heyd2003hybrid} whose implementation within our framework is a subject worthy of further investigation.}  
\begin{figure}[ht]
\centering
\subfloat{\scalebox{1.0}
{\includegraphics[width=0.50\textwidth]{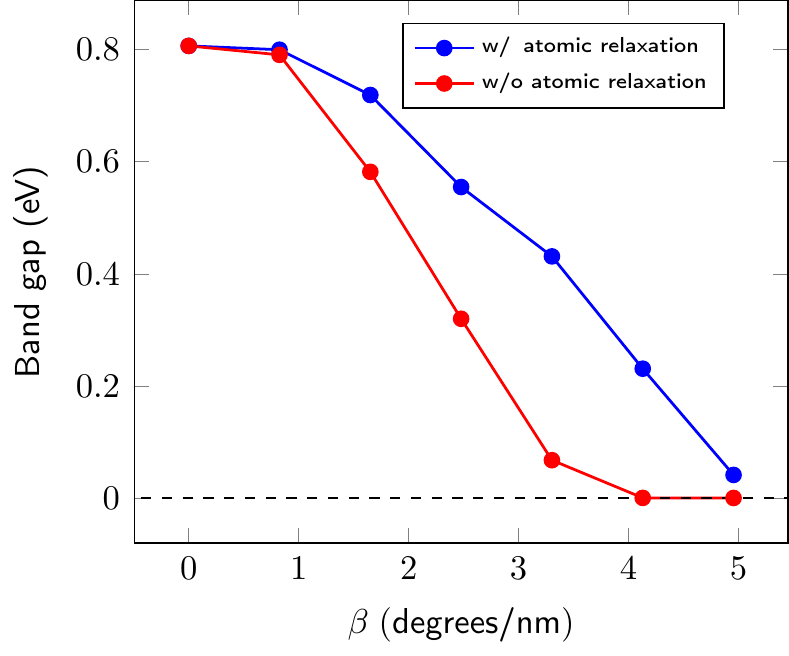}
}}
\caption{Effect of torsional deformation on the band gap of a zigzag phosphorene nanotube (radius $= 3.4$ nanometers).}
\label{fig:bandgap_vs_alpha}
\end{figure}

To further illustrate the above electronic transition, we compute the (electronic) density of states of the nanotube without and with twist ($\beta \approx 5$ degrees/nanometer). Following \citep{Defranceschi_LeBris}, we write this at a given energy level $E$ and an electronic temperature of $T_{\text{e}}$ as:
\begin{align}
\aleph_{T_{\text{e}}}(E) =2\,\int_{\mathfrak{I}} \frac{1}{\mathfrak{N}}\sum_{\nu = 0}^{\mathfrak{N}-1}\bigg(\sum_{j=1}^{\infty} f^\prime_{T_{\text{e}}}\big(E - \lambda_j(\eta,\nu)\big)\bigg)\,d\eta, 
\end{align}
and evaluate it along a fine mesh of values of $E$ in the range $[-1,1]$. The results are shown in Figure \ref{fig:Density_of_States}. It is apparent from the figure that the electronic states in the system undergo significant change as the nanotube is subjected to twisting. In particular, the number of states at or near the Fermi level $\lambda_{\text{F}}$, increases to values well above $0$, indicating onset of metallic behavior when the tube is twisted.
\begin{figure}[ht]
\centering
\subfloat[Density of states of untwisted nanotube.]{\scalebox{0.9}
{\includegraphics[width=0.49\textwidth]{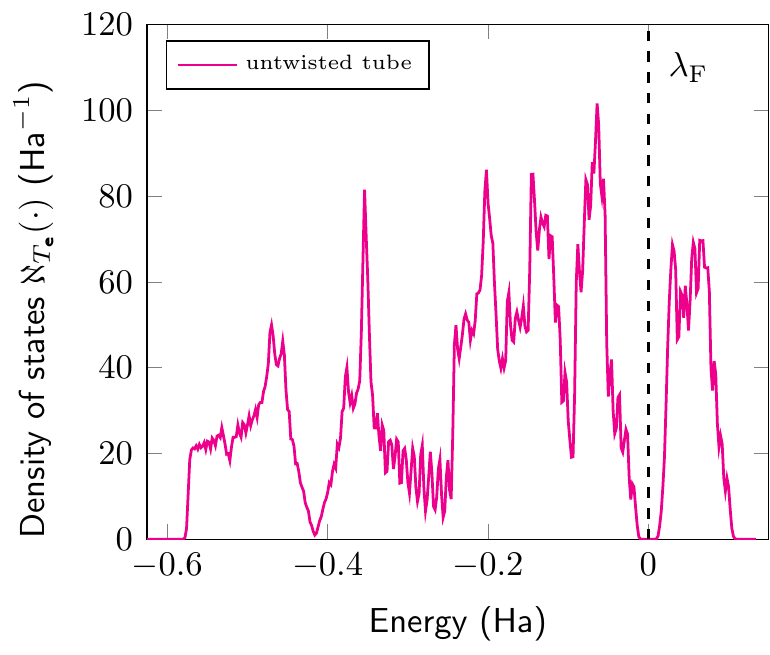}
}}\quad
\subfloat[Density of states of twisted nanotube.]{\scalebox{0.9}
{\includegraphics[width=0.49\textwidth]{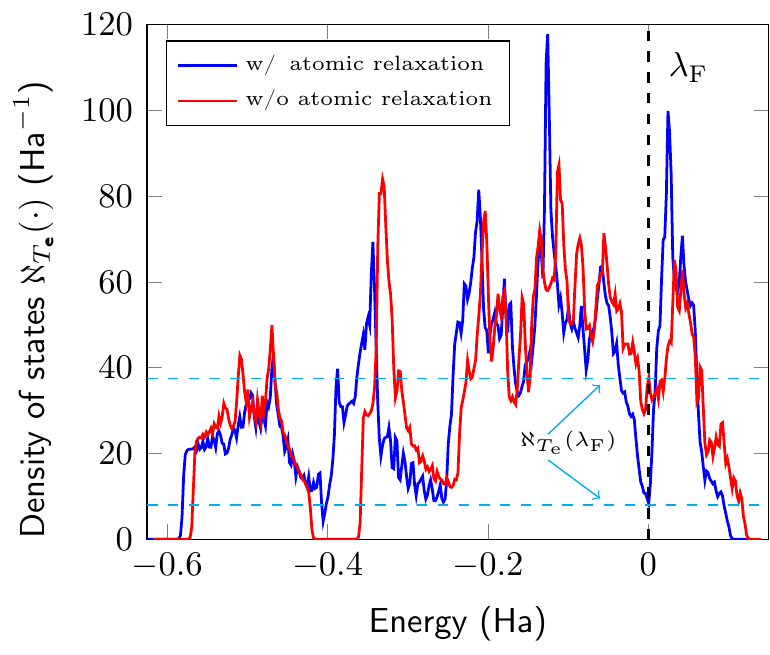}} 
}
\caption{{Electronic density of states plot for a zigzag phosphorene nanotube (radius $=3.4$ nanometers), without and with twist. $\lambda_{\text{F}}$ denotes the Fermi level (shifted to $0$).}}
\label{fig:Density_of_States}
\end{figure}
\subsection{Electronic properties: Behavior of chiral tubes under axial strains}
\label{subsec:Electronic_Properties_Chiral_Axial}
Finally, we use Helical DFT to study a chiral phosphorene nanotube. We choose a tube with $\mathfrak{N} = 64$ and $\tau = 8.244$ Bohr as before, and also set $\alpha = 0.005$. We relax the positions of the atoms within the simulation cell and use the resulting nanotube (observed to have an indirect band gap of about $0.24$ eV) as the starting structure for subsequent simulations. We subject this tube to (both tensile and compressive) axial strains by varying $\tau$, and relax the atomic positions in each case. Figure \ref{fig:band_gap_chiral} shows the variation in the band gap of this tube at different values of axial strain. The effect of not including atomic relaxation after the tube is subjected to the strains is also shown. 
\begin{figure}[ht]
\centering
\subfloat{\scalebox{1.0}
{\includegraphics[width=0.49\textwidth]{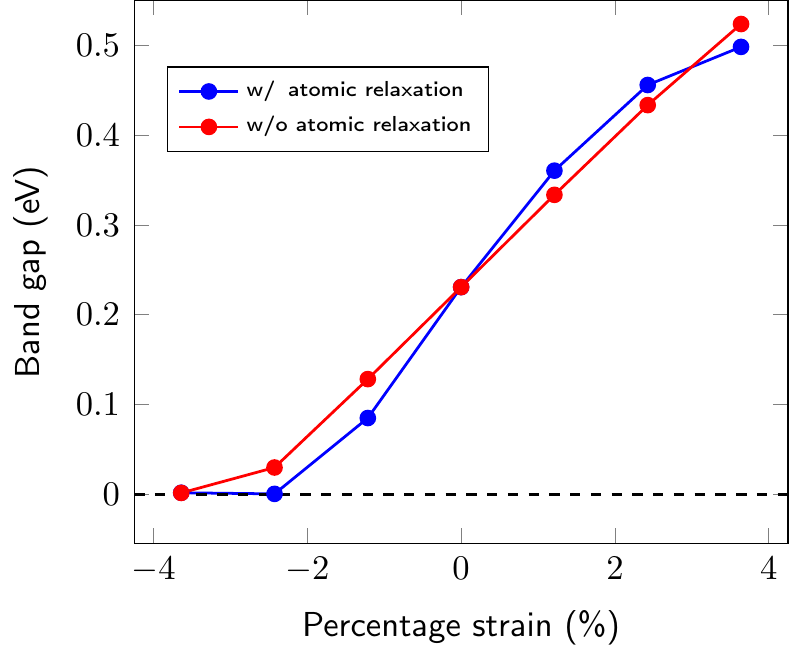}}
}
\caption{Effect of axial strain on the band gap of a chiral phosphorene nanotube ($\mathfrak{N} = 64$, $\alpha = 0.005$).}
\label{fig:band_gap_chiral}
\end{figure}
From the figure, we see that in the range of strains considered, compressive strains tend to reduce the band gap, while tensile strains appear to increase it. The tube appears to transition into a metallic state at about $4 \%$ compressive strain. Based on the nature of the plot and motivated by earlier work on chiral carbon nanotubes \citep{yang1999band, vercosa2010torsional}, we fitted a one-term Fourier series to this band gap data, and found that this produces a high quality fit for both the relaxed and unrelaxed cases. This suggests that like the case of carbon nanotubes, it might be possible to build (approximate) tight-binding type models of the band gap behavior \citep{endo2013carbon} for phosphorene nanotubes. This is a topic that warrants further investigation. Notably, the parameters in such models can be provided through high-quality ab initio simulations based on Helical DFT. 

The above {ab initio} simulations of the phosphorene nanotubes (both zigzag and chiral) suggests that these materials have highly adjustable electronic states. Therefore, they might find future applications as nanomaterials with tunable electronic/optical/transport properties. The simulations also highlight the possibility of using inhomogeneous strain modes for altering these properties, instead of homogeneous strain modes which are considered in the strain-engineering literature \citep{pereira2009strain, ghassemi2012field, fei2014strain} more commonly. Further investigations of this material and others, along these lines, is the scope of future work.

Finally, we find it worthwhile to point out that the phosphorene nanotube simulations described above (both electronic and mechanical)  would be very challenging, or  well-nigh impossible to carry out using conventional first principles techniques, even with the aid of massively parallel high-performance computing resources. In contrast, our MATLAB implementation of Helical DFT often  allows such simulations to be carried out within a few hours of simulation wall time on a desktop workstation or a single node of a supercomputing cluster, thus serving to reinforce the novelty and practical utility of the approach.
\section{Conclusions and Future Directions}
\label{sec:Conclusion}
In summary, we have presented Helical DFT --- a novel, systematic first principles simulation framework for  systems with helical symmetries. We have presented a derivation of the equations of Kohn-Sham theory, as they apply to the case of helical structures. Our derivation is systematic, self-contained and for the most part, mathematically rigorous. We have then solved these governing equations numerically by using a finite difference method based on helical coordinates. Using this working realization of the proposed approach, we have carried out simulations involving phosphorene nanotubes, extracted their mechanical behavior \textit{ab initio}, and identified changes in the electronic properties of this material as it undergoes twisting. 

Having laid out this foundational work on Helical DFT, we now discuss a number of avenues for  future investigation:
\begin{itemize}
\item {Development of an efficient spectral solution scheme:} While the current finite difference based implementation of Helical DFT  allows us to investigate a number of materials systems of interest, it also suffers a number of computational deficiencies (Section \ref{sec:Numerical_Implementation}). Following our earlier work on the ab initio simulations of cluster systems with arbitrary point group symmetries, we have already formulated, and are currently in the process of implementing a spectral scheme for solving the governing equations \citep{My_PhD_Thesis, agarwal2021spectral, My_Shivang_HelicES_paper}. This is expected to completely resolve the issues with the finite difference formulation and provide a more efficient numerical implementation, thus opening the door to the simulation of  more complex helical structures.
\item {Mechanistic simulations of quasi-one-dimensional systems:} The materials science literature is rich with examples of quasi-one-dimensional structures that have been discovered and/or synthesized through experimental means. The simulation tools developed here can be used to characterize these materials computationally. In particular,  the use of helical symmetry adapted ab initio molecular dynamics \citep{Hutter_abinitio_MD} can be used to study the mechanical behavior of these materials under axial and torsional strains, as well as their instabilities and defects \citep{Dayal_James_viscometry1, Dumitrica_James_OMD, mukherjee2020symmetry}.
\item {Helical Wannier states and the helical Berry phase:} The helical Bloch Floquet transform $\calU$ introduced in eqs.~\ref{eq:operatorU} and \ref{eq:operatorU_G2} allows us to consider the notion of helical Wannier states. These, like their periodic counterparts, are defined through the action of $\calU^{-1}$ on the helical Bloch states \citep{Odeh_Keller}, and can be expected to be exponentially localized for insulating systems \citep{kohn1959analytic}. Additionally, like the case of periodic systems, a geometric phase (i.e., the Berry phase \citep{berry1984quantal}) associated with the helical Bloch phase factor may be defined. These observations are likely to spur the development of novel computational analysis methods \citep{damle2018disentanglement, damle2019variational}  for helical materials, as well as the discovery of novel topological materials \citep{lindner2011floquet, ando2013topological}.
\item {Search for exotic materials properties and study of multi-physics coupling:} Since Helical DFT is a first principles simulation technique, it allows investigation of the effect of torsional deformations on a material's optical, electronic, magnetic and transport properties. Like the case of the phospherene nanotube considered in this work, it is possible that torsional or axial deformations in certain helical  structures might induce a significant redistribution of electronic states in the material, leading to the appearance of interesting collective properties. In this regard, the investigation of the nanoscale flexoelectric effect \citep{deng2014flexoelectricity, ahmadpoor2015flexoelectricity, kalinin2008electronic, nguyen2013nanoscale, dumitricua2002curvature, chandratre2012coaxing} in helical systems would be of particular interest, since owing to the quasi-one-dimensional nature of these materials, as well as the appearance of strain gradients in connection with torsional deformations, a significant polarization may appear along the axis of a nanostructure when twisted, thus leading to a strong flexoelectric effect. 
\item {Search for new phases of matter and coherent phase transformations:} Finally, as remarked in \citep{feng2018phase, feng2019phase}, the first principles techniques developed here might be instrumental in the discovery of novel phases of matter. Additionally, the discovery of transformations between such phases \citep{feng2019phase, ganor2016zig} --- particularly, coherent ones \citep{song2013enhanced, zhang2009energy} --- are likely to lead to new classes of active materials. The methods developed here are likely to be very useful in the characterization of energy barriers of such transformations and help in their design.
\end{itemize}
\begin{center}
---
\end{center}
\appendix
\section{Verifying that the operator \texorpdfstring{$\mathfrak{h}_{\eta}^{\textsf{aux}}$}{} in Eq.~\ref{eq:phi_equation_1} is symmetric}
\label{appendix:h_eta_symmetric}
The boundary conditions associated with the operator $\mathfrak{h}^{\textsf{aux}}_{\eta}$ introduced in Eq.~\ref{eq:phi_equation_1} are quite non-conventional. Here, we work through the steps of verifying that the operator is symmetric when augmented with these boundary conditions. For this, we consider smooth functions $f_1, f_2$ over $\calD$ (the interior of the fundamental domain $\calD_{\calG_1}$), obeying the same boundary conditions as $\phi$ in the proof of Theorem \ref{Thm:Existence_Helical_Bloch}, i.e.,  $f_{1,2}(\bfx) = f_{1,2}(\Upsilon_{\mathsf{h}} \circ \bfx)$, $\nabla f_{1,2}(\bfx)=\bfR_{2\pi\alpha}^{-1}\nabla f_{1,2}(\Upsilon_{\mathsf{h}} \circ \bfx)$ for $\bfx \in \partial\calD^{\,z=0}$ and $ f_{1,2}(\bfx) = 0$ for $\bfx \in \partial\calD^{\,r=R}$. Then we have:
\begin{align}
\nonumber
&\innprod{\mathfrak{h}^{\textsf{aux}}_{\eta} f_1}{f_2}{\Lpspc{2}{}{\calD}} =\\ &-\half \int_{\calD} \Delta f_1 \overline{f_2}\,d\bfx -  i\frac{2\pi\eta}{\tau}\int_{\calD} \pd{f_1}{x_3} \overline{f_2}\,d\bfx - \frac{2\pi^2\eta^2}{\tau^2}\int_{\calD}f_1 \overline{f_2} \,d\bfx + \int_{\calD} V f_1  \overline{f_2}\,d\bfx\;.
\label{eq:symmetric_Bloch_hamiltonian_1}
\end{align}
On using integration by parts \citep{Evans_PDE}, the first term on the right-hand side becomes:
\begin{align}
-\half \int_{\calD} \overline{f_2}\Delta f_1 \,d\bfx = \half \bigg(\int_{\calD} \nabla \overline{f_2}\cdot \nabla f_1 \,d\bfx - \int_{\partial \calD}\overline{f_2} \nabla f_1 \cdot d\bfs\bigg)\,,
\label{eq:symmetric_Bloch_hamiltonian_2}
\end{align}
where $d\bfs$ is the oriented surface measure. The boundary terms on the right-hand side of eq.~\ref{eq:symmetric_Bloch_hamiltonian_2} can be split as:
\begin{align}
\int_{\partial \calD}\overline{f_2} \nabla f_1 \cdot d\bfs =   \int_{\partial \calD^{\,r=R}}\!\overline{f_2} \nabla f_1 \cdot d\bfs + \int_{\partial \calD^{\,z=\tau}}\!\overline{f_2} \nabla f_1 \cdot d\bfs + \int_{\partial \calD^{\,z=0}}\!\overline{f_2} \nabla f_1 \cdot d\bfs\,.
\label{eq:symmetric_Bloch_hamiltonian_3}
\end{align}
The first term on the right-hand side of eq.~\ref{eq:symmetric_Bloch_hamiltonian_3} goes to zero due to the Dirichlet boundary condition $f_2(\bfx) = 0$ for $\bfx \in \partial \calD^{\,r=R}$. To evaluate the second and the third terms, collectively denoted as $I_1$ henceforth, we write the surface measures in terms of the local unit normals to get:
\begin{align}
\nonumber
I_1 &= \int_{\partial \calD^{\,z=\tau}}\!\overline{f_2} \nabla f_1 \cdot (\bfe_3\,d\text{s}) + \int_{\partial \calD^{\,z=0}}\!\overline{f_2} \nabla f_1 \cdot (-\bfe_3\, d\text{s}) \\ &= \int_{\partial \calD^{\,z=\tau}}\!\overline{f_2} \nabla f_1 \cdot (\bfe_3\,d\text{s}) - \int_{\partial \calD^{\,z=0}}\!\overline{f_2} \nabla f_1 \cdot (\bfe_3\, d\text{s})\,.
\label{eq:symmetric_Bloch_hamiltonian_4}
\end{align}
We now express $\bfy \in \calD^{\,z=\tau}$ as $\Upsilon_{\mathsf{h}} \circ \bfx$ with $\bfx \in \calD^{\,z=0}$, so that the first integral on the right-hand side of eq.~\ref{eq:symmetric_Bloch_hamiltonian_4} above can be rewritten as (with the dependence on $\bfx$ and $\bfy$ shown explicitly):
\begin{align}
\int_{\bfy \in \partial \calD^{\,z=\tau}}\!\overline{f_2(\bfy)} \nabla f_1(\bfy) \cdot \big(\bfe_3\,d\text{s}_{\bfy}\big)  = \int_{\bfx \in \partial \calD^{\,z=0}}\!\overline{f_2(\Upsilon_{\mathsf{h}} \circ \bfx)} \nabla f_1(\Upsilon_{\mathsf{h}} \circ \bfx) \cdot \big(\bfe_3\,d\text{s}_{\Upsilon_\mathsf{h} \circ \bfx}\big)\,.
\label{eq:symmetric_Bloch_hamiltonian_5}
\end{align}
Here, the notation $d\text{s}_{\bfy}$ is used to denote the surface measure centered at the point $\bfy$, and similarly $d\text{s}_{\Upsilon_{\mathsf{h}} \circ \bfx}$ denotes the surface measure centered at the point $\Upsilon_\mathsf{h} \circ \bfx$. Since $\Upsilon_{\mathsf{h}}$ is an isometry, $d\text{s}_{\Upsilon_{\mathsf{h}} \circ \bfx}$ has the same magnitude as $d\text{s}_{\bfx}$. Furthermore, the boundary conditions imply that $\overline{f_2(\Upsilon_{\mathsf{h}} \circ \bfx)} = \overline{f_2(\bfx)}$, and:
\begin{align}
\nonumber
\nabla f_1(\Upsilon_{\mathsf{h}} \circ \bfx) \cdot \bfe_3 &= \big(\bfR_{2\pi\alpha}\bfR_{2\pi\alpha}^{-1}\nabla f_1(\Upsilon_{\mathsf{h}} \circ \bfx)\big) \cdot \bfe_3 \\
&= \big(\bfR_{2\pi\alpha}^{-1}\nabla f_1(\Upsilon_{\mathsf{h}} \circ \bfx)\big) \cdot \big(\bfR_{2\pi\alpha}^{T}\bfe_3\big) = \nabla f_1(\bfx) \cdot \bfe_3\,,
\label{eq:symmetric_Bloch_hamiltonian_6}
\end{align}
since $\bfR_{2\pi\alpha}$ (and hence $\bfR_{2\pi\alpha}^{T}$) has axis $\bfe_3$. Combining the above results, we arrive at:
\begin{align}
I_1 &= \int_{\partial \calD^{\,z=0}}\!\overline{f_2(\bfx)} \nabla f_1(\bfx) \cdot \big(\bfe_3\,d\text{s}_{\bfx}\big) - \int_{\partial \calD^{\,z=0}}\!\overline{f_2(\bfx)}\!\nabla f_1(\bfx) \cdot \big(\bfe_3\,d\text{s}_{\bfx}\big) = 0\,,
\label{eq:symmetric_Bloch_hamiltonian_7}
\end{align}
and therefore, eq.~\ref{eq:symmetric_Bloch_hamiltonian_2} reduces to:
\begin{align}
-\half \int_{\calD}\!\overline{f_2}\Delta f_1 \,d\bfx = \half \bigg(\int_{\calD}\!\nabla \overline{f_2}\cdot \nabla f_1 \,d\bfx\bigg)\,.
\label{eq:symmetric_Bloch_hamiltonian_8}
\end{align}
The second term on the right-hand side of eq.~\ref{eq:symmetric_Bloch_hamiltonian_1} is:
\begin{align}
-i\frac{2\pi\eta}{\tau}\bigg[\int_{\calD}\pd{f_1}{x_3}\overline{f_2}\,d\bfx \bigg ]= -i\frac{2\pi\eta}{\tau}\bigg[\int_{\calD}\!\bigg(\pd{f_1\overline{f_2}}{x_3} - f_1\pd{\overline{f_2}}{x_3} \bigg)\,d\bfx\,\bigg]\,.
\label{eq:symmetric_Bloch_hamiltonian_9}
\end{align}
We now consider the first integral on the right-hand side of eq.~\ref{eq:symmetric_Bloch_hamiltonian_9}, rewrite the integrand in divergence form, and use the Divergence Theorem \citep{Evans_PDE} to get:
\begin{align}
\int_{\calD}\!\pd{f_1\overline{f_2}}{x_3}\,d\bfx = \mathlarger{\int}_{\calD}\!\nabla \cdot 
\begin{pmatrix}
0 \\ 0 \\ f_1\overline{f_2}
\end{pmatrix}\,d\bfx\, =  \mathlarger{\int}_{\partial\calD}\!\begin{pmatrix}
0 \\ 0 \\ f_1\overline{f_2} 
\end{pmatrix} \cdot d\bfs\,,
\label{eq:symmetric_Bloch_hamiltonian_10}
\end{align}
with $d\bfs$ denoting the oriented surface measure (as earlier). The surface integral can be split as:
\begin{align}
\nonumber
&\mathlarger{\int}_{\partial\calD}\!\begin{pmatrix}
0 \\ 0 \\ f_1\overline{f_2} 
\end{pmatrix} \cdot d\bfs\, 
\\= &\mathlarger{\int}_{\partial\calD^{\,r=R}}\!\begin{pmatrix} 
0 \\ 0 \\ f_1\overline{f_2} 
\end{pmatrix} \cdot d\bfs + 
\mathlarger{\int}_{\partial\calD^{\,z=\tau}}\!\begin{pmatrix} 
0 \\ 0 \\ f_1\overline{f_2} 
\end{pmatrix} \cdot d\bfs +
\mathlarger{\int}_{\partial\calD^{\,z=0}}\!\begin{pmatrix} 
0 \\ 0 \\ f_1\overline{f_2} 
\end{pmatrix} \cdot d\bfs\,,
\label{eq:symmetric_Bloch_hamiltonian_11}
\end{align}
from which the first term on the right-hand side  vanishes due to Dirichlet boundary conditions on $\partial\calD^{\,r=R}$. For the second and the third terms, denoted collectively as $I_2$ henceforth, we write the oriented surface measures in terms of the local unit normals to get:
\begin{align}
\nonumber
I_2 = &\mathlarger{\int}_{\partial\calD^{\,z=\tau}}\!\begin{pmatrix} 
0 \\ 0 \\ f_1\overline{f_2} 
\end{pmatrix} \cdot \bfe_3\,d\text{s} -
\mathlarger{\int}_{\partial\calD^{\,z=0}}\!\begin{pmatrix} 
0 \\ 0 \\ f_1\overline{f_2} 
\end{pmatrix} \cdot \bfe_3\,d\text{s} \\
=&\int_{\partial\calD^{\,z=\tau}}\! f_1\overline{f_2}\,d\text{s} - \int_{\partial\calD^{\,z=0}}\! f_1\overline{f_2}\,d\text{s}
\label{eq:symmetric_Bloch_hamiltonian_12}
\end{align}
As earlier, we express $\bfy \in \calD^{\,z=\tau}$ as $\Upsilon_{\mathsf{h}} \circ \bfx$ with $\bfx \in \calD^{\,z=0}$, so that the first integral on the right-hand side of eq.~\ref{eq:symmetric_Bloch_hamiltonian_12} above can be rewritten by use of the boundary conditions as (the dependence on $\bfx$ and $\bfy$ has been shown explicitly):
\begin{align}
\label{eq:symmetric_Bloch_hamiltonian_13}
\nonumber
\int_{\bfy \in \partial\calD^{\,z=\tau}}\! f_1(\bfy)\overline{f_2}(\bfy)\,d\text{s}_{\bfy} &= \int_{\bfx \in \partial\calD^{\,z=0}}\!f_1(\Upsilon_{\mathsf{h}}\circ \bfx)\overline{f_2}(\Upsilon_{\mathsf{h}}\circ\bfy)\,d\text{s}_{\Upsilon_\mathsf{h}\circ\bfx} \\
&=  \int_{\bfx \in \partial\calD^{\,z=0}}\!f_1(\bfx)\overline{f_2}(\bfx)\,d\text{s}_\bfx\,.
\end{align}
It follows that,
\begin{align}
I_2 = \int_{\partial\calD^{\,z=0}}\! f_1\overline{f_2}\,d\text{s} - \int_{\partial\calD^{\,z=0}}\! f_1\overline{f_2}\,d\text{s} = 0\,,
\label{eq:symmetric_Bloch_hamiltonian_14}
\end{align}
and therefore, the second term on the right-hand side of  eq.~\ref{eq:symmetric_Bloch_hamiltonian_1} is:
\begin{align}
-i\frac{2\pi\eta}{\tau}\bigg[\int_{\calD}\pd{f_1}{x_3}\overline{f_2}\,d\bfx \bigg ]= i\frac{2\pi\eta}{\tau}\bigg[\int_{\calD}\!f_1\pd{\overline{f_2}}{x_3}\,d\bfx\,\bigg]\,.
\label{eq:symmetric_Bloch_hamiltonian_15}
\end{align}
Combining eqs.~\ref{eq:symmetric_Bloch_hamiltonian_1}, \ref{eq:symmetric_Bloch_hamiltonian_8} and \ref{eq:symmetric_Bloch_hamiltonian_15}, we get:
\begin{align}
\nonumber
&\innprod{\mathfrak{h}^{\textsf{aux}}_{\eta} f_1}{f_2}{\Lpspc{2}{}{\calD}} =\\ &\half \int_{\calD} \nabla \overline{f_2} \cdot \nabla {f_1}\,d\bfx +  i\frac{2\pi\eta}{\tau}\int_{\calD} f_1 \pd{\overline{f_2}}{x_3}\,d\bfx - \frac{2\pi^2\eta^2}{\tau^2}\int_{\calD}f_1 \overline{f_2} \,d\bfx + \int_{\calD} V f_1  \overline{f_2}\,d\bfx\;.
\label{eq:symmetric_Bloch_hamiltonian_16}
\end{align}
On the other hand, we can express $\innprod{f_1}{\mathfrak{h}^{\textsf{aux}}_{\eta} f_2}{\Lpspc{2}{}{\calD}}$ as:
\begin{align}
\nonumber
&\int_{\calD} f_1 \overline{(-\half \Delta f_2)} \,d\bfx + \int_{\calD} f_1 \overline{\bigg(-i\frac{2\pi\eta}{\tau}\pd{{f_2}}{x_3}\bigg)}\,d\bfx +  \int_{\calD}f_1 \overline{\bigg(-\frac{2\pi^2\eta^2}{\tau^2} f_2\bigg)} \,d\bfx + \int_{\calD}f_1  \overline{\big(V f_2\big)}\,d\bfx\;\\
&= -\half \int_{\calD} f_1  \Delta \overline{f_2}\,d\bfx +  i\frac{2\pi\eta}{\tau}\int_{\calD} f_1 \pd{\overline{f_2}}{x_3}\,d\bfx - \frac{2\pi^2\eta^2}{\tau^2}\int_{\calD}f_1 \overline{f_2} \,d\bfx + \int_{\calD} V f_1  \overline{f_2}\,d\bfx\;.
\label{eq:symmetric_Bloch_hamiltonian_17}
\end{align}
Integrating by parts the first term on the right-hand side of eq.~\ref{eq:symmetric_Bloch_hamiltonian_17}, we have:
\begin{align}
\label{eq:symmetric_Bloch_hamiltonian_18}
-\half \int_{\calD}{f_1}\Delta \overline{f_2} \,d\bfx = \half \bigg(\int_{\calD} \nabla {f_1}\cdot \nabla \overline{f_2} \,d\bfx - \int_{\partial \calD} f_1 \nabla \overline{f_2} \cdot d\bfs\bigg)\,.
\end{align}
The second term on the right-hand side of eq.~\ref{eq:symmetric_Bloch_hamiltonian_18} can be sent to zero by application of the boundary conditions (using a procedure similar to the one outlined in eqs.~\ref{eq:symmetric_Bloch_hamiltonian_3} - \ref{eq:symmetric_Bloch_hamiltonian_7}). This leaves us with:
\begin{align}
\nonumber
&\innprod{f_1}{\mathfrak{h}^{\textsf{aux}}_{\eta} f_2}{\Lpspc{2}{}{\calD}} \\\nonumber &=
\half \int_{\calD} \nabla {f_1}\cdot \nabla \overline{f_2} + i\frac{2\pi\eta}{\tau}\int_{\calD} f_1 \pd{\overline{f_2}}{x_3}\,d\bfx -\frac{2\pi^2\eta^2}{\tau^2}\int_{\calD}f_1 \overline{f_2} \,d\bfx + \int_{\calD} V f_1  \overline{f_2}\,d\bfx\;\\
&= \innprod{\mathfrak{h}^{\textsf{aux}}_{\eta} f_1}{f_2}{\Lpspc{2}{}{\calD}}\,, \label{eq:symmetric_Bloch_hamiltonian_19}
\end{align}
which implies that the operator $\mathfrak{h}^{\textsf{aux}}_{\eta}$ is symmetric on smooth functions obeying the  boundary conditions outlined above. Since such functions are dense in the domain of $\mathfrak{h}^{\textsf{aux}}_{\eta}$ (the boundary conditions being interpreted in the trace sense in that case), the symmetry of the operator follows. 
\section{Direct integral decomposition of the single electron Hamiltonian}
\label{appendix:direct_integrals}
The formalism of direct integrals \citep{Reed_Simon4, Garrett_Spectral_Theorem} generalizes the idea of direct sums in Hilbert spaces and it allows us to make the idea of diagonalizing or block-diagonalizing an unbounded self-adjoint operator mathematically precise. Here we provide a brief summary of some of the key ideas associated with direct integrals. We then demonstrate how the helical  Bloch-Floquet transform can be employed to effectively ``block-diagonalize'' the single electron Hamiltonian in the sense of direct integrals. This can be viewed as a natural extension of the block-diagonal decomposition of the single-electron Hamiltonian in the sense of direct sums, that applies when structures associated with finite symmetry groups are considered \citep{My_PhD_Thesis, banerjee2016cyclic}.

If $\calH'$ is a (separable) Hilbert space and $(M, \mu)$ is a ($\sigma$-finite) measure space, then the Hilbert space $\calH = \mathsf{L}^2(M, d\mu;\calH')$ of $\calH'$ valued functions which are square integrable (against the measure $\mu$), is defined to be a (constant fiber) direct integral\footnote{This definition subsumes the notion of direct sums: if $\mu$ is a sum of point measures at a finite set of points $s_1,s_2,\ldots,s_K$, then any $f \in \mathsf{L}^2(M, d\mu;\calH')$ is determined by the set of $K$ values $\{f(s_1),f(s_2),\ldots,f(s_K)\}$. Thus, $\calH = \mathsf{L}^2(M, d\mu;\calH')$ is isomorphic to the direct sum $\displaystyle\bigoplus_{i=1}^{K} \calH'$.}, and we denote this relationship as:
\begin{align}
\label{eq:calH}
\calH = \int_{M}^{\oplus} \calH'\,d\mu\,.
\end{align}
Vector addition and scalar multiplication are defined pointwise in this space, i.e., for $s \in M$, $f_1,f_2 \in \calH$ and $z\in \cz$, we have $\big(f_1 + f_2\big)(s) = f_1(s) + f_2(s)$ and $\big(z\,f\big)(s)=z\,f(s)$, while the inner product is defined as:
\begin{align}
\label{eq:innprod_direct_integral}
\innprod{f_1}{f_2}{\calH} = \int_{M}\innprod{f_1(s)}{f_2(s)}{\calH'}\,d\mu\,. 
\end{align}
The above definitions allows us to decompose operators on $\calH$ in terms of operators on $\calH'$ in a particular sense. Let $\calL(\calH')$ denote the set of the bounded linear operators on the space $\calH'$. A bounded linear operator $\calA$ on the space $\displaystyle \calH \bigg(= \int_{M}^{\oplus} \calH'\,d\mu \bigg)$ is said to be decomposed in the sense of a direct integral decomposition if there exists a function $A(\cdot)$ in $\mathsf{L}^{\infty}\big(M, d\mu;\calL(\calH')\big)$ such that\footnote{A function $A(\cdot)$ from $M$ to $\calL(\calH')$ is called measurable if for every $f_1,f_2 \in \calH'$, the map $s \mapsto \innprod{f_1}{A(s) f_2}{\calH'}$ is measurable \citep{Reed_Simon4}.} for all $f \in \calH$, the relationship $(\calA f)(s) = A(s)f(s)$ holds. We then call $\calA$ \textit{decomposable}, we refer to the operators ${A}(s)$ as the \textit{fibers} of $\calA$, and we denote this relationship as:
\begin{align}
\calA = \int_{M}^{\oplus} A(s)\,d\mu\,.
\end{align}

Conversely,  it is also possible to ``build'' operators on the space $\calH$, starting from operators on the space $\calH'$. Specifically, given a measurable function $A(\cdot)$ from $M$ to the set of (bounded or unbounded) self-adjoint operators\footnote{A function $A(\cdot)$ from $M$ to the set of (bounded or unbounded) self-adjoint operators on $\calH'$ is called measurable if the function $(A(\cdot) + i)^{-1}$ is measurable \citep{Reed_Simon4}.} the on $\calH'$, we may define an operator $\calA$ on $\calH$ with domain:
\begin{align}
\label{eq:domainA}
\text{Dom.}(\calA) = \bigg\{f\in \calH: f(s) \in \text{Dom.}(A(s))\;\text{a.e.}; \int_{M} \norm{A(s)f(s)}{\calH'}^2\,d\mu < \infty \bigg \}\,,
\end{align}
as:
\begin{align}
\label{eq:operatorA}
(\calA f)(s) = A(s)f(s)\,.
\end{align}
As before, we will use the notation:
\begin{align}
\label{eq:operatorA_direct_integral}
\calA = \int_{M}^{\oplus} A(s)\,d\mu\,,
\end{align}
to denote the above relationship between the operator $A$ and its fibers $A(s)$.

With these definitions in place, we now apply the above apparatus to carry out a suitable decomposition of the single-electron Hamiltonian associated with a helical structure. We will consider the case of a structure associated with a helical group $\calG_1$ that is generated by a single element $\Upsilon_\mathsf{h}$. As discussed in \ref{subsec:Kohn_Sham_Helical}, the single-electron Hamiltonian in this case is the operator $\hamil = - \half \Delta + V(\bfx)$ over the space $\widetilde{\calH} = \Lpspc{2}{}{\calC}$. The relevant boundary condition is $\psi(\bfx) = 0$ for $\bfx \in \partial\calC$ and the potential $V(\bfx)$ is group invariant, i.e., $V(\bfx) = V(\widetilde{\Upsilon} \circ \bfx)$ for every $\widetilde{\Upsilon} \in \calG_1$. 

We may choose $M = [-\half, \half) = \mathfrak{I}$, $\calH' = \Lpspc{2}{}{\calD}$, $\mu$ as the Lebesgue measure on $\mathfrak{I}$ (denoted as $d\eta$ henceforth). Then,
\begin{align}
\calH = \int_{\mathfrak{I}}^{\oplus} \calH'\,d\eta\,.
\label{eq:calH_eta}
\end{align}
is the space $\Lpspc{2}{}{\calD \times \mathfrak{I}}$. The helical Bloch-Floquet transform $\calU$ introduced in eq.~\ref{eq:operatorU} allows us to map functions and operators in between the spaces $\widetilde{\calH} = \Lpspc{2}{}{\calC}$ and $\calH =  \Lpspc{2}{}{\calD \times \mathfrak{I}}$. Specifically, for each $\eta \in \mathfrak{I}$, let $\hamil_{\eta}$ be the restriction of $\hamil$ to $\calD$ along with the boundary conditions $\psi(\Upsilon_{\mathsf{h}} \circ \bfx)=e^{-i2\pi\eta}\,\psi(\bfx)$, $\bfR_{2\pi\alpha}^{-1}\nabla\psi(\Upsilon_{\mathsf{h}} \circ \bfx)=e^{-i2\pi\eta}\,\nabla\psi(\bfx)$ for $\bfx \in \partial\calD^{\,z=0}$ and, $\psi(\bfx) = 0$ for $\bfx \in \partial\calD^{\,r=R}$. Then, the operator $\hamil$ can be decomposed into the set of operators $\big\{\hamil_{\eta}\big\}_{\eta \in \mathfrak{I}}$ (referred to as the fibers of $\hamil$) using the direct integral decomposition, in the sense:
\begin{align}
\label{eq:Bloch_block_diagonal}
\calU\,\hamil\,\calU^{-1} =  \int_{\mathfrak{I}}^{\oplus} \hamil_{\eta}\,d\eta\,,
\end{align}
To demonstrate this result, we will establish direct integral decompositions of the Laplacian and the potential operator individually as operators on $\Lpspc{2}{}{\calC}$, i.e., we will show successively that:
\begin{align}
\label{eq:Bloch_block_diagonal_Laplacian}
\calU\,\big(-\Delta\big)\,\calU^{-1} =  \int_{\mathfrak{I}}^{\oplus} (-\Delta)_{\eta}\,d\eta\,,
\end{align}
and,
\begin{align}
\label{eq:Bloch_block_diagonal_Potential}
\calU\, V \,\calU^{-1} =  \int_{\mathfrak{I}}^{\oplus} V_{\eta}\,d\eta\,.
\end{align}
The sought result will then follow from Theorem XIII.85g of \citep{Reed_Simon4}.

{To perform the direct integral decomposition of the Laplacian on $\Lpspc{2}{}{\calC}$, subject to zero Dirichlet boundary condition for $\bfx \in \calC$, we first define the fibers $(-\Delta)_{\eta}$ as the restriction of $-\Delta$ to $\calD$ along with the boundary conditions $p(\Upsilon_{\mathsf{h}} \circ \bfx)=e^{-i2\pi\eta}\,p(\bfx)$, $\bfR_{2\pi\alpha}^{-1}\nabla p(\Upsilon_{\mathsf{h}} \circ \bfx)=e^{-i2\pi\eta}\,\nabla p(\bfx)$ for $\bfx \in \partial\calD^{\,z=0}$ and, $p(\bfx) = 0$ for $\bfx \in \partial\calD^{\,r=R}$. Now, let $\calB$ be the operator on the right-hand side of eq.~\ref{eq:Bloch_block_diagonal_Laplacian} acting on the space $\calH =  \Lpspc{2}{}{\calD \times \mathfrak{I}}$, and let $f$ be a Schwartz class function that obeys $f(\bfx) = 0$ for $\bfx \in \calC$. Note that since $f$ is a Schwartz class function, so are all its derivatives. With this setup, it suffices to show that $\calU f \in \text{Dom.}(\calB)$, and that $\calU(-\Delta f) = \calB(\calU f)$, in order to establish eq.~\ref{eq:Bloch_block_diagonal_Laplacian} (see e.g. Theorem XIII.87 and the lemma following that theorem in \citep{Reed_Simon4}).

By definition, $\calU f$ is given as:
\begin{align}
\label{eq:operatorU_again}
(\calU f)(\bfx,\eta) = \sum_{m\in \gz} {f}(\Upsilon^m_{\mathsf{h}} \circ \bfx)\,e^{i 2 \pi m\eta}\,,
\end{align}
and due to the properties of $f$, it is a smooth function in $\bfx$ over $\calD$ that obeys $(\calU f)(\bfx,\eta) = 0$ for $\bfx \in \partial\calD^{\,r=R}$ for any $\eta \in \mathfrak{I}$. Furthermore, by evaluating the above expression at $\Upsilon_{\mathsf{h}} \circ \bfx$, it is clear that (also see Footnote  \ref{footnote:U_translation}),the relationship:  
\begin{align}
\label{eq:calU_translation}
(\calU f)(\Upsilon_{\mathsf{h}} \circ \bfx, \eta)=e^{-i2\pi\eta}\,(\calU f)(\bfx, \eta)
\end{align}
holds. Computing the gradient on both sides, we get:
\begin{align}
\bfR_{2\pi\alpha}^{-1}\nabla(\calU f)(\Upsilon_{\mathsf{h}} \circ \bfx, \eta)=e^{-i2\pi\eta}\,\nabla(\calU f)(\bfx, \eta)\,.
\label{eq:grad_calU_translation}
\end{align}
Now, evaluating the above expressions at $\bfx \in \partial\calD^{\,z=0}$, we see that $(\calU f)(\bfx,\eta)$ obeys all the boundary conditions necessary for it to be in $\text{Dom.}(\calB)$. Next, we recall that the Laplacian is invariant under isometries\footnote{{Specifically, if $\Upsilon$ is an isometry on $\rz^3$, $u(\bfx)$ is a $\mathsf{C}^2$ function and $v(\bfx) = u(\Upsilon \circ \bfx)$, then it holds that $\Delta v(\bfx) = \Delta u|_{\Upsilon \circ \bfx}$.}}. Consequently, it holds that:
\begin{align}
\nonumber
\Delta\big((\calU f)(\bfx,\eta)\big) &= \sum_{m\in \gz} \Delta \big({f}(\Upsilon^m_{\mathsf{h}} \circ \bfx)\big)\,e^{i 2 \pi m\eta} \\
&= \sum_{m\in \gz} \Delta{f}\big|_{\Upsilon^m_{\mathsf{h}} \circ \bfx}\,e^{i 2 \pi m\eta} = \big(\calU(\Delta f)\big)(\bfx,\eta)\,,
\end{align}
and so, $\calU(-\Delta f) = \calB(\calU f)$, as required. This establishes\footnote{{We would like to thank an anonymous reviewer for helping fix certain technical aspects of this presentation.}}  eq.~\ref{eq:Bloch_block_diagonal_Laplacian}.}

Next, to establish \ref{eq:Bloch_block_diagonal_Potential}, we set the fibered operator $V_{\eta}$ on $\Lpspc{2}{}{\calD}$ as:
\begin{align}
(V_{\eta}f)(\bfx) = V(\bfx)f(\bfx)\,.
\end{align}
Then for any Schwartz class function $f$ over $\calC$, we have that:
\begin{align}
\calU\,(Vf)(\bfx;\eta) = \sum_{m\in \gz} {V}(\Upsilon^m_{\mathsf{h}} \circ \bfx)\,{f}(\Upsilon^m_{\mathsf{h}} \circ \bfx)\,e^{i 2 \pi m\eta}\,,
\end{align}
which, using the invariance of the potential under $\calG_{1}$, gives:
\begin{align}
\calU\,(Vf)(\bfx;\eta) = {V}(\bfx)\sum_{m\in \gz}{f}(\Upsilon^m_{\mathsf{h}} \circ \bfx)\,e^{i 2 \pi m\eta} = {V}(\bfx)\,\calU(\bfx;\eta) = V_{\eta}(\calU\,f)(\bfx,\eta)\,.
\end{align}
Thus \ref{eq:Bloch_block_diagonal_Potential} is established.

With \ref{eq:Bloch_block_diagonal} established, some conclusions can be immediately drawn regarding the structure of the spectrum of $\hamil$. Using Theorem XIII.85 in \citep{Reed_Simon4} for example, it follows that the collection of eigenvalues of the operators $\hamil_{\eta}$ together form the spectrum of $\hamil$, i.e., $\Lambda = \mathsf{spec.}(\hamil)$. Furthermore, Theorem XIII.86 of \citep{Reed_Simon4} implies that $\hamil$ has a purely continuous spectrum.

Finally, a result made use of in Section \ref{subsubsec:Governing_Equations} is that if $\calA$ is an operator that is invariant under $\calG_1$, i.e., it commutes with the unitary operators in the set: \begin{align}
\mathcal{T} = \big\{T_{\widetilde{\Upsilon}}: T_{\widetilde{\Upsilon}}f(\bfx) = f(\widetilde{\Upsilon}^{-1}\circ\bfx)\big\}_{\widetilde{\Upsilon} \in \calG_1}\,,
\end{align}
and it is is locally trace-class, we may assign meaning to the trace per unit fundamental domain as:
\begin{align}
\underline{\text{Tr.}}[\calA] = \int_{\mathfrak{I}}\!\text{Tr.}[\calA_{\eta}]\,d\eta\,.
\end{align}
Discussion and rigorous proofs of this result appear in \citep{de2011exponentially, panati2009geometric} for the case of periodic symmetries. To see why this result also holds true for the case of a helical symmetry group, we denote $\mathbb{I}_{\calD}$ as the indicator function of the fundamental domain. Let $\{f_j:j\in \nz \}$ be an orthonormal basis of $\Lpspc{2}{}{\calC}$. Due to the fact that $\calA$ is invariant under the group $\calG_1$, we may write:
\begin{align}
\calU\,\calA\,\calU^{-1} =  \int_{\mathfrak{I}}^{\oplus} \calA_{\eta}\,d\eta\,,
\label{eq:direct_integral_A_again}
\end{align}
Now, the trace per unit fundamental domain is:
\begin{align}
\underline{\text{Tr.}}[\calA] = \sum_{j=1}^{\infty}\innprod{\mathbb{I}_{\calD}\,\calA\,\mathbb{I}_{\calD}\,f_j}{f_j}{\Lpspc{2}{}{\calC}}\, = \sum_{j=1}^{\infty}\innprod{\calA\,\mathbb{I}_{\calD}\,f_j}{\mathbb{I}_{\calD}\,f_j}{\Lpspc{2}{}{\calC}}
\end{align}
Using eq.~\ref{eq:direct_integral_A_again}, this becomes:
\begin{align}
\nonumber
\underline{\text{Tr.}}[\calA] &= \sum_{j=1}^{\infty}\bigg\langle{\calU^{-1}\bigg( \int_{\mathfrak{I}}^{\oplus} \calA_{\eta}\,d\eta\bigg)\calU\,\mathbb{I}_{\calD}\,f_j},{\mathbb{I}_{\calD}\,f_j}\bigg\rangle_{\Lpspc{2}{}{\calC}}\\
&= \sum_{j=1}^{\infty}\bigg\langle{\bigg( \int_{\mathfrak{I}}^{\oplus} \calA_{\eta}\,d\eta\bigg)\calU\,\mathbb{I}_{\calD}\,f_j},{\calU\,\mathbb{I}_{\calD}\,f_j}\bigg\rangle_{\Lpspc{2}{}{\calC}}\,,
\end{align}
which, upon using \ref{eq:operatorA} and \ref{eq:innprod_direct_integral}, can be written as:
\begin{align}
\underline{\text{Tr.}}[\calA] = \sum_{j=1}^{\infty} \int_{\mathfrak{I}}\innprod{\calA_{\eta}\,\calU\,\mathbb{I}_{\calD}\,f_j}{\calU\,\mathbb{I}_{\calD}\,f_j}{\Lpspc{2}{}{\calD}}\,d\eta
\end{align}
Now, recognizing that the indicator function $\mathbb{I}_{\calD}$ acts as a projection operator on $\Lpspc{2}{}{\calC}$, we see that $\calU\,\mathbb{I}_{\calD}\,f_j$ are simply basis functions of 
$\Lpspc{2}{}{\calD}$ for every $\eta \in \mathfrak{I}$. Denoting $\tilde{f}_j(\bfx;\eta) = (\calU\,\mathbb{I}_{\calD}\,f_j)(\bfx;\eta)$ and exchanging the summation and the integral, the above equation can be re-written as:
\begin{align}
\underline{\text{Tr.}}[\calA] =   \int_{\mathfrak{I}}  \sum_{j=1}^{\infty} \innprod{\calA_{\eta}\tilde{f}_j(\cdot; \eta)}{\tilde{f}_j(\cdot; \eta)}{\Lpspc{2}{}{\calD}} \,d\eta = \int_{\mathfrak{I}}\!\text{Tr.}[\calA_{\eta}]\,d\eta\,,
\end{align}
as claimed.
\section{Helical structure associated with a group generated by two elements: Expressions for important physical quantities and governing equations}
\label{appendix:two_element_group_expressions}
For a helical structure associated with a group $\calG_2$ that is generated by two group elements (i.e., a screw transformation of the form $(\bfR_{2\pi m \alpha} | m\tau\bfe_3)$ and a pure rotation of the form $(\bfR_{n\Theta} | \mathbf{0})$ with $\Theta = \frac{2\pi}{\mathfrak{N}}$, and both rotations $\bfR_{2\pi m \alpha}, \bfR_{n\Theta}$ having $\bfe_3$ as the common axis) the form of the governing equations as well as the expressions for the various quantities of interest can be easily deduced by using the discussion in Section \ref{subsec:Kohn_Sham_Helical} as a starting point, and using the following rules of substitution. The helical Bloch phase factors are of the form $\displaystyle e^{-i2\pi (m\eta + \frac{n\nu}{\mathfrak{N}})}$ with $n,\nu \in \{0,1,\ldots,\mathfrak{N}-1\}$ and $\eta \in \mathfrak{I}$, instead of being $e^{-i2\pi m\eta}$ only. Consequently, helical Bloch states and the helical bands are labeled as $\psi_j(\bfx;\eta,\nu)$ and $\lambda_j(\eta,\nu)$, respectively. Integrals over  $\eta \in \mathfrak{I} = [-\half,\half)$ need to be replaced with integrals over $\eta \in \mathfrak{I}$ \textit{and}  summations over $\nu$ of the form $\displaystyle \frac{1}{\mathfrak{N}}\sum_{\nu = 0}^{\mathfrak{N}-1}$. Taking into account the effect of all the symmetry operations of the group (i.e., computing the group orbit of a point for example) amounts to summing over all group elements of the form $\Upsilon^{m}_{\mathsf{h}}\bullet\Upsilon^{n}_{\mathsf{c}} =  (\bfR_{2\pi m \alpha + n\Theta} | m\tau\bfe_3)$ with $m \in \gz$ and $n \in \{0,1,\ldots,\mathfrak{N}-1\}$. Also, in the expression for the forces (eq.~\ref{eq:force_expression_2}), the rotation matrix $(\bfR_{2\pi m \alpha})^{-1}$ needs to be replaced with $(\bfR_{2\pi m \alpha + n\Theta})^{-1}$. Finally, spatial integrals over the fundamental domain $\calD_{\calG_1}$ (or its interior $\calD$) have to be replaced by integrals over the fundamental domain $\calD_{\calG_2}$ (or its interior $\widetilde{\calD}$). The reciprocal space (or more specifically, the Brillouin zone of the reciprocal space) for the symmetry group $\calG_2$ will be denoted by the set $\mathfrak{B} = \displaystyle\mathfrak{I} \times \{0,{1},2,\ldots,{\mathfrak{N}-1}\}$.

The mathematical reasons behind the above substitution rules are as follows: The structure of the helical group generated by two elements is such that it can be expressed as the direct product of 
the helical group generated by a single element and a cyclic group. Therefore, the Bloch theorem for a structure associated with such a symmetry group can be established by first using Theorem \ref{Thm:Existence_Helical_Bloch}, and then using an appropriate version\footnote{Theorem 2.6 from \citep{banerjee2016cyclic} was established for the single electron Hamiltonian operator. However, it can be easily extended to any other linear elliptic self-adjoint operator that commutes the cyclic symmetries of the system, e.g. $\mathfrak{h}^{\text{aux}}_{\eta}$ defined in eq.~\ref{eq:phi_equation_1}, with $V(\bfx)$ for that operator invariant under $\calG_2$.} of Theorem 2.6 from \citep{banerjee2016cyclic}. The characters (i.e., one-dimensional complex irreducible representations) of the helical group generated by two elements can be obtained by multiplying out the  characters of the helical group generated by a single element and the cyclic group. This leads to the helical Bloch states in this case to obey the condition:
\begin{align}
\label{eq:helical_Bloch_two_generator}
\psi_j\big((\Upsilon^{m}_{\mathsf{h}}\bullet\Upsilon^{n}_{\mathsf{c}})\circ\bfx;\eta,\nu\big) = e^{-i2\pi (m\eta + \frac{n\nu}{\mathfrak{N}})} \psi_j(\bfx;\eta,\nu)\,.
\end{align}
The completeness of these states follows by a combination of Theorem \ref{Thm:Completeness_Helical_Bloch} and the properties of the Peter-Weyl projectors \citep{Folland_Harmonic, Barut_Reps, banerjee2016cyclic} for the cyclic group. This completeness result paves the way for a suitable Helical Bloch-Floquet transform $\calU:\Lpspc{2}{}{\calC} \to \Lpspc{2}{}{\widetilde{\calD}\times\mathfrak{B}}$:
\begin{align}
\label{eq:operatorU_G2}
(\calU f)(\bfx,\eta, \nu)=\sum_{m\in \gz} \bigg(\frac{1}{\mathfrak{N}}\sum_{n = 0}^{\mathfrak{N}-1}{f}\big((\Upsilon^{m}_{\mathsf{h}}\bullet\Upsilon^{n}_{\mathsf{c}})\circ\bfx\big)\,e^{i 2 \pi (m\eta + \frac{n\nu}{\mathfrak{N}})}\,\bigg)\,,
\end{align}
as well as a direct integral representation of the single electron Hamiltonian. From these, the key quantities of interest and the governing equations may be systematically deduced (as demonstrated for the case of the helical group generated by a single element in Section \ref{subsec:Kohn_Sham_Helical}).

Keeping the above discussion in mind, we now present the final mathematical expressions for key quantities of interest for a helical structure generated by $\calG_2$. For the sake of definiteness, we assume that the helical structure, is embedded in the infinite cylinder $\calC$ (eq.~\ref{eq:Infinite_Cylinder}), and as in eq.~\ref{eq:Cylindrical_Sector}, we let $\calD_{\calG_2} = \big\{(r,\vartheta,z):0 \leq r<R , 0 \leq \vartheta < \frac{2\pi}{\mathfrak{N}} , 0 \leq z < \tau\big\}$ denote a cylindrical sector that acts as the fundamental domain of the helical structure. Let $\widetilde{D}$ denote the interior of this fundamental domain. Furthermore, let the positions of the atoms within the fundamental domain be $\calP_{\calG_2} = \big \{\bfx_k \big\}_{k=1}^{M_{\calG_2}}$ and let $\big\{b_k(\bfx;\bfx_k)\big\}_{k=1}^{M_{\calG_2}}$ denote the corresponding nuclear pseudocharges. The helical Bloch states presented above in eq.~\ref{eq:helical_Bloch_two_generator} allow the single electron problem for the entire helical structure to be reduced to the fundamental domain $\calD_{\calG_2}$ (or its interior  $\widetilde{D}$).

Let $g_j(\eta,\nu) = f_{T_{\text{e}}}\big(\lambda_j(\eta, \nu)\big)$ denote the electronic occupation numbers at electronic temperature $T_{\text{e}}$. The electron density for $\bfx \in \widetilde{D}$ can be expressed as:
\begin{align}
\rho(\bfx) = 2\int_{\mathfrak{I}}\frac{1}{\mathfrak{N}}\sum_{\nu = 0}^{\mathfrak{N}-1}\sum_{j=1}^{\infty}\,g_j(\eta,\nu)\,\lvert \psi_j(\bfx; \eta, \nu)\rvert^2\,d\eta\,.
\label{eq:electron_density_G2}
\end{align}
The electron density needs to obey the constraint of integrating to a fixed number of electrons in the fundamental domain, i.e.,
\begin{align}
\int_{\widetilde{\calD}} \rho(\bfx)\,d\bfx = N_{\text{e}}\, \implies 2\int_{\mathfrak{I}}\frac{1}{\mathfrak{N}}\sum_{\nu = 0}^{\mathfrak{N}-1}\sum_{j=1}^{\infty} g_j(\eta, \nu) = N_{\text{e}}\,.
\label{eq:occupation_constraint_G2}
\end{align}
The electronic free energy per unit fundamental domain can be expressed as:
\begin{align}
\calF(\Lambda, \Psi, \calP_{\calG_2}, \widetilde{D}, \calG_2) = T_{\text{s}}&(\Lambda, \Psi,\widetilde{D}, \calG_2) \,+\,E_{\text{xc}}(\rho,\widetilde{D})\,+\,K(\Lambda, \Psi, \calP_{\calG_2}, \widetilde{D}, \calG_2)\\
&+ E_{\text{el}}(\rho, \calP_{\calG_2}, \widetilde{D}, \calG_2)\,-\,T_{\text{e}}\,S(\Lambda)\,.
\label{eq:Free_Energy_G2}
\end{align}
The terms on the right-hand side of the above equation are as follows. The first term is the kinetic energy of the electrons per unit fundamental domain and is expressible as:
\begin{align}
T_{\text{s}}(\Lambda, \Psi,\widetilde{\calD}, \calG_2) = \int_{\mathfrak{I}}\frac{1}{\mathfrak{N}}\sum_{\nu = 0}^{\mathfrak{N}-1}\bigg(\sum_{j=1}^{\infty}g_j(\eta,\nu)\big\langle{\Delta \psi_j(\cdot; \eta,\nu)},{\psi_j(\cdot; \eta,\nu)}\big\rangle_{\Lpspc{2}{}{\widetilde{D}}}\,\bigg)\,d\eta\,.
\label{eq:KE_G2}
\end{align}
The second term is the exchange correlation energy per unit fundamental domain and is expressible as:
\begin{align}
E_{\text{xc}}(\rho, \widetilde{\calD}) = \int_{\widetilde{\calD}}\varepsilon_{\text{xc}}[\rho(\bfx)]\,\rho(\bfx)\,d\bfx\,.
\end{align}
The third term is the nonlocal pseudopotential energy per unit fundamental domain and is expressible as:
\begin{align}
\nonumber
&K(\Lambda, \Psi, \calP_{\calG_2}, \widetilde{\calD}, \calG_2)\\
&= 2\,\sum_{k=1}^{M_{\calG_2}}\sum_{p \in \calN_k}\sum_{j=1}^{\infty}\gamma_{k,p}\int_{\mathfrak{I}}\frac{1}{\mathfrak{N}}\sum_{\nu = 0}^{\mathfrak{N}-1}\bigg(g_j(\eta, \nu)\bigg\lvert{\big\langle{\hat{\chi}_{k,p}(\cdot;\eta, \nu;\bfx_k)},{\psi_j(\cdot; \eta, \nu)}\big\rangle}_{\Lpspc{2}{}{\widetilde{\calD}}}\bigg\rvert^2\bigg)\,d\eta\,,
\end{align}
with the functions $\hat{\chi}_{k,p}(\bfx;\eta, \nu;\bfx_k) $ being related to the helical Bloch Floquet transform associated with the group $\calG_2$ (eq.~\ref{eq:operatorU_G2}) i.e.,:
\begin{align}
\hat{\chi}_{k,p}(\bfx;\eta, \nu;\bfx_k) = \sum_{m\in \gz}\sum_{n=0}^{\mathfrak{N}-1}{\chi}_{k;p}\big( (\Upsilon^{m}_{\mathsf{h}}\bullet\Upsilon^{n}_{\mathsf{c}})\circ \bfx;\bfx_k\big)\,e^{i2\pi (m\eta + \frac{n\nu}{\mathfrak{N}})}\,.
\end{align}
Here, as in eq.~\ref{eq:Kleinman_Bylander_form}, ${\chi}_{k;p}(\cdot;\bfp_k)$ denote the atom centered nonlocal pseudopotential operator projection functions.
The fourth term on the right-hand side of eq.~\ref{eq:Free_Energy_G2} is the net electrostatic interaction energy per unit fundamental domain, and can be written as:
\begin{align}
\nonumber
E_{\text{el}}(\rho, \calP_{\calG_2}, \widetilde{\calD}, \calG_2) =  \max_{\Phi}\bigg\{-&\frac{1}{8\pi}\int_{\widetilde{\calD}}\abs{\nabla \Phi}^2\,d\bfx + \int_{\widetilde{\calD}}\bigg(\rho(\bfx) + b(\bfx, ,\calP_{\calG_2},\calG_2)\bigg)\Phi(\bfx)\,d\bfx\bigg\}\\ 
&+ E_{\text{sc}}(\calP_{\calG_2},\calG_2, \widetilde{\calD})\,.
\label{eq:ES_G2}
\end{align}
Here, $b(\bfx, ,\calP_{\calG_2},\calG_2)$ is the net nuclear pseudocharge:
\begin{align}
b(\bfx,\calP_{\calG_2},\calG_2) = \sum_{m \in \gz}\sum_{n = 0}^{\mathfrak{N}-1}\sum_{k = 1}^{M_{\calG_2}}b_{k}\big(\bfx; (\Upsilon^{m}_{\mathsf{h}}\bullet\Upsilon^{n}_{\mathsf{c}})\circ \bfx_k\big)\,,
\label{eq:Net_pseudocharge_G2}
\end{align}
and $E_{\text{sc}}(\calP_{\calG_2},\calG_2, \widetilde{\calD})$ represents self interaction and correction terms. The net electrostatic potential $\Phi$ can be expressed as the Newtonian potential:\footnote{As in the main text (Section \ref{subsubsec:Governing_Equations}), we have, with a minor abuse of notation, used $\Phi$ to denote the ``trial''  electrostatic potentials involved in the maximization problem in eq.~\ref{eq:ES_G2} as well as the actual potential that achieves this maximum, i.e., eq.~\ref{eq:ES_phi_formula_finite_G2}.}
\begin{align}
\Phi(\bfx) = \int_{\calC} \frac{\rho(\bfy) +b(\bfy,\calP_{\calG_2},\calG_2)}{\norm{\bfx - \bfy}{\rz^3}}\,d\bfy\,.
\label{eq:ES_phi_formula_finite_G2} 
\end{align}
The last term on the right-hand side of eq.~\ref{eq:Free_Energy_G2} is the electronic entropy  contribution:
\begin{align}\nonumber
&S(\Lambda) \\&= -2\,k_{\text{B}}\int_{\mathfrak{I}}\frac{1}{\mathfrak{N}}\sum_{\nu=0}^{\mathfrak{N}-1}\bigg[\sum_{j=1}^{\infty}g_j(\eta, \nu)\,\log\big(g_j(\eta, \nu)\big) + \big(1 - g_j(\eta, \nu)\big)\,\log\big(1 - g_j(\eta, \nu)\big)\bigg]\,d\eta
\end{align}

The symmetry adapted Kohn-Sham equations for the system, as posed over the fundamental domain are:
\begin{align}
\hamil^{\text{KS}}_{[\Lambda, \Psi, \calP_{\calG_2}, \widetilde{D}, \calG_2]}\,\psi_j(\cdot;\eta, \nu) = \lambda_j(\eta, \nu)\,\psi_j(\cdot;\eta, \nu)\,,
\label{eq:Euler_Lagrange_G2}
\end{align}
for $j\in\nz,\eta \in\mathfrak{I}, \nu \in \{0,1,\ldots,\mathfrak{N}-1\}$. Here,
\begin{align}
\hamil^{\text{KS}}_{[\Lambda, \Psi, \calP_{\calG_2}, \widetilde{D}, \calG_2]}\equiv -\half \Delta + V_{\text{xc}} + \Phi + \widehat{\calV}^{\text{nl}}_{\widetilde{\calD}}\,,
\label{eq:Euler_Lagrange_G2_2}
\end{align}
is the Kohn-Sham operator, with its dependence on the helical Bloch sates, the helical bands, etc., made explicit\footnote{As in the discussion in Section \ref{subsubsec:Governing_Equations}, we may change notation and express the free energy per unit cell in terms of the occupations as $\widetilde{\calF}(\mathfrak{G}, \Psi, \calP_{\calG_2}, \widetilde{\calD}, \calG_2) = \calF(\Lambda, \Psi, \calP_{\calG_2}, \widetilde{\calD}, \calG_2)$, instead of using the eigenvalues. With this change in notation, the problem of determination of the Kohn-Sham ground state for a given helical structure, can be expressed as: 
\begin{align}
\widetilde{\calF}_{0}(\calP_{\calG_2}, \widetilde{\calD}, \calG_2) = \text{inf.}_{{\Psi,\mathfrak{G}}}\,\widetilde{\calF}(\mathfrak{G}, \Psi, \calP_{\calG_2}, \widetilde{\calD}, \calG_2)\,,
\end{align}
subject to the condition that $\Psi$ consists of helical Bloch states and eq.~\ref{eq:occupation_constraint_G2}. Similarly, the Kohn Sham operator in eqs.~\ref{eq:Euler_Lagrange_G2} and \ref{eq:Euler_Lagrange_G2_2} may be denoted as as $\hamil^{\text{KS}}_{[\mathfrak{G}, \Psi, \calP_{\calG_2}, \widetilde{D}, \calG_2]}$.} . As before, $ V_{\text{xc}}$ represents the exchange correlation potential, while the nonlocal pseudopotential operator $\widehat{\calV}^{\text{nl}}_{\widetilde{\calD}}$ has fibers (in coordinate representation):
\begin{align}
\widehat{\calV}^{\text{nl}}_{\widetilde{\calD}}(\bfx,\bfy;\eta,\nu) = \sum_{k=1}^{M_{\calG_2}}\sum_{p \in \calN_k}\gamma_{k,p}\,\hat{\chi}_{k,p}(\bfx;\eta, \nu;\bfx_k)\,\otimes\,\overline{\hat{\chi}_{k,p}(\bfy;\eta, \nu;\bfx_k)}
\end{align}
The helical Bloch states obey (for $i,j\in\nz,$):
\begin{align}
\innprod{\psi_i(\cdot;\eta, \nu)}{\psi_j(\cdot;\eta, \nu)}{\Lpspc{2}{}{\widetilde{\calD}}} = \delta_{i,j}\,.
\end{align}
The net electrostatic potential obeys:
\begin{align}
-\Delta \Phi &= 4\pi\,\big(\rho + b(\cdot,\calP_{\calG_2},\calG_2)\big)\,,
\label{eq:Poisson_Equation_G2}
\end{align}
and is invariant under $\calG_2$. Finally, the Harris-Foulkes energy functional is:
\begin{align}
\nonumber
{\calF}^{\text{HF}}(\Lambda, \Psi, \calP_{\calG_2}, \widetilde{\calD}, \calG_2) &=  2\int_{\mathfrak{I}}\frac{1}{\mathfrak{N}}\sum_{\nu = 0}^{\mathfrak{N}-1}\sum_{j=1}^{\infty}\lambda_j(\eta, \nu)\,g_j(\eta, \nu)\,d\eta + E_{\text{xc}}(\rho,  \widetilde{\calD}) \\\nonumber &-\int_{ \widetilde{\calD}}V_{\text{xc}}(\rho(\bfx))\rho(\bfx)\,d\bfx
+\half\int_{ \widetilde{\calD}}\bigg(b(\bfx, ,\calP_{\calG_2},\calG_2) - \rho(\bfx) \bigg)\Phi(\bfx)\,d\bfx \\&+ E_{\text{sc}}(\calP_{\calG_2},\calG_2, \widetilde{\calD}) - T_{\text{e}}\,S(\Lambda)\,.
\end{align}
and the Hellmann-Feynman forces on the atoms in the fundamental domain are:
\begin{align}
\nonumber
\mathbf{f}_k &=\sum_{m \in \gz}\sum_{n = 0}^{\mathfrak{N}-1}(\bfR_{2\pi m \alpha + n\Theta})^{-1} \int_{\widetilde{\calD}}\nabla b_{k}\big(\bfx;(\Upsilon^{m}_{\mathsf{h}}\bullet\Upsilon^{n}_{\mathsf{c}})\circ \bfx_k \big)\Phi(\bfx)\,d\bfx - \pd{E_{\text{sc}}(\calP_{\calG_2},\calG_2, \widetilde{\calD})}{\bfx_k}\\\nonumber
&-4\sum_{j=1}^{\infty}\Bigg(\int_{\mathfrak{I}} \frac{1}{\mathfrak{N}}\sum_{\nu = 0}^{\mathfrak{N}-1} g_j(\eta, \nu)\sum_{p \in \mathcal{N}_k}\gamma_{k;p}\,\text{Re.}\Bigg\{\bigg[\int_{\widetilde{\calD}} \hat{\chi}_{k,p}(\bfx;\eta, \nu;\bfx_k)\,\overline{\psi_j(\bfx;\eta, \nu)}\,d\bfx\bigg]\\
&\quad\quad\quad\quad\quad\quad\quad\quad\quad\quad\;\,\quad\quad\quad\quad\quad\times\bigg[\int_{\widetilde{\calD}} \psi_j(\bfx;\eta, \nu)\,\overline{\pd{\hat{\chi}_{k,p}(\bfx;\eta, \nu;\bfx_k)}{\bfx_k}}\,d\bfx\bigg]
\Bigg\}\Bigg)\,d\eta
\end{align}
A more comprehensive account of the above equations as well as extensive numerical simulations involving them appear in a follow-up contribution \citep{Helical_DFT_Paper_2}.
\section{Cartesian gradient operator, Laplacian operator and integrals in helical coordinates}
\label{appendix:grad_laplace_integral_helical}
We derive expressions in helical coordinates for the Cartesian gradient operator (useful in expressing forces on the atoms), the Laplacian operator (useful for expressing the Schrodinger operator), and integrals (useful for computing energies) in this Appendix. The calculations involved are straight forward applications of the chain rule and they  originally appear in \citep{My_PhD_Thesis}. We include these here for the sake of completeness.

Let $\xi(x_1,x_2,x_3)$ be a generic scalar quantity expressed in Cartesian coordinates. In helical coordinates, this can be expressed as:
\beq
\label{eqn_for_xi}
\xi(x_1,x_2,x_3)=\xihat \big(r(x_1,x_2,x_3),\theta_1(x_1,x_2,x_3),\theta_2(x_1,x_2,x_3) \big)\,,
\eeq
We wish to evaluate the Cartesian gradient:
\begin{align}
\label{nabla_xi}
\nabla \xi = \pd{\xi}{x_1} \bfe_1 +  \pd{\xi}{x_2} \bfe_2 + \pd{\xi}{x_3} \bfe_3\,,
\end{align}
and the Laplacian:
\begin{align}
\label{Laplacian_xi}
\Delta \xi = \hpd{\xi}{x_1}{2} +  \hpd{\xi}{x_2}{2} + \hpd{\xi}{x_3}{2}\,,
\end{align}
in terms of the function $\xihat$ and the coordinates $r,\theta_1,\theta_2$. Let $i,j=1,\ldots,3$. We then have by the chain rule:
\begin{align}
\label{chain_rule}
\pd{\xi}{x_i}=\pd{\xihat}{r}\pd{r}{x_i}+
\pd{\xihat}{\theta_1}\pd{\theta_1}{x_i}+
\pd{\xihat}{\theta_2}\pd{\theta_2}{x_i}\,,
\end{align}
and further,
\begin{align}
\label{chain_rule_again}
\nonumber
\frac{\partial^2 \xi}{\partial x_i \partial x_j}=& \pd{\xihat}{r}
\frac{\partial^2 r}{\partial x_i \partial x_j}+\pd{\xihat}{\theta_1}
\frac{\partial^2 \theta_1}{\partial x_i \partial x_j}+\pd{\xihat}{\theta_2}
\frac{\partial^2 \theta_2}{\partial x_i \partial x_j}\\
&+\pd{\big(\pd{\xihat}{r}\big)}{x_j}\pd{r}{x_i}
+\pd{\big(\pd{\xihat}{\theta_1}\big)}{x_j}\pd{\theta_1}{x_i}
+\pd{\big(\pd{\xihat}{\theta_2}\big)}{x_j}\pd{\theta_2}{x_i}\,.
\end{align}
The chain rule applied again to the last 3 terms gives us:
\begin{align}
\nonumber
\pd{\big(\pd{\xihat}{r}\big)}{x_j}\pd{r}{x_i}=\pd{r}{x_i}\bigg(
\hpd{\xihat}{r}{2}\pd{r}{x_j}+\frac{\partial^2 \xihat}{\partial r \partial \theta_1}\pd{\theta_1}{x_j}+
\frac{\partial^2 \xihat}{\partial r \partial \theta_2}\pd{\theta_2}{x_j}\bigg)\quad.\\\nonumber
\pd{\big(\pd{\xihat}{\theta_1}\big)}{x_j}\pd{\theta_1}{x_i}=\pd{\theta_1}{x_i}\bigg(
\frac{\partial^2 \xihat}{\partial \theta_1 \partial r}\pd{r}{x_j}+\hpd{\xihat}{\theta_1}{2}\pd{\theta_1}{x_j}+
\frac{\partial^2 \xihat}{\partial \theta_1 \partial \theta_2}\pd{\theta_2}{x_j}\bigg)\quad.\\
\pd{\big(\pd{\xihat}{\theta_2}\big)}{x_j}\pd{\theta_2}{x_i}=\pd{\theta_2}{x_i}\bigg(
\frac{\partial^2 \xihat}{\partial \theta_2 \partial r}\pd{r}{x_j}+
\frac{\partial^2 \xihat}{\partial \theta_2 \partial \theta_1}\pd{\theta_1}{x_j}+
\hpd{\xihat}{\theta_2}{2}\pd{\theta_2}{x_j}\bigg)\,.
\end{align}
For $i=j$, the above expressions can be combined to yield:
\begin{align}
\nonumber
\hpd{\xihat}{x_i}{2}&=
\hpd{\xihat}{r}{2}\bigg(\pd{r}{x_i}\bigg)^2+
\hpd{\xihat}{\theta_1}{2}\bigg(\pd{\theta_1}{x_i}\bigg)^2+
\hpd{\xihat}{\theta_2}{2}\bigg(\pd{\theta_2}{x_i}\bigg)^2\\
&+2\,\bigg(\frac{\partial^2 \xihat}{\partial r \partial\theta_1}\pd{r}{x_i}\pd{\theta_1}{x_i}+
\frac{\partial^2 \xihat}{\partial \theta_1 \partial\theta_2}\pd{\theta_1}{x_i}\pd{\theta_2}{x_i}+
\frac{\partial^2 \xihat}{\partial \theta_2 \partial r}\pd{\theta_2}{x_i}\pd{r}{x_i}\bigg)\quad.
\label{i_eq_j_laplace}
\end{align}
The helical coordinates and their first derivatives with respect to the Cartesian coordinates $x_1,x_2,x_3$ are as follows:
\begin{align}
\nonumber
r(x_1,x_2,x_3) &=\sqrt{x_1^2+x_2^2}\,,\,\theta_1(x_1,x_2,x_3) =\frac{x_3}{\tau}\,,\,\theta_2(x_1,x_2,x_3) =\frac{1}{2\pi}\arctan{(\frac{x_2}{x_1})}-\alpha \frac{x_3}{\tau}\,,\\\nonumber
\pd{r}{x_1} &=\frac{x_1}{\sqrt{x_1^2+x_2^2}}=
\frac{x_1}{r}\,,\,\pd{r}{x_2} =\frac{x_2}{\sqrt{x_1^2+x_2^2}}=\frac{x_2}{r}
\,,\,\pd{r}{x_3}=0\,,
\\\nonumber
\pd{\theta_1}{x_1} &= 0\,,\,\pd{\theta_1}{x_2}=0\;,\;\pd{\theta_1}{x_3}=
\frac{1}{\tau}\,,\\
\pd{\theta_2}{x_1} &= -\frac{1}{2\pi}\frac{x_2}{x_1^2+x_2^2}=
-\frac{1}{2\pi}\frac{x_2}{r^2}\,,\,
\pd{\theta_2}{x_2}=\frac{1}{2\pi}\frac{x_1}{x_1^2+x_2^2}
=\frac{1}{2\pi}\frac{x_1}{r^2}
\,,\,\pd{\theta_2}{x_3}=-\frac{\alpha}{\tau}\,.
\label{helical_1}
\end{align}
Using these, we get:
\begin{align}
\nonumber
\pd{\xi}{x_1} &= \pd{\xihat}{r}\pd{r}{x_1} + \pd{\xihat}{\theta_1}\pd{\theta_1}{x_1} + \pd{\xihat}{\theta_2}\pd{\theta_2}{x_1}\,\\
\label{grad_x1}
&=  \xihat_r \cos{\big(2\pi(\alpha\theta_1 + \theta_2)\big)} - \xihat_{\theta_2}\frac{\sin{\big(2\pi(\alpha\theta_1 + \theta_2)\big)}}{2\pi r}\,,\\\nonumber
\pd{\xi}{x_2} &= \pd{\xihat}{r}\pd{r}{x_2} + \pd{\xihat}{\theta_1}\pd{\theta_1}{x_2} + \pd{\xihat}{\theta_2}\pd{\theta_2}{x_2}\,\\
\label{grad_x2}
&= \xihat_r \sin{\big(2\pi(\alpha\theta_1 + \theta_2)\big)}+\xihat_{\theta_2}\frac{\cos{\big(2\pi(\alpha\theta_1 + \theta_2)\big)}}{2\pi r}\,,\\\nonumber
\pd{\xi}{x_3} &= \pd{\xihat}{r}\pd{r}{x_3} + \pd{\xihat}{\theta_1}\pd{\theta_1}{x_3} + \pd{\xihat}{\theta_2}\pd{\theta_2}{x_3}\,\\
\label{grad_x3}
&= \frac{1}{\tau}\big(\xihat_{\theta_1} - \alpha \xihat_{\theta_2}\big)\,,
\end{align}
which completes the calculation of the Cartesian gradient.

We now compute the second derivatives of the helical coordinates with respect to the Cartesian coordinates, but using eq.~\ref{i_eq_j_laplace}, we restrict ourselves only to the quantities that  would appear in the Laplacian:
\begin{align}
\label{helical_2}
\nonumber
\hpd{r}{x_1}{2}&=\frac{1}{r}-\frac{x_1^2}{r^3}\,,\,\hpd{\theta_1}{x_1}{2}=0\,,
\,\hpd{\theta_2}{x_1}{2}=\frac{1}{2\pi}\frac{2x_1x_2}{r^4}\,,\\\nonumber
\hpd{r}{x_2}{2}&=\frac{1}{r}-\frac{x_2^2}{r^3}\,,\,\hpd{\theta_1}{x_2}{2}=0\,,
\,\hpd{\theta_2}{x_2}{2}=-\frac{1}{2\pi}\frac{2x_1x_2}{r^4}\,,\\
\hpd{r}{x_3}{2} &= 0\,,\,\hpd{\theta_1}{x_3}{2}=0\,,\,\hpd{\theta_2}{x_3}{2}=0\,.
\end{align}
We are now ready to evaluate \eqref{chain_rule_again} through \eqref{i_eq_j_laplace}, \eqref{helical_1} and \eqref{helical_2}:
\begin{align}
\label{helical_3}
\nonumber
\hpd{\xi}{x_1}{2} &= \pd{\xihat}{r}\frac{\partial^2 r}{\partial x_1^2}+\pd{\xihat}{\theta_1}
\frac{\partial^2 \theta_1}{\partial x_1^2}+\pd{\xihat}{\theta_2}
\frac{\partial^2 \theta_2}{\partial x_1^2}+
\hpd{\xihat}{r}{2}\bigg(\pd{r}{x_1}\bigg)^2+
\hpd{\xihat}{\theta_1}{2}\bigg(\pd{\theta_1}{x_1}\bigg)^2+
\hpd{\xihat}{\theta_2}{2}\bigg(\pd{\theta_2}{x_1}\bigg)^2
\\\nonumber &+ 2\,\bigg(\frac{\partial^2 \xihat}{\partial r \partial\theta_1}\pd{r}{x_1}\pd{\theta_1}{x_1}+
\frac{\partial^2 \xihat}{\partial \theta_1 \partial\theta_2}\pd{\theta_1}{x_1}\pd{\theta_2}{x_1}+
\frac{\partial^2 \xihat}{\partial \theta_2 \partial r}\pd{\theta_2}{x_1}\pd{r}{x_1}\bigg)\\
&= \xihat_r\bigg(\frac{1}{r}-\frac{x_1^2}{r^3}\bigg)+\xihat_{\theta_2}\bigg(\frac{x_1x_2}{\pi r^4}\bigg)+\xihat_{rr}\frac{x_1^2}{r^2}+
\xihat_{\theta_1 \theta_2}\bigg(\frac{x_2^2}{4\pi^2r^4}\bigg)-\xihat_{\theta_2 r}\bigg(\frac{x_1x_2}{\pi r^3}\bigg)\,.
\end{align}
Similarly,
\begin{align}
\label{helical_4}
\nonumber
\hpd{\xi}{x_2}{2}&= \pd{\xihat}{r}\frac{\partial^2 r}{\partial x_2^2}+\pd{\xihat}{\theta_1}
\frac{\partial^2 \theta_1}{\partial x_2^2}+\pd{\xihat}{\theta_2}
\frac{\partial^2 \theta_2}{\partial x_2^2}+
\hpd{\xihat}{r}{2}\bigg(\pd{r}{x_2}\bigg)^2+
\hpd{\xihat}{\theta_1}{2}\bigg(\pd{\theta_1}{x_2}\bigg)^2+
\hpd{\xihat}{\theta_2}{2}\bigg(\pd{\theta_2}{x_2}\bigg)^2
\\\nonumber &+ 2\, \bigg(\frac{\partial^2 \xihat}{\partial r \partial\theta_1}\pd{r}{x_2}\pd{\theta_1}{x_2}+
\frac{\partial^2 \xihat}{\partial \theta_1 \partial\theta_2}\pd{\theta_1}{x_2}\pd{\theta_2}{x_2}+
\frac{\partial^2 \xihat}{\partial \theta_2 \partial r}\pd{\theta_2}{x_2}\pd{r}{x_2}\bigg)\\
&= \xihat_r\bigg(\frac{1}{r}-\frac{x_2^2}{r^3}\bigg)-\xihat_{\theta_2}\bigg(\frac{x_1x_2}{\pi r^4}\bigg)+\xihat_{rr}\frac{x_2^2}{r^2}
+\xihat_{\theta_2 \theta_2} \frac{x_1^2}{4 \pi^2r^4}+\xihat_{\theta_2 r}\frac{x_1x_2}{\pi r^3}\,.
\end{align}
and,
\begin{align}
\label{helical_5}
\nonumber
\hpd{\xi}{x_3}{2}&= \pd{\xihat}{r}\frac{\partial^2 r}{\partial x_3^2}+\pd{\xihat}{\theta_1}
\frac{\partial^2 \theta_1}{\partial x_3^2}+\pd{\xihat}{\theta_2}
\frac{\partial^2 \theta_2}{\partial x_3^2}+
\hpd{\xihat}{r}{2}(\pd{r}{x_3})^2+
\hpd{\xihat}{\theta_1}{2}\bigg(\pd{\theta_1}{x_3}\bigg)^2+
\hpd{\xihat}{\theta_2}{2}\bigg(\pd{\theta_2}{x_3}\bigg)^2
\\\nonumber &+ 2\,\bigg(\frac{\partial^2 \xihat}{\partial r \partial \theta_1}\pd{r}{x_3}\pd{\theta_1}{x_3}+
\frac{\partial^2 \xihat}{\partial \theta_1 \partial \theta_2}\pd{\theta_1}{x_3}\pd{\theta_2}{x_3}+
\frac{\partial^2 \xihat}{\partial \theta_2 \partial r}\pd{\theta_2}{x_3}\pd{r}{x_3}\bigg)\\
&= \xihat_{\theta_1 \theta_1}\frac{1}{\tau^2}+\xihat_{\theta_2 \theta_2}\frac{\alpha^2}{\tau^2}
-2\xihat_{\theta_1\theta_2}\frac{\alpha}{\tau^2}\,.
\end{align}
So, we have: 
\begin{align}
\nonumber
\Delta \xi &= \hpd{\xi}{x_1}{2}+\hpd{\xi}{x_2}{2}+\hpd{\xi}{x_3}{2}\\
\label{laplace_in_helical}
&= \xihat_{rr}+\frac{1}{r}\xihat_{r}+\frac{1}{\tau^2}\xihat_{\theta_1\theta_1}-
\frac{2\alpha}{\tau^2}\xihat_{\theta_1 \theta_2}+
\frac{1}{4\pi^2}\bigg(\frac{1}{r^2}+\frac{4\pi^2\alpha^2}{\tau^2}\bigg)\xihat_{\theta_2 \theta_2}\,.
\end{align}

Finally, we compute the Jacobian determinant of the transformation to helical coordinates as:
\begin{align}
\nonumber
&= \text{Det.}
\begin{bmatrix}
\pd{x_1}{r} & \pd{x_1}{\theta_1} & \pd{x_1}{\theta_2}\\
\pd{x_2}{r} & \pd{x_2}{\theta_1} & \pd{x_2}{\theta_2}\\
\pd{x_3}{r} & \pd{x_3}{\theta_1} & \pd{x_3}{\theta_2}
\end{bmatrix}\\\nonumber
&= \text{Det.}
\begin{bmatrix}
\cos(2\pi(\alpha \theta_1 + \theta_2)) & -2\pi\alpha r\sin(2\pi(\alpha\theta_1 + \theta_2)) & -2\pi r\sin(2\pi(\alpha\theta_1 + \theta_2))\\
\sin(2\pi(\alpha \theta_1 + \theta_2)) & 2\pi\alpha r\cos(2\pi(\alpha \theta_1 + \theta_2))  & 2\pi r\cos(2\pi(\alpha \theta_1 + \theta_2))\\
0 & \tau & 0
\end{bmatrix}
\\&= 2\pi\tau r
\end{align}
Thus, the integral of a scalar function $\xi(\bfx)$ over the simulation cell $\Omega$ (expressed in cylindrical coordinates as $\Omega = \big\{(r,\vartheta,z) \in \rz^3:R_{\text{in}} \leq r\leq R_{\text{out}}, 0 \leq \vartheta \leq \Theta, 0 \leq z \leq \tau\big\}$) can be written as:
\begin{align}
\nonumber
\int_{\Omega}\xi(\bfx)\,d\bfx &= \iiint_{(x_1,x_2,x_3) \in \Omega}\xi(x_1,x_2,x_3)\,dx_1\,dx_2\,dx_3\\
&=\int_{r=R_{\text{in}}}^{r=R_{\text{out}}}\int_{\theta_1 = 0}^{\theta_1 = 1}\int_{\theta_2 = 0}^{\theta_2 = \frac{1}{\mathfrak{N}}}\widetilde{\xi}(r,\theta_1,\theta_2)\,2\pi\tau r\,dr\,d{\theta_1}\,d{\theta_2}\,,
\end{align}
with $\widetilde{\xi}(r,\theta_1,\theta_2) = \xi(x_1(r,\theta_1,\theta_2) , x_2(r,\theta_1,\theta_2), x_3(r,\theta_1,\theta_2))$.
\begin{center}
---
\end{center}
\section*{Acknowledgement}
Some of the theoretical and computational foundations of this work were laid out while ASB was a graduate student at the University of Minnesota, and later a postdoctoral fellow at the Lawrence Berkeley National Laboratory. ASB would like to acknowledge support of the Scientific Discovery through Advanced Computing (SciDAC) program funded by U.S. Department of Energy, Office of Science, Advanced Scientific Computing Research and Basic Energy Sciences during his time at the Berkeley Lab, as well as the support of the following grants while at the University of Minnesota: AFOSR FA9550-15-1-0207, NSF-PIRE OISE-0967140, ONR N00014-14-1-0714 and the MURI project FA9550-12-1-0458 (administered by AFOSR).

ASB would like to acknowledge informative discussions and email communications with Carlos Garcia Cervera (Univ.~of California, Santa Barbara), Eric Cances (Ecole des Ponts ParisTech), Richard James (Univ.~of Minnesota), Ryan Elliott (Univ. of Minnesota), Phanish Suryanarayana (Georgia Institute of Technology), Paul Garrett (Univ.~of Minnesota), Kaushik Bhattacharya (Caltech), Lin Lin (Univ.~of California, Berkeley) and Chao Yang (Lawrence Berkeley National Lab). ASB would also like to thank Neha Bairoliya (Univ.~of Southern California) for her help in preparing some of the figures,  and also for providing encouragement and support during the preparation of this manuscript. Help from  Hsuan Ming Yu (UCLA)  in generating some of the simulation data in the paper is also gratefully acknowledged. 

Finally, ASB would like to acknowledge the anonymous reviewers for suggestions which helped in   improving the manuscript, as well as the Minnesota Supercomputing Institute (MSI) and UCLA's Institute for Digital Research and Education (IDRE) for making available the computing resources used in this work. %%%% COR-FRG acknowledge
%\section*{\refname}
\begin{center}
---
\end{center}
\bibliographystyle{elsarticle-num}
\bibliography{main}
\end{document}

%% file: preamble.tex
\newcommand{\rz}{\mathbb{R}}
\newcommand{\gz}{\mathbb{Z}}
\newcommand{\cz}{\mathbb{C}}

\newcommand{\nz}{\mathbb{N}}

\newcommand{\bfc}{{\bf c}}

\newcommand{\bfe}{{\bf e}}

\newcommand{\bfp}{{\bf p}}

\newcommand{\bfs}{{\bf s}}

\newcommand{\bfx}{{\bf x}}
\newcommand{\bfy}{{\bf y}}

\newcommand{\bfI}{{\bf I}}

\newcommand{\bfR}{{\bf R}}

\newcommand{\xihat}{\hat{\xi}}

\newcommand{\beq}{\begin{equation}}
\newcommand{\eeq}{\end{equation}}
\newcommand{\beqs}{\begin{eqnarray}}
\newcommand{\eeqs}{\end{eqnarray}}
\newcommand{\beql}{\begin{equation} \label}

\newcommand{\half}{\frac{1}{2}}
\newcommand{\calA}{{\cal A}}
\newcommand{\calB}{{\cal B}}
\newcommand{\calC}{{\cal C}}
\newcommand{\calD}{{\cal D}}

\newcommand{\calF}{{\cal F}}
\newcommand{\calG}{{\cal G}}
\newcommand{\calH}{{\cal H}}
\newcommand{\calI}{{\cal I}}

\newcommand{\calK}{{\cal K}}
\newcommand{\calL}{{\cal L}}

\newcommand{\calN}{{\cal N}}
\newcommand{\calO}{{\cal O}}
\newcommand{\calP}{{\cal P}}

\newcommand{\calS}{{\cal S}}

\newcommand{\calU}{{\cal U}}
\newcommand{\calV}{{\cal V}}

\newcommand{\calZ}{{\cal Z}}

\newcommand{\hamil}{\mathfrak{H}}

\newtheorem{theorem}{Theorem}[section]
\newtheorem{lemma}[theorem]{Lemma}
\newtheorem{proposition}[theorem]{Proposition}

\theoremstyle{definition}

\theoremstyle{definition}

\theoremstyle{remark}

\DeclareMathAlphabet{\mathcalligra}{T1}{calligra}{m}{n}
\DeclareFontShape{T1}{calligra}{m}{n}{<->s*[2.2]callig15}{}
\newcommand{\scriptr}{\mathcalligra{r}\,}

\newcommand{\abs}[1]{\lvert#1\rvert}
\newcommand{\norm}[2]{\lVert#1\rVert_{#2}}
\newcommand{\innprod}[3]{\langle#1,#2\rangle_{#3}}

\newcommand{\Lpspc}[3]{\textsf{L}^{#1}_{#2}(#3)}

\newcommand{\pd}[2]{\frac{\partial#1}{\partial#2}}
\newcommand{\hpd}[3]{\frac{\partial^{#3}#1}{\partial#2^{#3}}}

\let\oldFootnote\footnote
\newcommand\nextToken\relax

\renewcommand\footnote[1]{%
    \oldFootnote{#1}\futurelet\nextToken\isFootnote}

\newcommand\isFootnote{%
    \ifx\footnote\nextToken\textsuperscript{,}\fi}
    
\DeclareMathOperator{\arctantwo}{arctan2}